%% file: paper.tex
\documentclass[acmsmall]{acmart}

\usepackage{hyperref}
\usepackage{xspace}
\usepackage{graphicx}

\usepackage{amssymb, amsthm, amsfonts, latexsym, amsmath}
\usepackage{mdwlist}
\usepackage{algorithm}
\usepackage{booktabs}
\usepackage{framed}
\usepackage{url}
\usepackage[utf8]{inputenc} 

\newif\ifappendix
\appendixtrue

\newcommand{\eat}[1]{}
\newcommand{\pages}[1]{} 
\newcommand{\histore}{Histor$\epsilon$\xspace}
\newcommand{\Sys}{\mbox{PSC}\xspace}
\newcommand{\sys}{PSC\xspace}
\newcommand{\Paragraph}[1]{\vspace{.1in}\noindent{\bf #1.}\quad}
\newcommand{\rgets}{ \overset{\lower.5em\hbox{$\text{\tiny R}$}}{\leftarrow}\xspace }
\newcommand{\f}[2][]{$\mathcal{F}_{\textsc{#2}}^{\textsc{#1}}$}

\newcommand{\hybrid}{(\f{BC}, \f{SKGD}, \f{ZKP-S}, \f{ZKP-RRD}, \f{ZKP-DL})-hybrid }
\newcommand{\real}{$\textsc{REAL}_{\pi_{\text{\sys}},A,Z}(k)$\xspace}
\newcommand{\ideal}{$\textsc{IDEAL}_{\mathcal{F}_{\text{\sys},S,Z}}(k)$\xspace}

\newcommand{\eg}{\emph{e.g.}\xspace}

\newcommand{\ie}{\emph{i.e.}\xspace}
\newcommand{\etc}{\emph{etc.}\xspace}

\newcommand{\numinput}{d\xspace}
\newcommand{\numcompute}{m\xspace}
\newcommand{\numcounters}{b\xspace}
\newcommand{\numnoise}{n\xspace}

\newcommand{\dataparties}{Data Parties\xspace}
\newcommand{\dataparty}{Data Party\xspace}

\newcommand{\dparty}{DP\xspace}
\newcommand{\dps}{DPs\xspace}
\newcommand{\dpmath}[1]{\textrm{DP}_{#1}}
\newcommand{\computationparties}{Computation Parties\xspace}
\newcommand{\computationparty}{Computation Party\xspace}
\newcommand{\cp}{CP\xspace}
\newcommand{\cps}{CPs\xspace}
\newcommand{\cpmath}[1]{\textrm{CP}_{#1}}

\hypersetup{
     bookmarks=true,
     colorlinks=true,
     linkcolor=black,
     citecolor=black,
     filecolor=black,
     urlcolor=black,
     pdfauthor = {Ellis Fenske, Akshaya Mani, Aaron Johnson, and Micah Sherr},
     pdftitle = {Accountable Private Set Cardinality for Distributed Measurement},
     pdfkeywords = {secure computation, privacy-preserving measurement}
}

\begin{document}
\title{Accountable Private Set Cardinality for Distributed Measurement}
\author{Ellis Fenske}
\affiliation{%
    \institution{U.S.~Naval Academy}
    \city{Annapolis}
    \country{USA}}
\email{fenske@usna.edu}
\author{Akshaya Mani}
\affiliation{%
    \institution{Georgetown University}
    \city{Washington}
    \country{USA}}
\affiliation{%
	\institution{University of Waterloo}
	\city{Waterloo}
	\country{Canada}}
\email{akshaya.mani@uwaterloo.ca}
\author{Aaron Johnson}
\affiliation{%
    \institution{U.S.~Naval Research Laboratory}
    \city{Washington}
    \country{USA}}
\email{aaron.m.johnson@nrl.navy.mil}
\author{Micah Sherr}
\affiliation{%
    \institution{Georgetown University}
    \city{Washington}
    \country{USA}}
\email{msherr@cs.georgetown.edu}


\acmJournal{TOPS}
\setcopyright{licensedusgovmixed}
\copyrightyear{2022}
\acmYear{2022} \acmVolume{25} \acmNumber{4} \acmArticle{25} \acmMonth{5} \acmPrice{15.00}\acmDOI{10.1145/3477531}

\input{sections/abstract}

\maketitle

\input{sections/intro}
\input{sections/related}
\input{sections/background}
\input{sections/problem}
\input{sections/protocol}
\input{sections/security}
\input{sections/implementation}
\input{sections/evaluation}
\input{sections/intersection}
\input{sections/conclusion}
\input{sections/acks}

\bibliographystyle{ACM-Reference-Format}
\bibliography{bibs/micah-long,bibs/others}

\ifappendix
\appendix
\input{sections/appendix.tex}
\else
\fi

\end{document}

%% file: sections/abstract.tex
\begin{abstract}
We introduce cryptographic protocols for securely and efficiently computing the cardinality of set
union and set intersection. Our private set-cardinality protocols (\Sys) are designed for the
setting in which a large set of parties in a distributed system makes observations, and a small
set of parties with more resources and higher reliability aggregates the observations.
\Sys allows for secure and useful statistics gathering in privacy-preserving distributed systems.
For example, it allows operators of anonymity networks such as Tor to securely answer the
questions: {\em How many unique users are using the network?} and {\em How many hidden services are
being accessed?}

We prove the correctness and security of \Sys in the Universal Composability framework against an
active adversary that compromises all but one of the aggregating parties. Although successful output
cannot be guaranteed in this setting, \Sys either succeeds or terminates with an abort, and we
furthermore make the adversary \emph{accountable} for causing an abort by blaming at least one
malicious party. We also show that \Sys prevents adaptive corruption of the data parties from
revealing past observations, which prevents them from being victims of targeted compromise, and we
ensure safe measurements by making outputs differentially private.

We present a proof-of-concept implementation of \Sys and use it to demonstrate that \Sys operates
with low computational overhead and reasonable bandwidth. It can count tens of thousands of
unique observations from tens to hundreds of data-collecting parties while completing within hours.
PSC is thus suitable for daily measurements in a distributed system.
\end{abstract}

%% file: sections/intro.tex
\section{Introduction}

Measurements are essential to understanding the use of distributed systems and monitoring them for abuse.  In a privacy-preserving distributed system, such as an anonymity network, measurement is complicated by the system's privacy requirements.  Storing records of system activity can pose significant risks to the system's users, and consequently the ethics of such techniques~\cite{tor-usage} have been widely debated~\cite{tor-research,partridge2016ethical}. Ideally, the measurements produced should satisfy strong privacy definitions, and during the measurement process the system should protect sensitive intermediate data.

For example, in an anonymous-communication system such as Tor~\cite{tor}, it
is very helpful to measure statistics such as the number of users, the popular
applications used over Tor, and the number of errors that occur. With such
information, users can better understand their anonymity, policymakers can
consider the value of anonymity, and network designers can identify its
problems. However, the privacy goal of the network would be undermined if
detailed records were collected of who uses the system, when it is used,
and for what purposes. As other example applications, structured overlay networks such as Chord~\cite{chord} may wish to measure network activity without
exposing individual actions, and online services may want to know how many users
they share without revealing their identities.

Two general privacy technologies can help solve this problem: secure multiparty
computation protocols~\cite{lindell2005secure} (MPC) and differentially private
mechanisms~\cite{dwork2014algorithmic}. Generic MPC protocols can be used by a
set of parties to compute any function of their private inputs without revealing
those inputs. Differentially private mechanisms guarantee that the value they
compute depends little on any specific input, thereby only revealing information
implied by the collection of inputs. Thus to safely measure a distributed
system, the members could execute a generic MPC protocol, using their
observations as input, in order to compute a differentially private measurement
function~\cite{pettai2015combining}. However, generic protocols have relatively
high computational and communication costs. Moreover, this approach does not
protect the observations being recorded by the parties to serve as inputs to the
secure computation, which could be revealed if such parties can be corrupted or
otherwise compelled to reveal their internal state during the measurement
period.

Some specific protocols have been developed to more efficiently and securely
compute aggregate measurements in anonymity networks, discussed in detail in the following section. We extend this line of work by presenting a protocol which is capable of privately measuring unique counts, calculating private set cardinality operations for the set union and
set intersection operations. In the scenario we consider, a set of \dataparties
(\dps) each observes some set, and we wish to compute the size of the union or
intersection of those sets. More formally, if there are $\numinput$ \dps with
$\dpmath{k}$ observing set $\mathcal{I}_k$, we want to be able to privately
compute $\lvert \bigcup_{k=1}^\numinput \mathcal{I}_k \rvert$ for set-union
cardinality and $\lvert \bigcap_{k=1}^\numinput \mathcal{I}_k \rvert$ for
set-intersection cardinality. We assume that there exists a smaller set of
\computationparties (\cps) that can be used to aggregate the observations of
the $\dps$, which will improve efficiency and allow us to tolerate failures of
the $\dps$.

We seek to satisfy stringent privacy and security goals to make the protocol
suitable even in very adversarial settings, appropriate for a distributed system
in which many of the parties may be malicious. Anonymity networks are a key
desired application of our protocols, and several active attacks have been
observed on the Tor network. To this end, we require that the protocol is secure
against an \emph{active} adversary that can deviate from the prescribed
protocol. Moreover, we desire security against a \emph{dishonest majority} of
the \computationparties, to ensure the security of user data as long as at least
one \cp is honest. We seek \emph{composable} security, so that measurement can
be repeatedly and concurrently run while maintaining security. We further target
\emph{adaptive security} among the \dataparties such that corruptions during the
lengthy measurement periods do not reveal past observations. Adaptive security is
particularly important for \dps as they may otherwise be susceptible to a legal
compulsion attack targeted at data stored about the observations they
make~\cite{privex}. We also require that the protocol output satisfies
\emph{differential privacy} to ensure that a securely computed function poses no
privacy risk. Finally, we require that no \dataparty can prevent the protocol
from completing successfully, and, while for efficiency we allow a
\computationparty to cause a failure, we require \emph{accountability} of a
responsible party so that the party can be excluded and the protocol
restarted.

We present the \Sys protocol for private set-cardinality operations. \Sys
satisfies all our functionality and security goals, and it does so in a way
that is efficient and practical to be used for distributed-system measurements,
including especially measurements of anonymity systems such as Tor. \Sys has
two set-cardinality variants: one for set union and the other for set
intersection. We present the set-union variant and describe the changes needed
to it for set intersection in Section ~\ref{sec:intersection}.

We prove the correctness and security of \Sys in the Universally Composable (UC) framework~\cite{uc-focs2001}. This proof methodology allows us to easily
express the main security and privacy goals via a single ideal functionality.
It also guarantees that security is maintained even if the protocol is run
concurrently with itself or as part of a larger cryptographic system.

We additionally introduce an implementation of \Sys, released as open-source
software written in memory-safe Go.  To achieve greater efficiency, our
implementation uses some subprotocols that are not proven UC-secure, such as a
verifiable shuffle~\cite{neff-shuffle} that is more practical than verifiable
shuffles proven secure in the UC model.  With these subprotocols, our
implementation can still be proven secure in the classical (\ie standalone)
model and, as we demonstrate via at-scale evaluation experiments, incurs only
moderate bandwidth and computation costs and can be practically deployed.

This paper expands upon an earlier version of this work~\cite{psc-ccs2017} in
many ways, including in particular modifying the protocol to provide accountability, adding measures to prevent input modification, and detailing
how to accomplish set-intersection cardinality.

%% file: sections/related.tex
\section{Related Work}
\label{sec:related}
We consider three main parallel lines of work related to our central problem: protocols designed to measure anonymity networks, protocols designed to compute secure distributed set operations, and fully generic multiparty computation protocols. Finally, we discuss related work surrounding accountability properties in cryptographic protocols.

\paragraph{Privacy-Preserving Measurements of Distributed Systems}
The PrivEx system of
\mbox{\citet{privex}} uses partially-homomorphic aggregation and differential
privacy~\cite{dwork2014algorithmic} to privately collect statistics about the
Tor network~\cite{tor}. PrivCount~\cite{privcount}, which we extend in this
work, improves upon PrivEx by offering multi-phase measurements and also an
optimal allocation of the $\epsilon$ privacy budget. \histore~\cite{histore} is
also inspired by PrivEx and uses histogram-based queries to provide integrity
guarantees by bounding the influence of malicious data contributors.
\citet{succinct-sketches} present a system to perform distributed item-frequency
counting, with one application being computing median statistics from Tor
relays. It uses homomorphic encryption to securely aggregate count-min sketches.

While these systems provide significant efficiency improvements over generic
protocols, they lack the ability to perform set-cardinality operations. For
example, while PrivEx, PrivCount, and \histore can answer the question {\em ``How many
clients were observed entering the Tor network?''}, they cannot determine the
number of unique clients (i.e. set-union cardinality). Similarly, while
count-min sketches can help determine the median number of clients across all
entry relays, they cannot determine the how many clients are seen at all relays
(i.e. set-intersection cardinality).

\paragraph{Private Set Operations Protocols}
Brandt suggests a protocol~\cite{brandt2005efficient} very similar to the
aggregate-shuffle-rerandomize-decrypt scheme our \cps execute. However, our
construction differs in a few crucial parts: we need to include separate \dparty
parties to provide the input that must be adaptively secure and limit as much as
possible their computational work, and in our protocol the parties jointly
generate noise to satisfy a differential privacy guarantee.  Beyond these
modifications, the bulk of our contribution is a thorough proof of security of
the protocol in the UC model (to make the proof go through, we needed to add
re-encryption during the re-randomization phase), a specific application for the
general theoretical protocol (making measurements in privacy-preserving systems
like Tor), and a functional implementation of the protocol alongside empirical
data measuring computation and communication costs gathered through experiments.

Several protocols to securely compute set-union cardinality (or the related problem of
set-intersection cardinality) have been
proposed~\cite{setint-jcs2005,setops-crypto2005,setunion-cans2012,setunion-acisp2015,distributed-cardinality-with-dp}
using a variety of techniques including bloom filters, homomorphic encryption, sketches, and
polynomial evaluation. However, none of these provides all of the security properties that we desire
for distributed measurement in a highly-adversarial setting: malicious security against a dishonest
majority, adaptive security for Data Parties, and differentially-private output. Similar protocols
have been designed to securely compute set operations~\cite{setops-joc2012,
private-set-intersection,privaterecordInan}, but these protocols output a set rather than its
cardinality.

\paragraph{Secure Multiparty Computation}
General secure multiparty computation (MPC) protocols can realize any
functionality including set cardinality, even in the UC
model~\cite{canetti2002universally,uc-bgw-asharov}. Moreover, advances in
efficient MPC have been made in the multi-party, dishonest-majority, active
security setting that we
require~\cite{spdz-crypto2012,mpc-bin-crypto2014,lindell2015efficient,wang2017global}.
However, these works assume that the same parties that have the inputs perform
the computation and thus do not describe how to securely transfer inputs if
these parties are different. Moreover, such protocols make use of a relatively
expensive ``offline'' phase, while we intend to allow for measurements that are
run on a continuous basis.

Wails et al.~\cite{stormy-ccs2019} have explored how to apply these generic MPC protocols in the
context of measuring the Tor network. They present new protocols to transfer inputs from a large
number of observing parties to a smaller number of computation parties. They show how to securely
compute approximate set-union cardinality using a sketch (viz. LogLog~\cite{durand2003loglog}).
While their construction does not provide differential privacy, Choi et
al.~\cite{diffPrivateSketches} show that sketch-based counts can be made differentially private with high
accuracy. Compared to \Sys, this approach of using sketches and generic MPC protocols has the
advantage of scaling well to large counts (e.g. billions), but it has the disadvantages of poor
accuracy when the count is small and error from noise that is proportional to the count rather than
constant. These disadvantages arise from the use of sketch-based counting. Furthermore, the generic
MPC protocols allow an adversary to abort the protocol without accountability, which
enables a continuous denial of service.

\paragraph{Accountability}
To handle denial of service attacks on secure multiparty computation protocols, there has been recent work developing protocols that can identify to all honest parties a malicious party who has prevented the protocol from terminating successfully. This type of termination is called \textit{identifiable abort}, and there are generic multiparty computation protocols with this property~\cite{gmw-mpc,mpc-identifiable-abort,ishai2014secure,constant-ia-mpc-public}. Taking another approach to accountability, some protocols provide a stronger public-accountability property by producing a proof of malicious behavior on some trusted publicly readable trusted ledger or bulletin board~\cite{fair-mpc-global-ledger,constant-ia-mpc-public}. While all of these contributions using either approach give accountable methods for general private measurements of any kind, including those we wish to collect, these contributions do not provide implementations or experiments, and require that the parties providing inputs are those that perform the computations, meaning using them for the applications we intend would incur heavy computation and communication costs to our Data Parties. Additionally, these approaches do not consider the problem of adaptive corruptions intended to extract measurement data during a long collection period, a property we consider to be crucial for many typical use cases such as measuring anonymous communication systems like Tor.

%% file: sections/background.tex
\section{Background}
Before describing the problem and our solution, we briefly review some concepts and background that are necessary for understanding the protocol, Private Set-Union Cardinality (\sys).

\paragraph{Differential Privacy}
Differential privacy~\cite{dwork2014algorithmic} is a privacy definition that
offers provable and quantifiable privacy of database queries.  Differential privacy
guarantees that the query responses look nearly the same regardless of whether or not the input of
any one user is included in the database. Thus anything that is learned from the queries is learned
independently of the inclusion of any single user's data~\cite{dp-semantics-jp08}. Several
mechanisms have been designed to provide differential privacy while maximizing
accuracy~\cite{calibrating-noise,exponential-mech-focs07,smooth-sensitivity-stoc07,blum2013learning}.

More formally, an \mbox{($\epsilon$, $\delta$)-differentially-private} mechanism is an algorithm
$\mathcal{M}$ such that, for all datasets $D_1$ and $D_2$ that differ only on the input of one
user, and all $S \subseteq \textrm{Range}(\mathcal{M})$, the following holds:
\begin{equation} \label{eq:diffpriv}
Pr[\mathcal{M}(D_1) \in S] \leq e^{\epsilon} \times Pr[\mathcal{M}(D_2) \in S] + \delta.
\end{equation}
$\epsilon$ and $\delta$ quantify the amount of privacy provided by the mechanism, where smaller
values of each indicate more privacy.
\citet{ourdata} proves that binomial noise, that is, the sum of $n$ uniformly random
binary values, provides $(\epsilon,\delta)$-differential privacy for queries that each user can
affect by at most one when
 \begin{equation}
n \geq \bigg(\dfrac{64\ln(2/\delta)}{\epsilon^2}\bigg) \label{eq:noise}
\end{equation}
Eq.~\ref{eq:noise} presents a trade-off between privacy and utility, an issue inherent to
differential privacy.  Put alternatively, for the privacy level given by $\epsilon$ and
$\delta$, Eq.~\ref{eq:noise} yields the amount of binomial noise that is required to be added to the
output of a query that each user can change by at most one.
In this paper, we use this binomial noise technique to achieve differential privacy.

\paragraph{Termination and Accountability}
In any multiparty computation protocol, weak trust assumptions like a dishonest majority preclude properties like guaranteed output delivery~\cite{ishai2014secure}. This means even if a protocol guarantees its output is accurate, malicious parties may prevent the output from being produced at all by causing the protocol to hang indefinitely or to abort. Exacerbating this problem, dishonest majority protocols cannot guarantee \textit{fairness}, meaning malicious parties could learn the result before choosing whether to allow a protocol to successfully conclude.

As a solution to this, many protocols introduce an accountability guarantee called \textit{identifiable abort}~\cite{mpc-identifiable-abort}, which guarantees that either the correct output or the identity of some number of malicious parties is given to the honest parties at protocol termination. This allows honest parties to identify a source of failure and take appropriate steps to remedy the problem: identify network failures, rerun the protocol excluding certain parties, or publicly levy accusations of malicious behavior.

This guarantee only mitigates denial of service attacks in practice if a protocol comes with a guarantee of termination within finite time, or malicious parties could indefinitely delay the accountability messages as well as the intended protocol output. This is handled often in literature by defining protocols in the \textit{synchronous communication model}, where messages are sent in a sequence of explicit synchronized communication rounds. In every round, each party may send messages to any set of other parties, and these messages are guaranteed to be delivered to their recipient at the end of that round. Synchronous, round-based communications can be constructed explicitly using bounded delay and synchronized clocks~\cite{uc-sync}.

\paragraph{Zero-Knowledge Proofs}
\Sys uses a few types of zero-knowledge proofs demonstrating knowledge and relationships between elements in some group $G$ of order $q$ with respect to some generator $g$. For example, we require a proof of knowledge of the discrete log of $y = g^x$, $x \in \mathbb{Z}_q$. In general, a zero-knowledge proof
system~\cite{goldreich-foc1-2004} is a
protocol between a prover $\mathcal{P}$ and verifier $\mathcal{V}$ in which the prover demonstrates
the truth of some statement without revealing more than that truth, where the statement may be, for
example, the existence or knowledge of a witness to membership in a language.

\Sys uses many $\Sigma$-protocols~\cite{damgaard2002sigma}, which are proofs of knowledge that are three-round interactive protocols starting with a commitment by the prover, followed by a random challenge from the verifier, and ended by a response from the prover. $\Sigma$-protocols are honest-verifier zero knowledge (HVZK), that is, their transcript is simulatable assuming the verifier behaves honestly.

\sys also makes use of {\em verifiable shuffles}~\cite{nguyen2004verifiable}.  Informally, a
verifiable re-encryption shuffle takes as input ciphertexts, outputs a permutation of a
re-encryption of those ciphertexts, and proves that the output is a re-encryption and permutation of
the input. There are two security requirements for verifiable shuffles: privacy and verifiability. Privacy requires an honest shuffle to protect its secret permutation. Verifiability requires that any attempt by a malicious shuffle to produce an incorrect output must be detectable.
Several protocols for verifiable shuffling have been
proposed~\cite{neff-shuffle,shuffle-furukawa-fc2003,shuffle-groth-pkc2003,shuffle-bayer-eurocrypt2012}.

The $\Sigma$-protocols we use in \sys and the verifiable shuffles with public-coin interactive proofs can be made non-interactive using the Fiat-Shamir heuristic~\cite{fiat-shamir-crypto1986}, in which the
random challenges are generated by the prover by applying a cryptographic hash function. Non-interactive proofs (both shuffles and $\Sigma$-protocols) can be sent to many verifiers through a broadcast, but \Sys additionally uses interactive multi-verifier proofs, where a single proof is sent to multiple verifiers such that all honest verifiers agree on the statement to be proved and the correctness of the proof to provide accountability.

\paragraph{Broadcasts}
\Sys uses an accountable broadcast communication functionality. The security property that we require for this broadcast is that honest parties agree on the message from the broadcaster, whether it is an \textit{equivocation}, where a broadcaster sends distinct messages to different parties, a single consistent message to all parties, or no message at all. Dolev and Strong~\cite{dolev-strong} give a protocol that satisfies this property assuming synchronous point-to-point communication channels in $t+1$ rounds, where $t$ is the number of malicious parties. For protocols with many parties the round- and communication-complexity of this primitive make it too inefficient to be practical if deployed na\"ively.

\paragraph{Universal Composability and Setup Assumptions}
\Sys, like many complex cryptographic protocols is constructed from smaller cryptographic primitives and subprotocols. The sequential composition theorem~\cite{goldreich-foc1-2004} justifies constructing protocols modularly in this way, but for protocols to remain secure under general concurrent composition (i.e. simultaneously with arbitrary other protocols, even those designed specifically to interact with them) we require a stronger property. In 2002, Canetti gave a composition theorem for this case under a framework called Universal Composable (UC) security~\cite{canetti2002universally}. This composition theorem guarantees that compositions of UC-secure subprotocols are UC-secure.

Given impossibility results for UC secure protocols without setup assumptions, UC protocols are often defined in the Common Reference String (CRS) model, under which any function of a given collection of parties' private inputs can be securely computed~\cite{canetti2002universally}. In particular, there are natural constructions of UC-secure zero-knowledge proofs from $\Sigma$-protocols~\cite{setops-crypto2005} and UC-secure verifiable shuffles in this model~\cite{wikstrom-shuffle}. Protocols that meet this standard of security can be inefficient, so in practice many protocols rely on stronger setup assumptions and provide weaker composability guarantees. We highlight the Random Oracle Model, where all parties have access to a random oracle that returns random strings to all queries, but must return the same result given two identical queries. In the Random Oracle Model, the Fiat-Shamir heuristic constructs auxiliary-input zero-knowledge proofs~\cite{rom}, which are secure under sequential composition~\cite{seqcom}. Random Oracles are usually instantiated with hash functions, a leap which cannot be rigorously justified for any real hash function~\cite{rom-imposs} but which is widely deployed in practice.

%% file: sections/problem.tex
\section{Problem and Setting}


We consider the problem of distributed private measurement. In particular, we define a finite set of \textit{observations}, which could be explicit identifiable events or abstract counters intended to estimate a count of these events (\eg, in a hash table), to be collected across a large distributed set of $\numinput$ Data Parties (\dps). Each \dparty collects its own set of observations, but we wish to learn how many distinct events has occured over a given time period across all \dps.  We wish to calculate the cardinality of the set-union and set-intersection of these observations.

We collect data on these observations during a \textit{collection period}, during which \dps record observations. After the collection period, the \dps participate in a distributed {\em protocol phase} that produces the cardinality of the set-union or set-intersection of the \dps' observations.  The protocol is privacy-preserving, which informally means that (i)~no information collected by an honest DP is ever exposed and (ii) an adversary (defined below) cannot identify any individual data from the aggregated result (i.e. the cardinality).

Additionally, we desire {\em accountability}, meaning that the honest parties can identify some dishonest party if the execution fails.  (We provide a more formal definition of accountability later in the paper.)  Importantly, our accountability property does not cover the veracity of \dps' reported observations---a malicious \dparty can fabricate or omit an observation.  Rather, our aim is to detect parties that interfere with the correct computation of the aggregated result.





\subsection{System Model}

To solve this problem, we introduce a set of $\numcompute$ Computation Parties (\cps) whose purpose is to offload a majority of the computational- and bandwidth-work from the \dps, and to perform a multi-party computation to aggregate, process, and clean the data. We assume that $\numinput \gg \numcompute$ so that the number of \dps is much larger than the number of \cps.

We further assume that \dps are limited in terms of computational resources and bandwidth, so we measure efficiency in this setting by aggregate computational and communication resources consumed but also by the extent to which resources consumed by the \dps for measurement purposes is minimized.

We consider the problem in the synchronous communication model and do not specifically prove results about or implement precise synchronized rounds, leaving them abstract for the purpose of exposition, reasoning about security properties, and prototype implementation.

\subsection{Threat Model}

We consider in our threat model arbitrary adversaries who corrupt any number of \dps and all but one of the \cps, and consider primarily adversaries who wish to learn private data: which observations have been recorded, and by which \dps. Therefore, it is crucial that an adversary who adaptively attacks \dps during the collection period not does learn more private data by attacking \dps who execute our measurement protocol than that same adversary could learn by attacking \dps who do not, so any satisfactory measurement protocol must resist adaptive attacks against \dps during the collection period.

Further, we expect the protocol to be run repeatedly (\eg, to be run continuously with collection periods beginning and ending simultaneously). This means even simple count data can reveal a large amount of private information over time, especially against active adversaries who can insert inputs through malicious \dps and observe differences in the protocol output. This means the measurements released by the protocol must be differentially private.

Finally, \dps do not always exclusively take measurements, and often these measurements are associated with a large, complex cryptographic protocol whose execution is the primary purpose of the \dps. Further, the protocol is intended to be run concurrently with other copies of itself (measuring different events, or over different time periods). This means the protocol should be privacy-preserving even when run concurrently with other arbitrary protocols, meaning a measurement protocol should be secure in the UC model.

\cps should be taken from a distributed set of parties not likely to collude in order to recover the sensitive private data they process, and we formalize this intuition by assuming at least one of the \cps is honest. We do not consider the problem of adaptively corrupted \cps.

While we accept that malicious \dps may report false data, we may perform repeated sampling of \dps among those possible in order to provide robust measurements over time, or exclude obviously malicious \dps (those who report having observed every event in order to produce a trivial result from a Set-Union Cardinality measurement, for example). However, since \cps participate in every round, it is important that \cps may not adjust the results or prevent the protocol from completing without consequences. This takes the form of a formal \textit{accountability property} in the dishonest-majority setting, which guarantees honest parties agree on a set of malicious parties to blame in the case the protocol fails to terminate successfully.

\subsection{Problem Statement}

We define the problem we wish to solve as follows:

\begin{definition}
Fix $\numcompute$ \cps, $\numinput$ \dps, $\numcounters$ possible observations, and a time period $T$.  Represent the observations made by each $\dpmath{j}$ during $T$ as $\vec{O^j}$. We define the \textit{Private Set-Union Cardinality Problem} (with Set-Intersection defined similarly) as calculating $\big| \bigcup \vec{O^j} \big|$
and distributing this result to all \cps, subject to the following formal and informal constraints:
\begin{itemize}
\item The result should be \textit{accurate}, meaning that malicious \cps may not alter the data submitted by honest \dps or selectively exclude data submissions.
\item The result should be \textit{private}, meaning that no adversary may learn information about individual observations of the \dps exceeding that allowed by the Differential Privacy guarantee. In particular, we require that compromise of \dps during the collection period reveal no information about the observations recorded pre-compromise.
\item The solution should be \textit{efficient}, meaning that \dps involved in the measurement process do not consume excessive resources and that \cps can output a measurement within hours of the termination of the collection period in a practical, real-world setting.
\item The solution should be \textit{accountable}, so every time the protocol is invoked, it terminates successfully or all honest \cps come to consensus on the identity of the \cp or \cps who prevented successful termination.
\item The solution should be \textit{secure}, meaning it satisfies all above requirements so long as there is one honest \cp, and be accompanied by a thorough proof of security.
\item The solution should be \textit{practical} enough to implement and deploy in practice today, and rely only on well-tested and well-understood cryptographic assumptions.
\end{itemize}
\end{definition}
In the following, we present, implement, evaluate, and prove secure a protocol that solves the above \sys problem.
\subsection{Solution Overview}

We provide the solution in three parts. First, we describe the protocol in a hybrid model using abstract ideal functionalities. Then we proceed along two parallel paths:
\begin{enumerate}
\item A proof of security for the required primitives in the Universal Composability (UC) setting, followed by a proof of security for the hybrid-model protocol, both in Section \ref{sec:security}.
\item An implementation of the protocol using heuristically secure primitives (namely, the Fiat-Shamir Heuristic) for the purpose of practicality and efficiency, open-sourced and evaluated for efficiency in Section \ref{sec:eval}.
\end{enumerate}
The hybrid-model proof presented in Section \ref{sec:psc-proof} then applies to both the implemented protocol and as a proof of the existence and security of an unimplemented theoretical UC-secure protocol with the same structure as our implementation.

In particular, making replacements as described below, we can convert a proof of security in the UC framework in the Common Reference String (CRS) model to a proof of security in the random oracle model under sequential composition. The replacements primarily rely on the fact that the Fiat-Shamir heuristic provides non-interactive Zero-Knowledge Proofs in the Random Oracle model~\cite{rom}, and we may then apply the sequential composition theorem~\cite{goldreich-foc1-2004} in place of the UC composition theorem. Explicitly:
\begin{itemize}
	\item The broadcast functionality we use is realized by the Dolev-Strong protocol, which can be implemented with UC signatures, but in the implementation we use Schnorr signatures and apply optimistic optimizations to the protocol as described in Section \ref{sec:consensus}.
	\item We require one-to-many $\Sigma$-protocols in our proof. These are constructed by compiling $\Sigma$-protocols to UC secure accountable one-to-many Zero-Knowledge Proofs as described in Section \ref{sec:accountable-sigma}, but for implementation purposes we simply apply the Fiat-Shamir heuristic and send non-interactive zero-knowledge proofs over our broadcast channel.
	\item In the UC setting, we use the shuffle provided by Wikstr\"om and outline how it can be made accountable. For implementation purposes, we use an implementation of the shuffle given by Neff~\cite{neff-shuffle} and again apply the Fiat-Shamir heuristic. We use this shuffle because it is simple and the code is available and open source, and because additional protocols are required to construct public parameters in Wikstr\"om's shuffle that cannot be generated in a natural way by public coins. This is discussed further in Section \ref{sec:accountable-shuffles}.
\end{itemize}

%% file: sections/protocol.tex
\section{Protocol Description} \label{sec:protocol}

At a high level, the protocol proceeds in two main sections, which we separate into seven \textit{phases}. The goal is to construct a vector of noisy shuffled counters, containing one component for each possible event and additional components for noise. The counters are ciphertexts that are zero or nonzero, representing a binary value. These counters are split into two shares, an initial encrypted blind and a corresponding plaintext counter which \dps modify, recording observations obliviously. Once submissions are complete, the \cps insert encrypted noise counters, shuffle the counters, and clean them, removing non-binary information contained within the plaintext counters in a process called rerandomization. Finally, the counters can be decrypted, with the number of nonzero counters in the final tally giving the final result. We visually present the protocol in phases in Figure \ref{fig:overview}, and continue our high-level description in more detail as follows, noting that while we give all phases in sequence in our discussion, implementation, and proof, the data collection period provides a natural ``offline phase'' where \cps remain idle, meaning that Noise Generation can be trivially executed in parallel with data collection.

\begin{figure*}[t!]
	\centering
	\includegraphics[width=0.8\linewidth, height=0.45\textwidth]{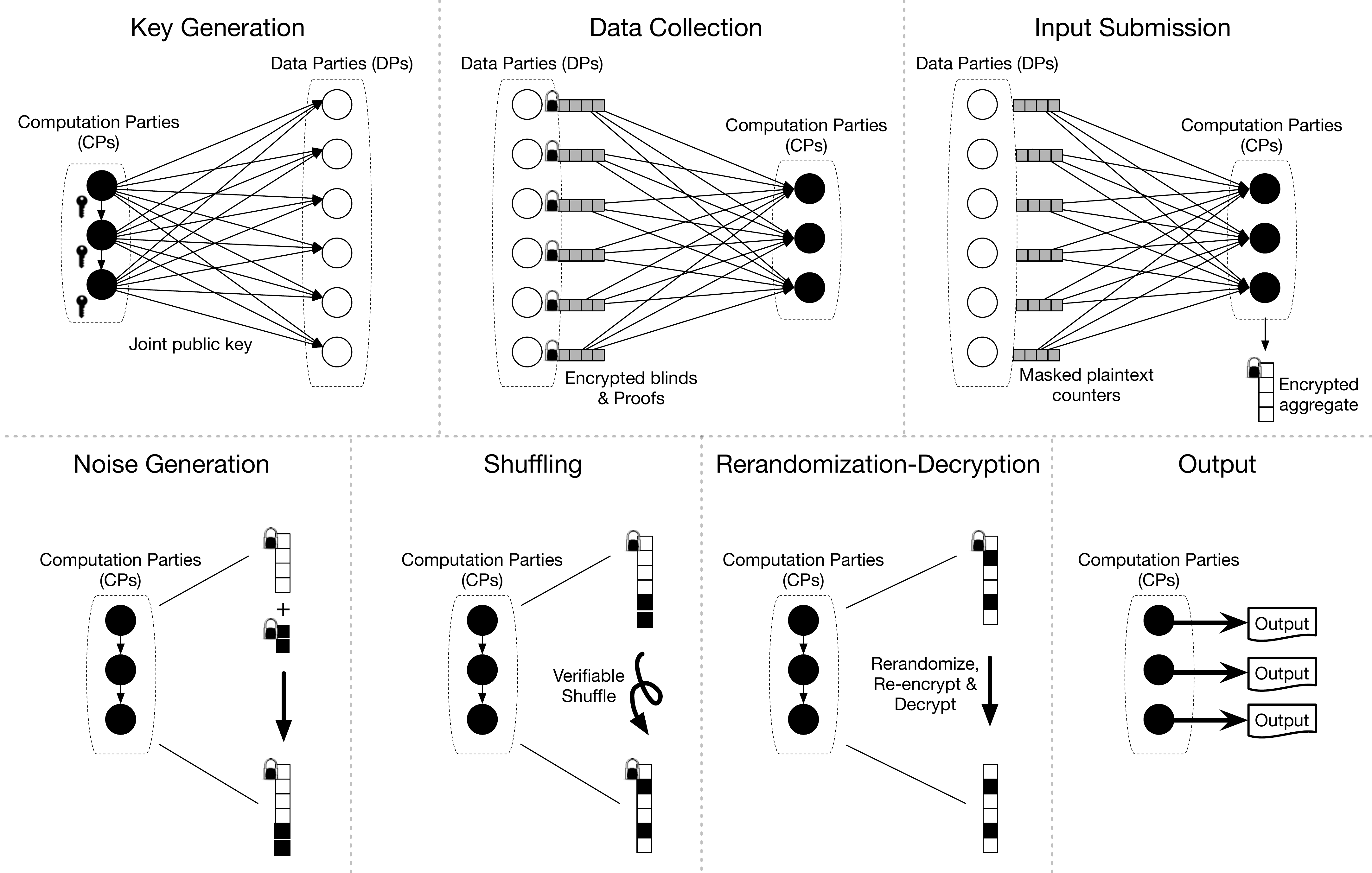}
		\vspace{-2ex}
	\caption{An overview of the major steps of the \Sys
		protocol --- the white, gray, and black square represent plaintext, blind (or blinded data), and noise respectively. The lock represents that the underlying data (\ie, the blinds, noise, or plaintext) is encrypted.}
	\label{fig:overview}
\end{figure*}

\begin{enumerate}
\item \textbf{Key Generation}. In the Key Generation phase, the \cps generate a joint public session key for each run of the protocol and distribute it to all \dps.
\item \textbf{Data Collection}. \dps construct and distribute to \cps an encrypted blind and corresponding proof of knowledge of cleartext for each counter, which represent the set of observable events to be measured, storing the negation of each value in plaintext. During the data collection period, \dps record an observation by replacing its corresponding counter with a random value.
\item \textbf{Input Submission}. When the collection period is over, the \dps broadcast their counters to the \cps, who jointly encrypt and sum these counters with the blinds submitted previously. If a \dparty has not modified its counter, the recorded counter and plaintext inside the corresponding blind will cancel. Otherwise, their sum will be some random nonzero value.
\item \textbf{Noise Generation}. \cps perform a verifiable shuffle of encryptions of the pair $(0, 1)$ $\numnoise$ times and extract the first component of the resulting shuffled vector to add $\numnoise$ elements of binomial noise to the aggregate vector of ciphertext counters.
\item \textbf{Shuffling}. \cps perform a verifiable shuffle of the noise-and-data vector, mixing the generated noise with the true counter values.
\item \textbf{Rerandomization-Decryption}. \cps rerandomize, re-encrypt, and decrypt the shuffled vector of data and noise to ensure the plaintexts carry no information beyond the parity property they represent in the protocol.
\item \textbf{Output}. The last \cp to perform this decryption broadcasts the plaintexts and the number of nonzero elements is the noisy cardinality.
\end{enumerate}

\subsection{Primitives and Assumptions}
We define two groups of parties, Data Parties or \dps, who make observations from some set of possible observations during a \textit{collection period}, and Computation Parties or \cps who receive, aggregate, clean, and publicize the privacy-safe observation statistic. We make distinct adversarial assumptions with respect to the two groups: for \dps, we work in the erasure model with adaptive corruptions, assuming that an adversary can compromise a \dparty during the collection period but requiring that it learns nothing about measurements recorded and erased before the corruption. For \cps, we work in the static corruption model and make no assumptions about secure erasures.
We denote a vector by $\vec{v}$, and write $\vec{v}_i$ to mean the $i$th component of $\vec{v}$, with superscript notation $\vec{v}^i$ indicating a sequence of vectors.
We separate the protocol into phases, each containing some number of synchronous communication rounds. Often, these phases will be split into subphases.

Throughout, we assume $g$ is a generator of a large group $G$ of prime order $q$ for which the Decisional Diffie Hellman assumption holds. We use as our primary encryption mechanism the exponential version of ElGamal, where rather than a group element $g^m$ as the plaintext, we consider the plaintext to be the integer $m$ itself. This scheme is additively homomorphic, but now requires calculating a discrete log for decryption as in the encryption scheme outlined by Benaloh~\cite{benaloh1994dense}. For our application, plaintexts encode binary values so we are interested only in whether or not a plaintext $m=0$, which is trivial to determine given $g^m$ so decryption poses no problem for us. We by convention write an ElGamal encryption of a message $m$ as
$E_y(r;m) = (g^r, y^rg^m)$
and freely refer to $r$ as its randomization factor and to $g^r$ as its ``first component''. Throughout, we refer to three homomorphic manipulations of ElGamal ciphertexts, which can be constructed only using knowledge of the public parameter $y$:
\begin{itemize}
\item \textit{Re-Encryption}, where a ciphertext has its randomization factor shifted by a constant additive factor $s$ without modifying the plaintext, converting a ciphertext $(g^r, y^rg^m) \to (g^{r+s}, y^{r+s}g^m)$.
\item \textit{Rerandomization}, where a ciphertext has its randomization factor shifted by a scalar multiplicative factor $s$, applying the same multiplication to the plaintext component, so that a ciphertext $(g^r, y^rg^m) \to (g^{rs}, y^{rs}g^{ms})$.
\item \textit{Aggregation}, where an arbitrary vector of ciphertexts encrypted to the same key are homomorphically added together by taking the componentwise product of each of the two components.
\end{itemize}

We make use of the following primitives, all in the synchronous communication model:
\vspace{-0.5mm}
\begin{itemize}
\item \f{BC}, an accountable consensus broadcast protocol, given in Figure \ref{fig:broadcast-func}.
\item \f{SKGD}, a functionality encapsulating the sub-protocol for Session Key Generation and Distribution, producing and distributing a joint ElGamal public session key shared by the \cps to all parties. This is given in Figure \ref{fig:skgd}.
\item \f{ZKP-DL}, \f{ZKP-S}, \f{ZKP-RRD}, one-to-many zero-knowledge proofs for knowledge of a discrete log, a shuffle, and for a combined re-encryption, rerandomization, and decryption operation. These are all implemented with a general zero-knowledge proof functionality given in Figure \ref{fig:psc-zkp} and each of these specific proofs is realized by a different protocol, given in detail in Section \ref{sec:security}.
\end{itemize}
We provide the security proofs for each of these functionalities including their accountability properties in Section \ref{sec:security} and provide implementation details in Section \ref{sec:eval}.

Beyond cryptographic assumptions, we assume all parties have permanent signing keys through a Public Key Infrastructure and use these to construct signatures in order to send authenticated messages. Keys for encryption are generated on a per-session basis for each protocol run so we describe their construction explicitly. The zero-knowledge proofs are given in a hybrid model using ideal functionalities for commitments that may be realized in the Common Reference String (CRS) model.
\vspace{-2mm}
\subsection{Protocol}
At a high level, the protocol proceeds in a sequence of phases, where the session keys are established and blinds are set by the \dps to guard against adaptive corruptions. The data is then collected over some period, and then transmitted to the \cps. The \cps proceed in a three-step process to prepare the data for release: adding random noise to provide differential privacy, shuffling the noise and data together, and ``rerandomizing'' the nonzero elements that encode $1$ so they do not carry any information about the exact representation selected by the \dps.

More precisely, we define the protocol $\pi_{\sys}$ in the \hybrid model by describing each phase.
\begin{enumerate}
	\item \textbf{Key Generation}.
		\begin{itemize}
			\item (Round 1). \cps send (\textsc{GenerateKeys}) to \f{SKGD}. \cps and \dps wait to receive the public session key $y$ from \f{SKGD}.
			\item (Blame). \cps receive either a set of session keys from \f{SKGD} or a blame message. If a \cp receives a blame message, they output it and terminate.
		\end{itemize}

	\item \textbf{Data Collection}.
		\begin{itemize}
			\item (Round 1). Each $\dpmath{j}$ generates random $\beta_1^j\dots \beta_\numcounters^j \in \mathbb{Z}_q$, encrypts each, and collects these ciphertexts into a vector $\vec{\beta}^j$. Each $\dpmath{j}$ broadcasts $\vec{\beta}^j$ to all \cps using \f{BC}.
			\item (Round 2). Each $\dpmath{j}$ proves knowledge of the cleartext for each component of $\vec{\beta}^j$ by proving knowledge of the discrete log of the first component of each ciphertext, $r_i^j$ using \f{ZKP-DL} to all \cps. Each $\dpmath{j}$ saves a vector of plaintexts $\vec{c}^j$ by $\vec{c}^j_i = -\beta_i^j$ and erases all information about $\vec{\beta}^j$ and the proofs from memory.
			\item (Collection Period). While the collection period is active, $\dpmath{j}$ records observing event $i$ upon receiving a message from the environment of the form (\textsc{Observation}, $i$) by setting
				$\vec{c}_i^j \gets r$
				for $r$ random in $\mathbb{Z}_q$. This phase ends when the collection period is complete.
			\item (Blame). If a broadcast or one-to-many proof from a \dparty fails, the \cps exclude the input from that \dparty during this phase and for the remainder of the protocol. If a \cp is blamed, the protocol terminates.
		\end{itemize}
	\item \textbf{Noise Generation}.
	For every component $i$ of $\vec{n}$, all \cps set $\vec{N}^{i0}$ to be the pair of deterministic encryptions to $y$ of $0$ and $1$ using first component $g^1$. \cps calculate $\vec{N}^{i\numcompute}$ in $\numcompute$ rounds:
		\begin{itemize}

			\item (Round $ij$, $i \in 1..\numnoise, j \in 1..\numcompute$).
				\begin{itemize}
					\item (Subround $ij_1$). $\cpmath{j}$ shuffles $\vec{N}^{i(j-1)} = \textsc{Shuffle}(\vec{N}^{i(j-1)}, \vec{r}^{ij}, \pi^{ij})$ and sends $\vec{N}^{ij}$ to \f{BC} with $\cpmath{j}$ generating two random integers $\vec{r}^{ij} \in \mathbb{Z}_q^2$ and a random permutation $\pi^{ij}$ on two elements.
					\item (Subround $ij_2$). $\cpmath{j}$ sends (\textsc(ZKP), ($\vec{N}^{i(j-1)}$, $\vec{N}^{ij}$), ($\vec{r}^{ij}$, $\pi^{ij}$)) to \f{ZKP-S}.
					\item (Subround $ij_3$). \cps continue iff they receive (\textsc{Proof}, $\cpmath{j}$, ($\vec{N}^{i(j-1)}$, $\vec{N}^{ij}$)).

				\end{itemize}
				\item \cps take the first ciphertext in the resulting vector, setting
		$\vec{n}_i = \vec{N}^{i\numcompute}_1$
				\item (Blame). In each round $j$ all \cps expect to begin with a known ciphertext, receive a message, and a proof that verifies. If a $\cpmath{j}$ uses the incorrect ciphertext, is blamed through \f{BC} sending the shuffled ciphertext, or the proof fails, \cps blame $\cpmath{j}$ and exit.
			\end{itemize}
	$\vec{n}$ is then a vector of length $\numnoise$.
	\item \textbf{Input Submission}.
		\begin{enumerate}
			\item (Round 1). Each $\dpmath{j}$ sends its plaintext vector $\vec{c}^j$ to all \cps using \f{BC}.
			\item (Round 2). \cps verify they received vectors of ciphertexts of the correct length and that the proofs from the Data Collection phase verify from each \dparty. For each remaining $\dpmath{j}$, \cps encrypt each component of the vector of counters $\vec{c}^j$ with the session key $y$, and componentwise homomorphically add the ciphertexts to construct a new aggregate data vector $\vec{d}$:
				$\vec{d}_i = \bigoplus_j \left( b_i^j \oplus (g, y\cdot g^{c_i^j}) \right)$
				where $\oplus$ denotes homomorphic addition of plaintexts.
			\item Blame cannot arise in this phase since \dps cannot be blamed, only excluded from the tally.

		\end{enumerate}
	\item \textbf{Shuffling}.
		All \cps set $\vec{s}^0$ as the concatenation of $\vec{n}$ and $\vec{d}$.
		\begin{enumerate}
			\item (Round $j \in 1..\numcompute$).
				\begin{enumerate}
					\item (Subround $j_1$). $\cpmath{j}$ shuffles $\vec{s}^j = \textsc{Shuffle}(\vec{s}^{(j-1)}, \vec{r}^j, \pi^j)$
						and sends $\vec{s}^j$  to \f{BC} using random integers $\vec{r}^j \in \mathbb{Z}_q^{(\numcounters + \numnoise)}$ and a random permutation $\pi^j$ on $b+n$ elements.
					\item (Subround $j_2$). $\cpmath{j}$ sends (\textsc{ZKP}, ($\vec{s}^j$, $\vec{s}^{(j-1)}$), ($\vec{r}^j$, $\pi^j$)) to \f{ZKP-S}.
					\item (Subround $j_3$). \cps continue iff they receive (\textsc{Proof}, $\cpmath{j}$, ($\vec{s}^j$, $\vec{s}^{(j-1)}$)), otherwise they blame $\cpmath{j}$.
				\end{enumerate}
				\item (Blame). The phase is structured identically to Noise Generation so blame arises in exactly the same way.
		\end{enumerate}

	\item \textbf{Rerandomization-Decryption}.
		\cps set $\vec{p}^0 \gets \vec{s}^\numcompute$.
		\begin{enumerate}
			\item (Round $j \in 1..\numcompute$).
			\begin{enumerate}
				\item (Subround $j_1$). $\cpmath{j}$ selects random $\vec{\sigma}^j$, $\vec{r}^j \in \mathbb{Z}_q^{(\numcounters + \numnoise)}$. For each $i$, denote the ciphertext $\vec{p}^{(j-1)}_{i}$ as $(a_i, b_i)$. Then $\cpmath{j}$ calculates for each $i$,
				$\alpha_i \gets \left(a_ig^{\sigma_i}\right)^{r_i}$
				If $\alpha_i$ is the identity element, $\cpmath{j}$ restarts subround $j_1$ with fresh random values. Otherwise, $\cpmath{j}$ calculates
				$\beta_i = \left(b_iy^{\sigma_i}\right)^{r_i}\alpha_i^{-x_j}$ and sets $\vec{p}^j_i \gets (\alpha_i, \beta_i)$
				 and sends $\vec{p}^j$ to \f{BC}.
			\item (Subround $j_2$). $\cpmath{j}$ sends (\textsc{ZKP}, ($y_j$, $\vec{p}^j, \vec{p}^{(j-1)}$), ($\vec{\sigma}^j$, $\vec{r}^j$, $x_j$)) to \f{ZKP-RRD}.
			\item (Subround $j_3$). \cps continue iff they receive (\textsc{Proof}, ($y_j$, $\vec{p}^j$, $\vec{p}^{(j-1)}$)), the first component of $\vec{p}_i^j$ is not the identity element, and the $y_j$ in this message is the same as the recorded $y_j$ from key exchange, otherwise they blame $\cpmath{j}$.
		\end{enumerate}
		\item (Blame). $\cpmath{j}$ is blamed in this round if $\cpmath{j}$ is blamed by \f{BC} during subround $j_1$, if the proof fails to verify or appear in subround $j_2$, or if the first element of $\vec{p}^j_i$ for any $i$ is the identity.
\end{enumerate}
	\item \textbf{Output}. \cps count the number of nonzero plaintexts in the vector $\vec{p}^\numcompute$ and subtract the expected number of noise counters $\numnoise/2$ to output the final measurement.
\end{enumerate}

%% file: sections/security.tex
\section{Security Proof}
\label{sec:security}
\subsection{Consensus Broadcast}
While certain signed messages may be used as proof of malicious behavior, in order to provide a robust accountability property against denial of service attacks we require a method by which parties can agree on whether not a given message has or has not been sent, which means we require consensus broadcast, encapsulated in a broadcast functionality \f{BC} defined in Figure \ref{fig:broadcast-func}. To achieve a consensus broadcast, we use the Dolev-Strong protocol.
\vspace{-3mm}
\subsubsection{Dolev-Strong}
\label{sec:consensus}
\begin{figure}[h]
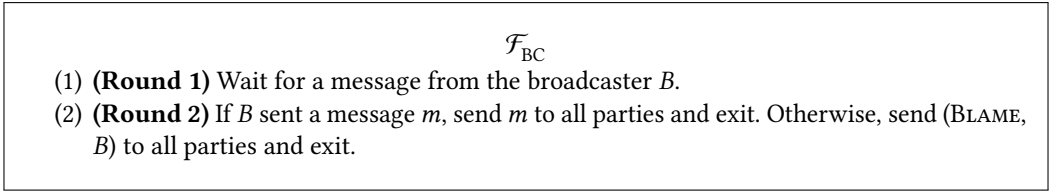

\begin{framed}

\center{\f{BC}}\\
\begin{enumerate}
	\item \textbf{(Round 1)} Wait for a message from the broadcaster $B$.
	\item \textbf{(Round 2)} If $B$ sent a message $m$, send $m$ to all parties and exit. Otherwise, send (\textsc{Blame}, $B$) to all parties and exit.
\end{enumerate}
\end{framed}
\caption{\f{BC}, the ideal broadcast functionality.}
\label{fig:broadcast-func}
\end{figure}
The Dolev-Strong protocol is thoroughly described and proved correct by the authors~\cite{dolev-strong}, in that all parties come to consensus on the (possibly empty) set of messages sent by $B$. For our specific application, when the size of this set exceeds $1$ for any honest party it may simply send these two messages to all parties and terminate,  blaming the broadcaster for equivocating. We briefly sketch the protocol below for completeness and add modifications required to include the blame messages, and finally describe optimizations that will increase performance in the case that malicious behavior does not occur.

\begin{enumerate}
	\item The broadcaster $B$ sends ($m$, \textsc{Sign}($M$)) to all parties, outputs $m$ and terminates. All other parties $P_j$ initialize a set $\textsc{accept}_j$.
	\item For rounds $j \in 1..\numcompute$,
		\begin{itemize}
			\item Upon receiving a message $M$ with valid signatures from $j$ distinct parties (including $B$) of $M$, add $M$ to the set \textsc{accept}.
			\item If the cardinality of \textsc{accept} increased this round, sign up to two new additions to \textsc{accept} chosen arbitrarily and send these messages with all of their signatures to every party in the next round.
		\end{itemize}
	\item In round $\numcompute+1$, if there is only one message in \textsc{accept}, output it. Otherwise, output (\textsc{Blame}, $B$).
\end{enumerate}

\paragraph{Optimizations}

We observe that in the scenarios we intend to deploy the protocol, the number of \cps $\numcompute$ is small, and since the protocol has an accountability property, there is a cost to malicious behavior (e.g., being removed as a participant in the protocol in an out-of-band process) so this behavior should be uncommon. In this case, we can optimize the broadcast protocol by having each party explicitly indicate to each other party in a given round when there is no message to send with a ``heartbeat'' message. This optimization allows for an all-honest group of \cps to quickly complete a broadcast without waiting for \numcompute rounds to complete, at the cost of increasing the communication cost of the protocol with these heartbeat messages proportional to $\numcompute^2$ per-\cp.

In addition, we can optimistically hope that the broadcaster sends the data: when a non-broadcaster party echoes the message to another party, they can simply echo the hash, and let the party respond as to whether or not it has already received the message. If so, nothing needs to be done with this message so we are done. If not, the sender sends the full message as well.
\subsection{Rerandomization-Decryption $\Sigma$-Protocol}
We outline and prove correct a $\Sigma$-protocol that proves a \cp has performed a rerandomization, a re-encryption and a partial decryption.
\begin{theorem}
Suppose we have a ciphertext $(A,B)$, with $g,y_i,y$ publicly known with $y_i = g^{x_i}$, $y = \prod y_i$ and we wish to present $(\alpha,\beta)$ as a re-encryption rerandomization and partial decryption of $(A,B)$ so that $\alpha = (Ag^{\sigma})^r$, $\beta = (By^{\sigma})^r\alpha^{-x_i}$ for some random shift value $\sigma$ and rerandomization value $r$ and the private key $x_i$ of party $i$. Then $A,B,\alpha,\beta,g,y$ are known to both the prover and the verifier. We describe the proof:
\begin{enumerate}
\item The Prover $P$ selects $t_1,t_2,t_3$ at random and sends $T_1 = A^{t_1}g^{t_2}$, $T_2 = B^{t_1}y^{t_2}\alpha^{t_3}$, $T_3 = g^{t_3}$
\item The Verifier $V$ sends a random challenge $c$ to $P$.
\item The Prover sends
	\begin{align*}
		r_1 = rc + t_1 & & r_2 = \sigma rc + t_2 & & r_3 = -x_ic + t_3\\
	\end{align*}
	\vspace{-6mm}
\item The Verifier accepts the proof if and only if the following three equations hold:
\begin{align*}
A^{r_1}g^{r_2} = \alpha^cT_1 & & B^{r_1}y^{r_2}\alpha^{r_3} = \beta^cT_2 & & g^{r_3} = y_i^{-c}T_3\\
\end{align*}
\vspace{-3mm}
\end{enumerate}
The above interactive proof is an HVZK proof that proves knowledge of $r,\sigma,x_i$ such that the equations above for $\alpha,\beta$ hold.
\end{theorem}

\begin{proof}
\begin{enumerate}
\item \textbf{Completeness.} The proof scheme is clearly complete: $P$ knows or generates $r,\sigma,t_1,t_2,t_3,-x_i$ so that it can properly generate $r_1,r_2,r_3$. Given these values, it is easy to see the three equations hold.
\item \textbf{Special Soundness.} Suppose the prover provides two proofs with the same commitment values $t_1,t_2,t_3$, with challenges $c_1$ and $c_2$.  Then we have:
\begin{align*}
r_1& = rc_1 + t_1&
 r_1'& = rc_2 + t_1\\
r_2& = \sigma rc_1 + t_2&
 r_2'& = \sigma rc_2 + t_2\\
r_3& = -x_ic_1 + t_3&
 r_3'& = -x_ic_2 + t_3
 \end{align*}
Then it is easy to see that $r = \frac{r_1 - r_1'}{c_1 - c_2}$ and then $\sigma = \frac{r_2 - r_2'}{r(c_1 - c_2)}$ and that $-x_i = \frac{r_3 - r_3'}{c_1 - c_2}$ so that special soundness is satisfied.
\item \textbf{Honest Verifier Zero Knowledge.} We define a simulator as follows:
The simulator behaves as an honest prover does until it receives $c$, then rewinds $V$, selects random $r_1,r_2,r_3$ and sets \begin{align*} T_1 = \frac{A^{r_1}g^{r_2}}{\alpha^c} & &
T_2 = \frac{B^{r_1}y^{r_2}\alpha^{r_3}}{\beta^c} & &
T_3 = \frac{g^{r_3}}{y_i^c}
\end{align*}
Since $V$ is an honest verifier, it, given the same randomness, provides the same challenge $c$ and the equations required for $V$ to verify both hold and the simulation is successful.
\end{enumerate}
\end{proof}
Finally, having provided the $\Sigma$-protocol and proof, we apply the compiler from Section \ref{sec:accountable-sigma} to convert it to an accountable UC secure protocol which UC-realizes the functionality defined in Figure \ref{fig:psc-zkp}, parameterized by the group element $g$, joint public key $Y$, and public keys $y_i$.

\subsection{Accountable $\Sigma$-Protocols}
\label{sec:accountable-sigma}
In order to satisfy the desired accountability property for the \cps we provide accountable zero-knowledge proofs. A natural way to solve this problem is with a non-interactive zero-knowledge proof. While there are general non-interactive zero-knowledge proof constructions for any relation secure in the UC model~\cite{UC-NIZKP} without a random oracle, they are theoretical. Instead, we provide one-to-many zero-knowledge proofs for the computations in the protocol which provide an accountability property with respect to the other participants in the protocol, but not public verifiability. Hazay and Nissim present a compiler which converts any $\Sigma$-protocol into a UC-secure zero-knowledge proof. Given an authenticated broadcast channel, we follow this approach and provide a similar compiler which converts any $\Sigma$-protocol into an accountable UC-secure one-to-many zero-knowledge proof. The two pieces required for this are a coin flipping protocol and a commitment protocol, both of which must be made accountable against faulty messages from \cps as well as instances where \cps prevent the protocol from completing by simply not sending messages required for the protocol to continue.

We obtain termination through a broadcast channel realized by the Dolev-Strong protocol~\cite{dolev-strong}, so that in every round all honest parties agree on the message sent by a broadcaster. This message may be empty or the broadcaster may equivocate, but in either case all honest parties agree that the broadcaster is faulty and the honest parties blame the broadcaster and exit.

The one-to-many commitments outlined in Canetti~\cite{canetti2002universally}, \f[1:M]{MCOM}, are non-interactive and so can be sent over the broadcast channel. We use these commitments over the broadcast channel to construct an accountable multi-party coin flipping protocol, \f{COIN}. The construction, functionality, and proof are natural and straightforward, and so we present these details in \ifappendix Appendix \ref{sec:appendix}. \else
 the appendix of an extended version of the paper, available online. \fi

 Finally, we outline how the UC-secure zero-knowledge proof of a shuffle given by Wikstr\"om~\cite{wikstrom-shuffle} can be constructed with these two accountable primitives as well, providing an accountable UC-secure verifiable shuffle.

We realize the following zero-knowledge proof functionality given in Figure 3 with the protocol provided in Figure 4, all parameterized by a relation $R$. We assume knowledge of the statement $x$ to be proved by all parties when the functionality is invoked, and that the identity of the prover is known and fixed.
\begin{figure}[h]
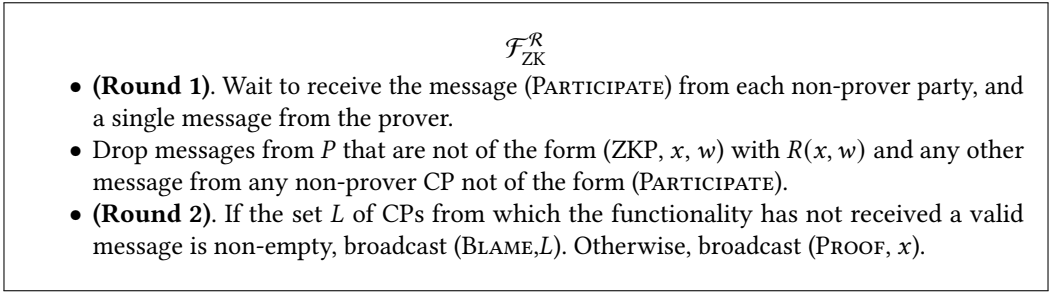

\begin{framed}

	\center{\f[$\mathcal{R}$]{ZK}}\\
\begin{itemize}
	\item \textbf{(Round 1)}. Wait to receive the message (\textsc{Participate}) from each non-prover party, and a single message from the prover.
	\item Drop messages from $P$ that are not of the form (\textsc{ZKP}, $x$, $w$) with $R(x,w)$ and any other message from any non-prover \cp not of the form (\textsc{Participate}).
	\item \textbf{(Round 2)}. If the set $L$ of \cps from which the functionality has not received a valid message is non-empty, broadcast (\textsc{Blame},$L$). Otherwise, broadcast (\textsc{Proof}, $x$).
\end{itemize}
\end{framed}
\caption{{\f[$\mathcal{R}$]{ZK}}, the ideal functionality for an interactive one-to-many zero-knowledge proof}
\label{fig:psc-zkp}
\end{figure}
\begin{figure}[h]
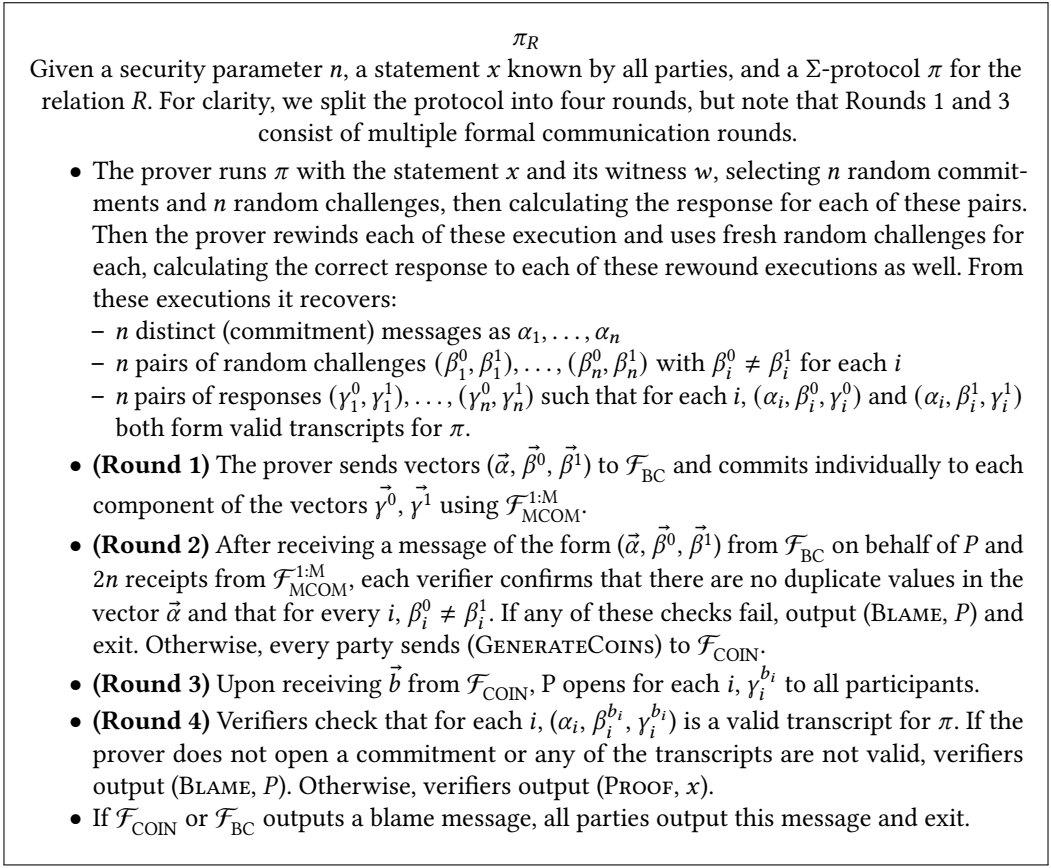

\begin{framed}

	\center{$\pi_R$}\\
	Given a security parameter $n$, a statement $x$ known by all parties, and a $\Sigma$-protocol $\pi$ for the relation $R$. For clarity, we split the protocol into four rounds, but note that Rounds 1 and 3 consist of multiple formal communication rounds.
\begin{itemize}
	\item The prover runs $\pi$ with the statement $x$ and its witness $w$, selecting $n$ random commitments and $n$ random challenges, then calculating the response for each of these pairs. Then the prover rewinds each of these execution and uses fresh random challenges for each, calculating the correct response to each of these rewound executions as well. From these executions it recovers: \begin{itemize}
			\item $n$ distinct (commitment) messages as $\alpha_1,\dots,\alpha_n$
			\item $n$ pairs of random challenges $(\beta_1^0,\beta_1^1),\dots,(\beta_n^0,\beta_n^1)$ with $\beta_i^0 \neq \beta_i^1$ for each $i$
			\item $n$ pairs of responses $(\gamma_1^0,\gamma_1^1),\dots,(\gamma_n^0,\gamma_n^1)$ such that for each $i$, $(\alpha_i, \beta_i^0, \gamma_i^0)$ and $(\alpha_i, \beta_i^1, \gamma_i^1)$ both form valid transcripts for $\pi$.
		\end{itemize}
	\item \textbf{(Round 1)} The prover sends vectors ($\vec{\alpha}$, $\vec{\beta^0}$, $\vec{\beta^1}$) to \f{BC} and commits individually to each component of the vectors $\vec{\gamma^0}$, $\vec{\gamma^1}$ using \f[1:M]{MCOM}.
\item \textbf{(Round 2)} After receiving a message of the form ($\vec{\alpha}$, $\vec{\beta^0}$, $\vec{\beta^1}$) from \f{BC} on behalf of $P$ and $2n$ receipts from \f[1:M]{MCOM}, each verifier confirms that there are no duplicate values in the vector $\vec{\alpha}$ and that for every $i$, $\beta_i^0 \neq \beta_i^1$. If any of these checks fail, output (\textsc{Blame}, $P$) and exit. Otherwise, every party sends (\textsc{GenerateCoins}) to \f{COIN}.
\item \textbf{(Round 3)} Upon receiving $\vec{b}$ from \f{COIN}, P opens for each $i$, $\gamma_i^{b_i}$ to all participants.
\item \textbf{(Round 4)} Verifiers check that for each $i$, ($\alpha_i$, $\beta_i^{b_i}$, $\gamma_i^{b_i}$) is a valid transcript for $\pi$.  If the prover does not open a commitment or any of the transcripts are not valid, verifiers output (\textsc{Blame}, $P$). Otherwise, verifiers output (\textsc{Proof}, $x$).
	\item If \f{COIN} or \f{BC} outputs a blame message, all parties output this message and exit.
\end{itemize}
\end{framed}
\caption{$\pi_R$, the hybrid-model protocol for an interactive one-to-many zero-knowledge proof}
\label{fig:uc-sigmaprotocol}
\end{figure}
\begin{theorem} If $\pi$ is a $\Sigma$-protocol for $R$, $\pi_R$ UC-realizes \f[$\mathcal{R}$]{ZK} in the static-corruption, synchronous, (\f[1:M]{MCOM}, \f{BC}, \f{COIN})-hybrid model if there is at least one honest \cp.
\end{theorem}
\begin{proof}
	We follow the approach in \cite{setops-joc2012}, but we must consider modifications required to construct the proof in a a one-to-many setting, and we must account for blame messages required for accountability. We follow the definition of $\Sigma$-protocol for $\pi$ given by Damg\aa rd in \cite{damgaard2002sigma} which defines two algorithms: a simulator $\pi_s$ for $\pi$ which accepts a uniformly random challenge value (i.e. one provided by an honest verifier) and outputs a transcript-triple indistinguishable from a real execution of $\pi$, and a witness extractor $\pi_e$ which accepts two transcripts for $\pi$ with identical first (commitment) values but distinct second (challenge) values, and outputs a valid witness $w$ for $\pi$.
	We define the simulator $S$ as follows. $S$ runs a copy of $A$ and runs every party honestly in the simulated real-model for the benefit of the environment $Z$ except where specified below:

	\textbf{Honest Prover}.
	\begin{itemize}
		\item Simulating Round 1. $S$ selects $2n$ challenge values, $\vec{\beta^0}$, $\vec{\beta^1}$ and flips $n$ coins to generate a coin vector $\vec{b}$. $S$ invokes the simulator $\pi_s$ with ($x$, $\beta_i^{b_i}$) as statement and challenge. $S$ recovers the commitment, $\alpha_i$, and response, $\gamma_i^{b_i}$ from each transcript produced by $\pi_s$, generates $n$ random values $\gamma_i^{\neg b_i}$ and collects these into vectors $\vec{\alpha}$, $\vec{\beta^0}$, $\vec{\beta^1}$, $\vec{\gamma^0}$, $\vec{\gamma^1}$, and broadcasts and commits to these values as appropriate in the protocol $\pi_R$.
		\item Simulating Round 2. $S$ sends (\textsc{GenerateCoins}) on behalf of the honest parties, and sends (\textsc{Participate}) to \f[$\mathcal{R}$]{ZK} for dishonest verifiers who sent (\textsc{GenerateCoins}) to \f{COIN} in this round.
		\item Simulating Round 3. This round only takes place if every dishonest verifier has sent \textsc{GenerateCoins} and $S$ has subsequently sent (\textsc{Participate}) on their behalfs. $S$ sends the vector $\vec{b}$ from Round 1 on behalf of \f{COIN} as a response to the verifiers. The prover is simulated honestly.
		\item Simulating Round 4. Honest verifiers are simulated honestly.
	\end{itemize}

	\textbf{Dishonest Prover}
	\begin{itemize}
		\item Simulating Round 1. $S$ recovers the vectors $\vec{\alpha}$, $\vec{\beta^0}$, $\vec{\beta^1}$, $\vec{\gamma^0}$, $\vec{\gamma^1}$ triples from the initial messages and commitments sent by $Z$ on behalf of $P$.
		\item Simulating Round 2. $S$ sends (\textsc{Participate}) to \f[$\mathcal{R}$]{ZK} on behalf of the honest parties, and on behalf of the dishonest parties that send (\textsc{GenerateCoins}) to \f{COIN}.
		\item Simulating Round 3. $S$ simulates \f{COIN} honestly, and if \f{COIN} sends a blame message instead of a bit vector, the protocol execution terminates. Otherwise, $S$ waits for $P$ to open its commitments. If $P$ opens the $n$ commitments and each of the $n$ transcripts is valid, $S$ checks the remaining $n$ unopened commitments. If any of these also forms a valid transcript, $S$ runs $\pi_e$ on one of these valid transcript pairs, recovering a witness $w$ and sends (\textsc{ZKP}, $x$, $w$) to \f[$\mathcal{R}$]{ZK}. If all of the remaining $n$ unopened commitments are invalid, $S$ halts and the simulation fails. If $P$ fails to open the $n$ commitments as required or any opened commitment generates an invalid transcript for $\pi$, $S$ sends $\perp$ to \f[$\mathcal{R}$]{ZK} on behalf of $P$.
		\item Simulating Round 4. Honest verifiers are simulated honestly.
	\end{itemize}
	We claim that the view of $Z$ in the ideal model as it interacts with $S$ and \f[$\mathcal{R}$]{ZK} is computationally indistinguishable from the view of $Z$ as it interacts with an adversary $A$ in the execution of $\pi_R$.

	\textbf{Honest Prover}. We define a sequence of hybrid executions, beginning with the real execution of $\pi_R$ (which we note is formally an execution in a hybrid model, but which we denote ``real execution'' here to avoid ambiguity) which we label $H_0$. In this case, $P$ has a witness $w$, constructs $n$ valid transcript pairs, opens one set of commitments, and the honest verifiers accept the proof. We note that any partial execution terminated by malicious verifiers is a prefix of a full execution of the protocol, with a blame component. In the case the protocol is aborted early, the set of \cps to blame is identical in all cases: $S$ simply sends participation messages on behalf of dishonest verifiers exactly when these verifiers participate in \f{COIN}. We define, for each $i>0$, $H_i$ to be a hybrid where for each $j \leq i$, the execution is modified as follows:
	\begin{itemize}
		\item $P$ constructs transcripts and behaves in round $1$ for each value with index $j > i$ as in the real execution $H_0$. For $\alpha_j$, $\beta^0_j$, $\beta^1_j$, $\gamma^0_j$, $\gamma^1_j$ with $j \leq i$, $P$ runs the protocol as $S$ does, generating these values using $\pi_s$ and without $w$.
		\item \f{COIN} behaves honestly except for $j \leq i$. For these values, \f{COIN} learns the value $b_j$ from $P$ determined in the simulation of Round 1 and if all parties send (\textsc{GenerateCoins}), \f{COIN} sends this mixed vector of $n$ values.
	\end{itemize}
	\begin{lemma}
		No PPT environment $Z$ can distinguish between $H_0$ and $H_n$ with non-negligible probability.
	\end{lemma}
	\begin{proof}
	It suffices to show that for $1 \leq i \leq n$, there is no PPT environment $Z$ that can distinguish with non-negligible probability between $H_i$ and $H_{i-1}$. Every message  of the transcripts from any execution of these two adjacent hybrids are identical except for $\alpha_i$, $\beta^0_i$, $\beta^1_i$, $\gamma^0_i$, $\gamma^1_i$, and $b_i$. $b_i$, $\beta^1_i$, and $\beta^0_i$ are identically distributed in each hybrid, as they are generated according to the same distribution (uniform in each case) by a party not controlled by $Z$. $\gamma^{\neg b_i}_i$ is never observed by $Z$ in any execution of any hybrid. Then if $Z$ can with non-negligible probability distinguish between $H_i$, $H_{i-1}$, then $Z$ can distinguish between the transcript ($\alpha_i$, $\beta_i^{b_i}$, $\gamma^{b_i}_i$) constructed by the simulator $\pi_s$ (in $H_{i}$) and a real execution of $\pi$ (in $H_{i-1}$). Since $\pi$ is computational zero knowledge, these transcripts are computationally indistinguishable so no such $Z$ can exist.
	\end{proof}
	Then it suffices to show that the difference between $H_n$ and the ideal model $I$ is a view change. We begin with $H_n$ and introduce the simulator $S$, noting that $P$ and \f{COIN} in $H_n$ and $S$ in $I$ behave identically by construction. In $H_n$, $P$ does not use its witness $w$ so all of its messages can be constructed by the simulator, which also simulates \f{COIN}, so that the coordination of the bits sent by \f{COIN} and the vector $\vec{b}$ generated by $P$ happens inside the execution of $S$. Honest \cps have no inputs and are simulated honestly so that their outputs always match the output of \f[$\mathcal{R}$]{ZK}. Finally, blame is determined based on the output of \f[$\mathcal{R}$]{ZK}, but in the ideal functionality makes this determination by simply determining whether the appropriate participation messages were sent by \cps and if the statement and witness ($x$, $w$) from the prover are valid. We note that we assume the protocol with an honest prover is run with valid input, meaning that an honest prover always sends a valid witness to \f[$\mathcal{R}$]{ZK}.

	\textbf{Dishonest Prover}
	\begin{itemize}
		\item In Round 1, 2, and 4 the simulator behaves exactly as instructed by $Z$. In round 3, the we observe that the probability that the simulation fails is the probability that the uniformly generated $n$ coins all land in the prover's favor, which is $\frac{1}{2^n} \in \textsc{Negl}(n)$.
	\end{itemize}
\end{proof}
\vspace{-3mm}
\subsection{Accountable Shuffles}
\label{sec:accountable-shuffles}
We sketch an accountable construction of Wikstr\"om's proof of a shuffle in the UC model, arguing that the abov techniques can be applied to construct a multi-verifier shuffle proof.

The proof consists of a sequence of six messages. Two are commitments, two are simple values that can be broadcast, one consists of a set of primes that may be generated by public coins, outlined in detail in~\cite{wikstrom-shuffle}. The final value that is missing is an RSA modulus which the prover cannot know. In the one-to-one case the verifier may trivially generate this value, but for our application it must be jointly generated by a group of verifiers, with no verifier knowing the factorization. We note that accountable efficient distributed protocols~\cite{efficient-rsa-keygen} to construct a public RSA modulus have been proved secure in the UC model since this proof of a shuffle was originally presented, and may be used in our setting to generate the required parameters in a distributed verifier setting with a dishonest majority and with an accountability property. Finally, we note that extraction of a witness in the proof of the shuffle is done by extracting the witness directly from the ideal zero-knowledge proof of knowledge of the cleartext functionality \f[$\mathcal{R}_C$]{ZK}. While this functionality is realized in the original proof by a version of verifiable secret sharing, we note this can also be realized by our accountable proof \f{ZKP-DL}, which we have realized above.
\subsection{Session Key Generation and Distribution}
While we assume a fixed PKI that distributes and certifies signing keys for every party, a new session key must be constructed for each invocation of the protocol, and it must be jointly constructed by all \cps in an accountable way and then distributed to the \dps so that all \dps agree on the joint key. We sketch how this may be done in the UC setting and outline a natural functionality to accomplish this, \f{SKGD}.
\begin{figure}[h!]
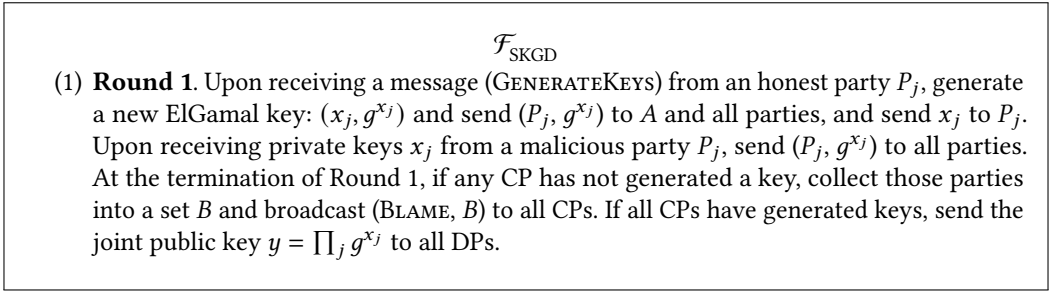

\begin{framed}
	\center{\f{SKGD}}
	\begin{enumerate}
		\item \textbf{Round 1}. Upon receiving a message (\textsc{GenerateKeys}) from an honest party $P_j$, generate a new ElGamal key: $(x_j, g^{x_j})$ and send ($P_j$, $g^{x_j}$) to $A$ and all parties, and send $x_j$ to $P_j$. Upon receiving private keys $x_j$ from a malicious party $P_j$, send ($P_j$, $g^{x_j}$) to all parties. At the termination of Round 1, if any \cp has not generated a key, collect those parties into a set $B$ and broadcast (\textsc{Blame}, $B$) to all \cps. If all \cps have generated keys, send the joint public key $y = \prod_j g^{x_j}$
		to all \dps.
	\end{enumerate}
\end{framed}
\caption{\f{SKGD}, an ideal functionality for ElGamal Key Generation and Distribution}
\label{fig:skgd}
\end{figure}
This functionality may be realized by letting each \cp broadcast its public keys and accountably prove knowledge of the corresponding private keys (including the \dps) but this is prohibitively inefficient even in theory with the expected thousands of \dps. Rather, we sketch a more efficient method that mirrors our implementation but may be realized in the UC setting:
\begin{enumerate}
	\item \cps generate a private key $x_j$.
	\item \cps broadcast the corresponding public key $g^{x_j}$ to all other \cps using \f{BC}.
	\item \cps prove knowledge of the discrete log $x_j$ accountably \textit{to each other only} using \f{ZKP-DL}.
	\item \cps each broadcast a signature of the joint public key $y = \prod_j g^{x_j}$, the product of the public keys of all \cps with signature $\sigma_j = \textsc{Sign}(\textsc{Session-Key} | \textit{sid} | j | y)$ using a UC signature functionality.
	\item \cps each send ($y$, $\vec{\sigma}$) to every \dparty via a point-to-point authenticated channel \f{auth}.
	\item \dps wait to receive a public key signed by all \cps. If they receive zero or more than one distinct signed key, the \dparty halts. Otherwise, they accept this single key and move forward with the protocol.
	\item Blame is assigned by all \cps if a \cp does not send a message in one of the above phases as expected, or if they send an invalid message: the ZKP is not valid, the signature is not valid, or the signature is on the wrong value.
\end{enumerate}
In this construction, each \cp sends each \dparty a single message. We sketch the security proof: all \cps agree on the keys of each \cp or one is blamed in the first phase. If one is blamed, the protocol terminates. If not, the honest \cp sends a valid key pair and vector of signatures to each \dparty, ensuring each \dparty has at least one valid message. Since one \cp is honest, let it be $CP_h$, the UC signature functionality constructs at most one $\sigma_h$, so no \dparty receives more than one valid public key pair.

\subsection{\Sys Ideal Functionality}
\label{proof-sys}
We give the ideal functionality for $\Sys$ in Figure \ref{f-psc-appendix}.
\begin{figure}[h!]
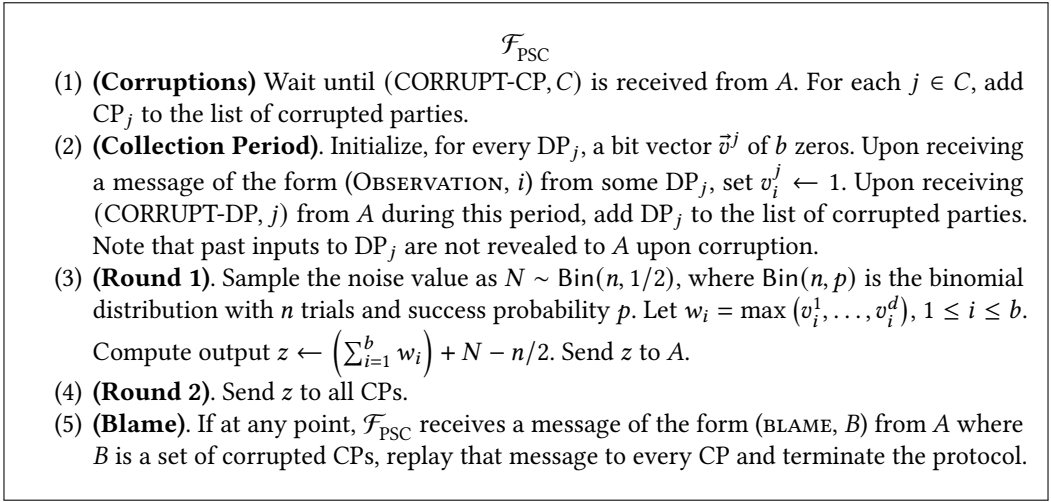

\begin{framed}
\center{\f{\sys}}
\begin{enumerate}
	\item \textbf{(Corruptions)} Wait until $(\textrm{CORRUPT-\cp}, C)$ is received from $A$. For each $j \in C$, add $\cpmath{j}$ to the list of corrupted parties.
	\item \textbf{(Collection Period)}. Initialize, for every $\dpmath{j}$, a bit vector $\vec{v}^j$ of $\numcounters$ zeros. Upon receiving a message of the form (\textsc{Observation}, $i$) from some $\dpmath{j}$, set $v^j_i \gets 1$. Upon receiving $(\textrm{CORRUPT-\dparty}, j)$ from $A$ during this period, add $\dpmath{j}$ to the list of corrupted parties. Note that past inputs to $\dpmath{j}$ are not revealed to $A$ upon corruption.
\item \textbf{(Round 1)}. Sample the noise value as $N\sim \mathsf{Bin}(n, 1/2)$, where $\mathsf{Bin}(n,p)$ is the binomial distribution with $n$ trials and success probability $p$. Let $w_i = \max\left(v^1_i,\ldots, v^{\numinput}_i\right)$, $1\le i\le \numcounters$. Compute output $z\gets \left(\sum_{i=1}^{\numcounters} w_i\right) + N - n/2$. Send $z$ to $A$.
\item \textbf{(Round 2)}. Send $z$ to all \cps.
\item \textbf{(Blame)}. If at any point, \f{\sys} receives a message of the form (\textsc{blame}, $B$) from $A$ where $B$ is a set of corrupted \cps, replay that message to every \cp and terminate the protocol.
\end{enumerate}
\end{framed}
\caption{\f{\sys}, the Private Set-Union Cardinality ideal functionality}
\label{f-psc-appendix}
\end{figure}
\subsection{Simulator Definition}
We define $S$, the ideal model adversary parametrized by a real-world adversary $A$, which interacts with \f{\sys} and an arbitrary environment machine $Z$. We assume a dummy adversary $A$ that simply relays messages to and from $Z$ since this generalizes all adversaries~\cite{canetti2002universally}. The primary challenges here are that $S$ must extract inputs of the corrupted parties before the result is calculated by \f{\sys} since these values are its inputs, and must fool $Z$ into believing the honest parties are executing the protocol faithfully on input given by $Z$. This means the behavior of honest parties with the ``dummy'' inputs from $S$ must be indistinguishable from honest parties executing the protocol with the true input from $Z$, and also that the final output of \f{\sys} is distributed as dictated by $Z$. Finally, $S$ must note observable malicious behavior that generates a blame message and forward these onto \f{\sys} accordingly. We define $\cp_h$ as the last honest \cp in the \cp ordering.
\begin{figure}
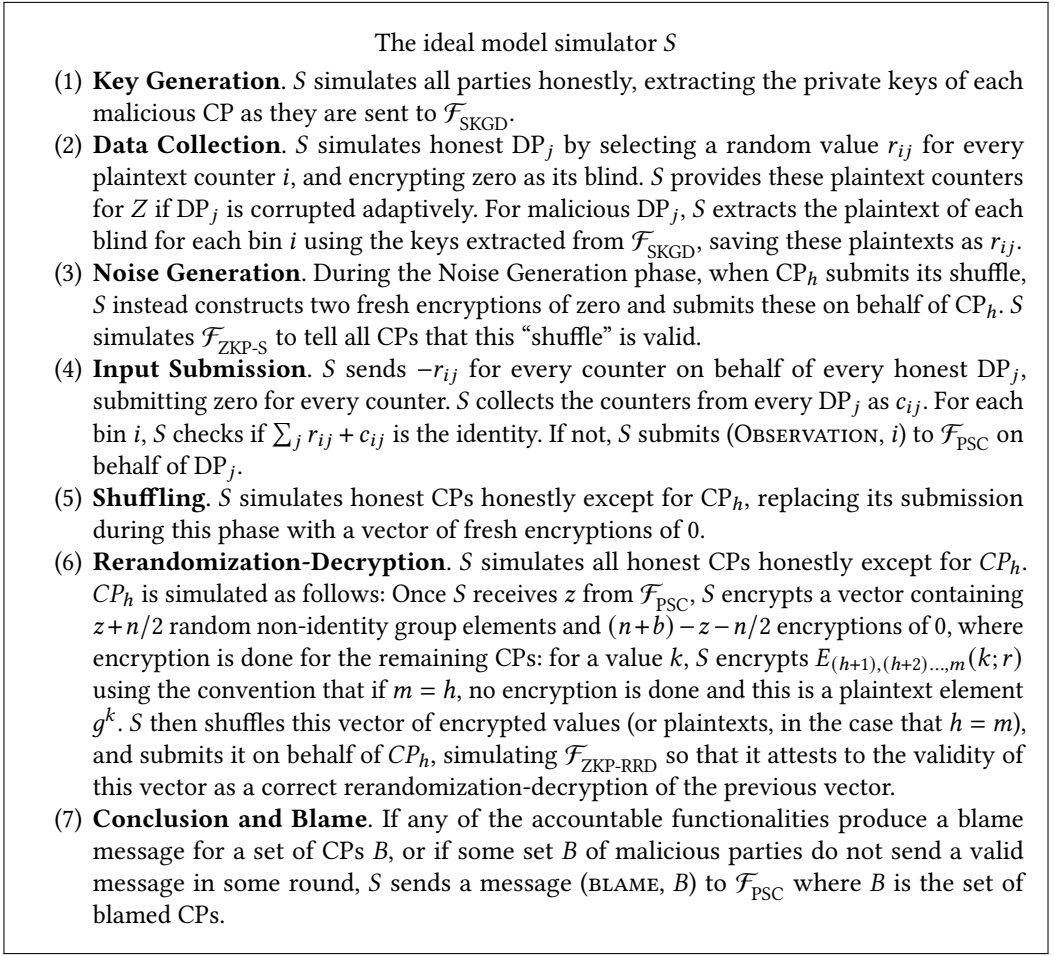

\begin{framed}
\center{The ideal model simulator $S$}

\begin{enumerate}
	\item \textbf{Key Generation}. $S$ simulates all parties honestly, extracting the private keys of each malicious \cp as they are sent to \f{SKGD}.
	\item \textbf{Data Collection}. $S$ simulates honest $\dpmath{j}$ by selecting a random value $r_{ij}$ for every plaintext counter $i$, and encrypting zero as its blind. $S$ provides these plaintext counters for $Z$ if $\dpmath{j}$ is corrupted adaptively. For malicious $\dpmath{j}$, $S$ extracts the plaintext of each blind for each bin $i$ using the keys extracted from \f{SKGD}, saving these plaintexts as $r_{ij}$.
	\item \textbf{Noise Generation}. During the Noise Generation phase, when $\cpmath{h}$ submits its shuffle, $S$ instead constructs two fresh encryptions of zero and submits these on behalf of $\cpmath{h}$. $S$ simulates \f{ZKP-S} to tell all \cps that this ``shuffle'' is valid.
	\item \textbf{Input Submission}. $S$ sends $-r_{ij}$ for every counter on behalf of every honest $\dpmath{j}$, submitting zero for every counter. $S$ collects the counters from every $\dpmath{j}$ as $c_{ij}$. For each bin $i$, $S$ checks if $\sum_j r_{ij} + c_{ij}$ is the identity. If not, $S$ submits \textsc{(Observation, $i$)} to \f{\sys} on behalf of $\dpmath{j}$.
	\item \textbf{Shuffling}. $S$ simulates honest \cps honestly except for $\cpmath{h}$, replacing its submission during this phase with a vector of fresh encryptions of $0$.
	\item \textbf{Rerandomization-Decryption}. $S$ simulates all honest \cps honestly except for $\cp_h$. $\cp_h$ is simulated as follows: Once $S$ receives $z$ from \f{\sys}, $S$ encrypts a vector containing $z +\numnoise/2$ random non-identity group elements and $(\numnoise + \numcounters) - z - \numnoise/2$ encryptions of $0$, where encryption is done for the remaining \cps: for a value $k$, $S$ encrypts $E_{(h+1),(h+2)\dots,\numcompute}(k;r)$ using the convention that if $\numcompute = h$, no encryption is done and this is a plaintext element $g^k$. $S$ then shuffles this vector of encrypted values (or plaintexts, in the case that $h=\numcompute$), and submits it on behalf of $\cp_h$, simulating \f{ZKP-RRD} so that it attests to the validity of this vector as a correct rerandomization-decryption of the previous vector.
	\item \textbf{Conclusion and Blame}. If any of the accountable functionalities produce a blame message for a set of \cps $B$, or if some set $B$ of malicious parties do not send a valid message in some round, $S$ sends a message (\textsc{blame}, $B$) to \f{\sys} where $B$ is the set of blamed \cps.
\end{enumerate}
\end{framed}
\caption{The ideal model simulator $S$}
\label{fig:psc-sim}
\end{figure}
\vspace{-3mm}

\subsection{Proof}
\label{sec:psc-proof}

Fix a security parameter $k \in \mathbb{N}$ and consider any PPT environment machine $Z$. $Z$ interacts with a given protocol $\pi$ by submitting environmental inputs to honest parties and reading their outputs, communicating with an adversary $A$, and corrupting parties at will according to our assumptions: \dps any time during the collection period, and \cps statically at the beginning of the protocol. We define the output of $Z \in \{0,1\}$ in the real execution of $\pi_{\textsc{\sys}}$ as $\textsc{REAL}_{\pi_{\textsc{\sys}},A,Z}(k)$ and the output of $Z$ in the ideal execution as $\textsc{IDEAL}_{\textrm{\f{\sys}},S,Z}(k)$, with the ideal functionality \f{\sys}, simulator $S$, defined above. For each $k$ each of the above outputs is a random variable, taken over the random input tapes provided to the parties and functionalities in each execution. If we fix the random input to $Z$, and then observe that the output of $Z$ is a function of the execution trace it observes, so if the messages sent by other parties in two distinct protocol runs with the same $Z$ are identically distributed, the output of $Z$ in each case is identically distributed as well.

We define the following sequence of (sequences of) hybrid executions which are protocols that interact with $Z$ defined by incremental modifications from $\textsc{REAL}_{\pi_{\textsc{\sys}},A,Z}(k,z)$, where each time the behavior of honest parties is adjusted, using information collected by the hybrid itself from its interaction with $Z$. These hybrids ``simulate'' the execution for the benefit of $Z$ as the simulator does, but generally can recover inputs that $S$ does not have access to like the direct inputs from $Z$ to the honest parties, the goal being that the behavior of the last hybrid in this sequence is the behavior of $S$. Each incremental transition is a transition between two executions that are identical (meaning the distribution of messages sent by honest parties is constructed identically before and after the change in hybrid executions) or are computationally indistinguishable assuming the Decisional Diffie-Helman assumption for our group $G$. We denote a hybrid execution with a single letter (e.g. $H$) with the fixed $k$, $Z$ implicit.

Since the ideal-model simulator $S$ does not know the inputs of honest parties, we transition from $\textsc{REAL}_{\pi_{\textsc{\sys}},A,Z}(k,z)$ by replacing the messages sent by honest parties with ``dummy'' values, which will be produced by $S$ in the ideal model. We reduce the environment $Z$'s ability to distinguish between two hybrid executions to winning the IND-CPA game for ElGamal, but in order to do this the hybrid must correctly ``decrypt'' a value encrypted to the public key of the IND-CPA challenger, which means we must recover these submitted secrets by other means to faithfully reproduce and swap-in the final result expected by $Z$ based on the inputs it has submitted.

The hybrid executions record and reproduce these final results in order to provide messages sent by honest parties that replace decryption of ciphertexts for which they do not have the keys. We formalize this process in a ``Ghost Execution''.

\subsubsection{The Ghost Execution and ``Fixing'' the Simulation}
\label{sec:proof-fix}
In the last phase of the protocol, Rerandomization-Decryption, we observe that the last honest $\cpmath{h}$ produces a vector of ciphertexts (or plaintexts, if $m=h$) as its output, but if we fix the plaintexts in the vector, then the output of $\cpmath{h}$ is a vector where each component is a uniformly random fresh encryption of its corresponding plaintext multiplied by a uniformly random nonzero Rerandomization factor. This is because $\cpmath{h}$ re-encrypts the ciphertext alongside the rerandomization-decryption. This means if the plaintexts are known to $\cpmath{h}$ in this phase, the output it constructs by honestly decrypting the vector it received to Rerandomize-Decrypt is identical in distribution to the output it constructs by constructing a fresh encryption of the decryption of that plaintext times a uniformly random nonzero value. We call this replacement process in the final phase ``fixing the execution''.

Any hybrid execution can extract parameters used for each operation in every phase by all parties, and use them to reconstruct the plaintext in the last phase as required, since each component that goes into the calculation of these plaintexts is constructed by an honest party (and therefore run by the hybrid directly) or for malicious parties, extracted as follows:
\begin{enumerate}
	\item Blinds from malicious parties are extracted from \f{ZKP-DL}.
	\item Malicious parties' shuffles in noise generation and shuffling are extracted from \f{ZKP-S}.
	\item The Rerandomization-Decryption factors are extracted from \f{ZKP-RRD}.
\end{enumerate}

Hybrid executions keep track of these values, extracting the values above selected by $Z$, and recording the values selected by honest parties. The values are stored in a parallel execution called the ghost execution, and in the last phase of the protocol, Rerandomization-Decryption, $\cpmath{h}$ uses these plaintexts constructed through direct extraction rather than its input. \f{ZKP-RRD} indicates to all other \cps that these values are a correct Rerandomization-Decryption as expected, and in this way we ``fix the execution''.

\subsubsection{Blind Submission}
We set $B_{10} = \textsc{REAL}_{\pi_{\textsc{\sys}},A,Z}(k)$. We define a sequence of $\numinput\numcounters$ executions:

\[B_{\textsc{seq}} = \langle B_{10},B_{11},B_{12},\dots,B_{1\numinput},\dots,B_{2\numinput},\dots,B_{\numcounters\numinput}\rangle\]
where each $B_{ij}$ with $i \in 1..\numcounters$, $j \in 1..\numinput$. Define $B_{ij}$ by:
\begin{itemize}
\item Run $B_{10}$ honestly, except that for counter $t$ and dp $D$, during blind submission, if $D$ is honest and either $t < i$ or $t = i$ and $D \leq j$, submit a fresh encryption of zero on behalf of $D$ for its blind in bin $t$. Record a random blind $b$ in the Ghost Execution for each of these honest \dps, and save as their corresponding plaintext counter $-b$.
\item $\cpmath{h}$ fixes the simulation as described in \ref{sec:proof-fix}, submitting values from the Ghost Execution as its message during the Rerandomization-Decryption phase.

\end{itemize}
\begin{lemma}
	Every adjacent pair of hybrid executions in $B_{\textsc{seq}}$ is computationally indistinguishable.
\end{lemma}

\begin{proof}

Suppose we have two adjacent hybrid executions $B_l$, $B_r$. We note that by construction, in each adjacent pair of hybrid executions the behavior of honest parties differs by exactly one ciphertext submitted by one \dparty. Assume there exists a PPT environment $Z$ that can distinguish between these executions. We define an algorithm $A$ that runs $Z$ and claim that the advantage of $Z$ distinguishing $B_l$ and $B_r$ is the same as the advantage of $A$ in the IND-CPA game for ElGamal. $A$ plays the IND-CPA game for ElGamal:
\begin{enumerate}
	\item A accepts a public key $y_c$ from the challenger in the IND-CPA game. A runs an execution of $B_l$ and replaces $y_h$ with $y_c$ in the protocol. Each time a functionality (more precisely, ZKPs in the Rerandomization-Decryption phase and \f{SKGD}) requires $\cpmath{h}$ prove knowledge of the secret key, the ZKP functionality accepts the ``proof'' provided by $\cpmath{h}$.
	\item For the blind in bin $t$ submitted by \dparty $D$, $A$ calculates a random blind $b$ value as an honest \dparty would and sends the pair of plaintexts ($b$, $0$) to the challenger in the IND-CPA game and receives a ciphertext $C$.
	\item A sends $C$ as the ciphertext blind for bin $t$ on behalf of \dparty $D$, further encrypting $C$ to the public keys of all other \cps.
	\item A sets $-b$ as the plaintext counter for bin $t$ and simulates the \dparty exactly as in $B_r$.
	\item Instead of decrypting, $A$ replaces each output of $\cpmath{h}$ during Rerandomization-Decryption as described in \ref{sec:proof-fix} and gives to the IND-CPA challenger the output of $Z$.
\end{enumerate}
If the challenger encrypts $b$, the blind for $t$ is random and the plaintext is its opposite and this is an execution of $B_l$.
If the challenger encrypts $0$, the ``fixed'' value output by $\cpmath{h}$ is the same as before, but this is an execution of $B_r$.

Therefore $A$ has the same nonzero advantage in the IND-CPA game for ElGamal that $Z$ does distinguishing between $B_l$, $B_r$ which is non-negligible, violating the DDH assumption.

\end{proof}
\subsubsection{Data Collection}
We set $D_{10} = B_{\numcounters\numinput}$. We define a sequence of $\numinput\numcounters$ hybrid executions:
\[D_{\textsc{seq}} = D_{10},D_{11},D_{12},\dots,D_{1\numinput},\dots,D_{2\numinput},\dots,D_{\numcounters\numinput}\]
where each $D_{ij}$ with $i \in 1..\numcounters$, $j \in 1..\numinput$. Define $D_{ij}$ by:
\begin{itemize}
\item Run $D_{10}$ honestly, except that when a currently-honest \dparty $D$ is instructed to record an observation for counter $t$, if $t < i$ or $t = i$ and $D \leq j$, ignore it. Add a random value $r_t$ to the sum for counter $t$ in the ghost execution.
\item $\cpmath{h}$ fixes the simulation as described in \ref{sec:proof-fix}.

\end{itemize}
\begin{lemma}
	Every adjacent pair of hybrid executions in $D_{\textsc{seq}}$ is identically distributed.
\end{lemma}
\begin{proof}
Suppose we have two adjacent hybrid executions $D_l$, $D_r$. We note that by construction each adjacent pair of hybrid executions differs by at most one counter $t$ submitted by one honest $\dpmath{j}$. If no instruction of the form (\textsc{Observation}, $t$) is sent to $\dpmath{j}$ while $\dpmath{j}$ is honest, the two executions do not differ at all. So assume such a message is received. Before this occurs, $\dpmath{j}$ stores a value $-b$ for counter $t$, and the ghost execution has recorded $b$ as the corresponding blind. Then, in $D_l$, when $\dpmath{j}$ is sent an \textsc{Observation} message, $\dpmath{j}$ replaces $-b$ with a fresh random value $b'$. In $D_r$, $\dpmath{j}$ sets $-b$, but the blind for $\dpmath{j}$ for counter $t$ is replaced with a fresh random value $r_t$ in ghost execution. We observe now that $\dpmath{j}$ may be adaptively corrupted at this point, or it may not. Either way, crucially the plaintext counter stored for $\dpmath{j}$ for counter $t$ is revealed to $Z$ \textit{after} the observation has taken place, and the changes described above have been made. We describe the differences between these executions, which consist of the value of the plaintext counter for $t$ held by $\dpmath{j}$ and the sum of the inputs for counter $t$ by $\dpmath{j}$ in the ghost execution after the Input Submission phase.

\begin{itemize}
	\item In $D_l$, $\dpmath{j}$'s plaintext counter is $b'$, a uniformly random value unrelated to any other value. The sum in the ghost execution for the submissions on behalf of $\dpmath{j}$ is $b$ and a value selected by $Z$ with no knowledge of $b$.
	\item In $D_r$, $\dpmath{j}$'s plaintext counter is $-b$, where $b$ is selected uniformly at random. The sum in the ghost execution for the submissions on behalf of $\dpmath{j}$ is $b + r_t$ and an element selected by $Z$ with knowledge of $-b$.
\end{itemize}
The plaintext counters $b'$ and $-b$ are both uniformly random values ($b'$ is chosen this way, $-b$ is random since $b$ is).
The sums in the ghost execution are uniformly random as well, since they both contain a summand which is uniformly random ($b$ in $D_l$ and $r_t$ in $D_r$), and about which $Z$ never learns any information.
\end{proof}
\subsubsection{Noise Generation}
We set $D_{\numcounters\numinput}  = N_0$ and define a sequence of hybrid executions beginning with the final execution in the previous step:
\[N_{\textsc{SEQ}} = \langle N_0, N_1, \dots, N_\numnoise\rangle\]
\begin{itemize}

	\item $N_j = N_{j-1}$ except that during the Noise Generation phase, for rounds $i \leq j$, when $\cpmath{h}$ calculates $\vec{N}^{ih}$, it instead replaces this vector with a vector of two fresh distinct encryptions of zero and \f{ZKP-S} indicates to all other \cps that these ciphertexts are a shuffle of the previous ciphertexts.
	\item $\cpmath{h}$ fixes the simulation as described in \ref{sec:proof-fix}.
\end{itemize}
\begin{lemma}
	Every adjacent pair of hybrid executions $N_j$, $N_{j+1}$ is computationally indistinguishable.
\end{lemma}
\begin{proof}
Suppose an environment $Z$ can distinguish between the two adjacent executions. We construct an IND-CPA adversary $A$ that has non-negligible advantage in the IND-CPA game for ElGamal.
We play the IND-CPA game for ElGamal:
\begin{enumerate}
	\item $A$ accepts a public key $y_c$ from the challenger in the IND-CPA game. $A$ runs $N_j$ with $y_c$ as the public key for $\cpmath{h}$ until the Noise Generation phase. Proof functionalities that require knowledge of the secret key for $y_c$ accept $\cpmath{h}$'s proofs as correct.
	\item $A$ sends the plaintexts ($0$, $1$) to the challenger and receives a ciphertext $C$.
	\item During Noise Generation, $\cpmath{h}$ performs the shuffle as before but replaces the ciphertext that is an encryption of $1$ with $C$, encrypted with the keys of the remaining \cps.
	\item $A$ fixes these ciphertexts during Rerandomization-Decryption as described in Section \ref{sec:proof-fix}.
	\item $A$ submits the output of $Z$ to the challenger in the IND-CPA game.
\end{enumerate}
Then we observe that if the challenger encrypts $1$, we have a fresh encryption of $1$ and our new re-encryption of $0$ in a random order, which is a valid shuffle of the previous ciphertexts and we are in execution $N_j$. If the challenger encrypts $0$, we submit two fresh encryptions of $0$ on behalf of $\cpmath{h}$ and this is $N_{j+1}$. Then our advantage in the IND-CPA game is the same as the advantage of $Z$ distinguishing between $N_j$, $N_{j+1}$ which is negligible.
\end{proof}
\subsubsection{Shuffling}
We define a sequence of hybrid executions beginning with the previous transition:
\[N_{\numnoise}  = S_0, S_1, \dots S_{\numnoise+\numcounters}\]
\begin{itemize}
	\item $S_j = S_{j-1}$ except that during the Shuffling phase, when $\cpmath{h}$ calculates $\vec{d}^h$, for components $i \leq j$, $\cpmath{h}$ broadcasts an encryption of zero and \f{ZKP-S} indicates to all other \cps that these ciphertexts are a shuffle of the previous ciphertexts.
	\item $\cpmath{h}$ fixes the simulation as described in \ref{sec:proof-fix}.
\end{itemize}
\begin{lemma}
	Every adjacent pair of hybrid executions $S_j$, $S_{j+1}$ is computationally indistinguishable.
\end{lemma}
\begin{proof}
Suppose an environment $Z$ can distinguish between the two adjacent executions. We construct an IND-CPA adversary $A$ that has non-negligible advantage in the IND-CPA game for ElGamal.
We play the IND-CPA game for ElGamal:
\begin{enumerate}
	\item $A$ accepts a public key $y_c$ from the challenger in the IND-CPA game. $A$ runs $S_j$ with $Z$ with $y_c$ as the public key for $\cpmath{h}$ until the Shuffling phase. ZKP functionalities that require the secret key for $\cpmath{h}$ instead simply accept all proofs from $\cpmath{h}$ without it.
	\item $A$ performs the shuffle and recovers the plaintext $t$ for component $j+1$. $A$ sends the plaintexts ($0$, $t$) to the challenger and receives a ciphertext $C$.
	\item $A$ during Shuffling, on behalf of $\cpmath{h}$, performs the shuffle as before but replaces the ciphertext for component $j+1$ with $C$, further encrypted with the keys of the remaining \cps.
	\item $A$ fixes the output during Rerandomization-Decryption as described in Section \ref{sec:proof-fix}.
	\item $A$ submits the output of $Z$ to the challenger in the IND-CPA game.
\end{enumerate}
Then we observe that if the challenger encrypts $t$ this is exactly $S_j$, while if the challenger encrypts $0$ this is $S_{j+1}$. Then our advantage in the IND-CPA game is the same as the advantage of $Z$ distinguishing between $S_j$, $S_{j+1}$ which is negligible.
\end{proof}

\subsubsection{Introduction of the Ideal Functionality}

In $S_{\numnoise + \numcounters}$, during the Rerandomization-Decryption phase  $\cpmath{h}$ outputs a vector of encryptions of the values  recorded in the Ghost Execution. This vector is completely characterized by three values: a number of nonzero elements $o \in \mathbb{N}$, a vector of plaintexts $\vec{o}$ of length $o$, and an permutation $\pi$ on $\numcounters + \numnoise$ elements that selects an ordering of the $o$ nonzero plaintexts and the remaining zeros. Any specific execution of any hybrid selects a precise value for each of these variables. We define the final hybrid $F$:

\begin{itemize}
	\item Execute identically to $S_{\numnoise + \numcounters}$ except instead of ``fixing'' the final vector of outputs during Rerandomization-Decryption on behalf of $\cpmath{h}$, extract only the number of nonzero counters from the Ghost Execution and save this value $o$.
	\item Construct a vector as follows: take $o$ nonzero elements uniformly at random, append zero $n + b - o$ times, apply a uniformly random shuffle to the vector then encrypt each component to the (possibly empty) joint public key $\prod_{j=h+1}^\numcompute y_j$. Send this as the output for $\cpmath{h}$.
\end{itemize}

\begin{lemma} $F$ and $S_{\numnoise + \numcounters}$ are identically distributed. \end{lemma}

\begin{proof}
The executions are identical to the point where $\cpmath{h}$ sends its output in the Rerandomization-Decryption phase by construction. Therefore, it suffices to show that the distribution of the output, characterized as $o$, $\vec{o}$, $\pi$ of $\cpmath{h}$ in the two executions $F$, $S_{\numnoise+\numcounters}$ is identical.

We observe that $o$ is determined  identically in both executions since it is selected before the processes formally diverge. In the execution $F$, $\vec{o}$ and $\pi$ are drawn uniformly at random among nonzero elements of $\mathbb{Z}_p$ (we refer to elements by their discrete log with respect to the generator $g$), and uniformly among permutations, respectively. We show this is the case in $S_{\numnoise + \numcounters}$ as well.

Consider a given element of $\vec{o}$. In $S_{\numnoise + \numcounters}$ it is constructed by application of a sequence of $h$ multiplications of nonzero elements $\mu_1...\mu_{h}$ to the original input, whether that input is generated at random by an honest party, constructed by noise, or submitted by malicious parties. Call the input $\alpha$. Then the output of $\cpmath{h}$ for this component is $\beta = \alpha\mu_1\dots\mu_{h}$ and we observe that $\mu_h$ is selected by the honest $\cpmath{h}$ uniformly at random when the value is produced, and that each $\mu$ is nonzero since \cps cannot submit zero rerandomization factors, and $\alpha$ is nonzero by the definition of $\vec{o}$. Therefore $\alpha\mu_1\dots\mu_{h-1}$ is nonzero hence invertible in a group of prime order, so a uniformly random selection of $\mu_h$ gives a uniformly random distribution among nonzero elements for $\beta$.

Second, consider $\pi$: $\pi$ is determined by application of the $\numcompute$ permutations $\pi = \pi_{h}\circ\dots\circ\pi_1$ where each $\pi_i$ is the permutation selected by $\cpmath{i}$ during the Shuffling phase and we observe that since every permutation is invertible, selection of $\pi_h$ uniquely determines $\pi$. $\pi_h$ is chosen during shuffling but no information about the choice is revealed until $\cpmath{h}$ submits its values during Rerandomization-Decryption, so $\cpmath{h}$ may without altering the execution at all make this choice when its final output is constructed, selecting $\pi_h$ at random and ensuring $\pi$ is as well.

\end{proof}

\begin{proposition} $F$ and the ideal model execution \ideal are statistically indistinguishable.
\label{prop:proof-id-I}\end{proposition}
\begin{proof}
After the sequence of transitions above, the honest parties in $F$ behave identically to how they are simulated by $S$ in Figure \ref{fig:psc-sim} except for the calculation of the value $o$, which is recovered directly from \f{\sys} in \ideal, and extracted from the aggregate in the ghost execution during Rerandomization-Decryption in $F$. We need to show that the value $o$ in both executions is the same. We split the value $o$ (from either execution) into two nonnegative summands since components of the final vector consist of two disjoint sets: data counters, determined by aggregating inputs, and noise counters, determined by shuffling ciphertexts in the Noise Generation phase.

\textbf{Nonzero Counters from Data}.

Fix a data counter $t$. We show that $t$ contributes to the count $o$ (i.e., is nonzero) in $F$ if and only if it contributes in \ideal. Throughout this section, we consider ``plaintexts'' to be the discrete log of these encrypted values with respect to $g$, and call a plaintext $0$ if it is $g^0$ i.e. the group identity element.

For data counter $t$, we write $\sigma_t^F$ as the sum of $3j$ values in the execution of $F$: each $\dpmath{j}$ submits a blind during the beginning of the Data Collection phase. Following notation in the simulator definition, call the plaintext of this blind $r_{tj}$. Each $\dpmath{j}$ also submits a plaintext counter $c_{tj}$. Finally, the Ghost Execution has recorded some number of random increments for each honest-at-the-time $\dpmath{j}$ that receives \textsc{Observation} from $Z$, call this sum $s_{tj}$. Then the plaintext that is used by $\cpmath{h}$ in $F$ for counter $t$ is zero if and only if $\sigma_t^F = \sum_j (r_{tj} + c_{tj} + s_{tj})$.
Further, since the processes do not formally diverge until after $\sigma_t$ is determined, the equivalent value that is used in \ideal by $S$ to determine whether to send an observation message to \f{\sys} is:
$\sigma_t^I = \sum_j (r_{tj} + c_{tj})$
since $S$ has no Ghost Execution.
We first observe that if a $\dpmath{j}$ is honest through the input submission phase, $r_{tj} + c_{tj} = 0$ since honest \dps submit $0$. We note that if either of these sums is nonzero, it guarantees that counter $t$ will be nonzero in the final tally in that execution (in \ideal, the observation is sent directly to \f{\sys} by $S$. In $F$, the value is inserted in the output of $\cpmath{h}$ during Rerandomization-Decryption directly).\vspace{-.5mm}

\begin{enumerate}
	\item $\sigma_t^I = \sigma_t^F = 0$. This means no $c_{jt}$ has been incremented, meaning that no honest \dparty has been sent an \textsc{Observation} message, and the value of $t$ is the same in both executions, or some nonzero number of uniformly randomly selected values sum to zero, which is an event with negligible probability $\frac{1}{|G|}$.
	\item $\sigma_t^I \neq 0$, $\sigma_t^F \neq 0$. In this case, $t$ is forced to count in the tallies of both executions, as the observation is directly sent to \f{\sys} in \ideal, and remains nonzero in $F$ since it is multiplied by a sequence of nonzero values.
	\item $\sigma_t^I \neq \sigma_t^F$, with one of these equal to zero. This means $\sum_j s_{tj} \neq 0$.
If $\sigma_t^F = 0$, this means randomly selected nonzero $\sum_j s_{tj}$ is exactly $-\sigma_t^I$, a sum calculated with no knowledge of $\sum_j s_{tj}$, which occurs with negligible probability $\frac{1}{|G|}$. If $\sigma_t^I = 0$, then a then-honest \dparty has received some \textsc{Observation} message, so in $F$ $t$ is included in the tally since $\sigma_t^F \neq 0$, and in \ideal \xspace $t$ is included in the tally because (\textsc{Observation}, $t$) has been sent to some honest \dparty who has forwarded it directly to \f{\sys}.
\end{enumerate}

\textbf{Nonzero Counters from Noise}.
The second portion of $o$ comes from the noise counters. In \ideal, this is a number drawn by an ideal functionality directly from $\texttt{Bin}(\numnoise;\frac{1}{2})$. In $F$, it is determined by the composition of permutations in each round $i$ of the Noise Generation phase, but the ciphertexts being shuffled in $F$ after $\cpmath{h}$ applies its shuffle for each bin of Noise Generation are both encryptions of zero, so all plaintexts for these ciphertexts are encryptions of $0$ and carry no information about the honest permutation $\pi^{ih}$. This means $\pi^{ih}$ for each noise generation round $i$ in $F$ is selected external to the execution during Noise Generation, and no messages in the execution depend on $\pi^{ih}$ until Rerandomization-Decryption, so the $\pi^{ih}$ may be chosen then. For each round one of the two permutations on two elements provides a result of $1$ and the other of $0$.  We draw all $\numnoise$ of them from these randomly, producing exactly the same distribution $\texttt{Bin}(\numnoise;\frac{1}{2})$.

Then we conclude the number $o$ is selected from the same distribution in both executions except with negligible probability, as the sum of a fixed number of nonzero data counters (a number which may differ in both executions but only with negligible probability) and a number which is drawn identically from the same distribution by honest parties in each case.
\end{proof}

\begin{theorem} The protocol $\pi_{\textsc{\sys}}$ $\pi_{\text{\sys}}$ UC-realizes \f{\sys} in the \hybrid model if there is at least one honest \cp and all \cp corruptions are static and the Decisional Diffie Hellman assumption holds for the group $G$.  \end{theorem}

\begin{proof}
We have the sequence of computationally indistinguishable hybrid executions, beginning with \real and followed by the sequences $B_{\textsc{SEQ}}$, $D_{\textsc{SEQ}}$, $N_{\textsc{SEQ}}$, $S_{\textsc{SEQ}}$, with the last element of this last sequence indistinguihable from \ideal. Since the number of transitions is finite, and the distinguishing advantage for any PPT environment between any adjacent two is bounded by a negligible function, we arrive at the statement of the theorem.

\end{proof}

%% file: sections/implementation.tex
\vspace{-6mm}

\section{Implementation and Evaluation}
\label{sec:eval}
We constructed an implementation of \Sys in Go to verify our protocol's correctness and to measure the system's computation and communication overheads. We run experiments over large synthetic datasets and measure our implementation's performance. We describe the implementation details and design choices in Section~\ref{subsec:impl} and then present our performance evaluation in Section~\ref{subsec:evaluation}.

\subsection{Implementation}
\label{subsec:impl}
We built an implementation of \Sys in Go using the Kyber~\cite{kyber} advanced cryptographic library. ElGamal encryption is implemented in the
Edwards 25519 elliptic curve group~\cite{edwards} and the \cps use Neff's verifiable shuffle~\cite{neff-shuffle} to shuffle the ElGamal ciphertexts.

The \cps perform secure broadcast using the optimized ``heartbeat'' version of the Dolev-Strong protocol described in Section~\ref{sec:consensus}. We do not implement
time-outs that explicitly signal the end of a Dolev-Strong round. Therefore, termination is not guaranteed in our implementation \ie honest \cps may wait forever on malicious
\cps that do not send their response. However, this does not affect our results as we analyze the performance of our protocol in an honest setting \ie, all \cps and \dps are considered honest (refer Section~\ref{subsec:evaluation}).

The \cps use the Schnorr signature algorithm~\cite{schnorr-joc1991} over the Edwards 25519 elliptic curve for signing and SHA-256 for computing digests. We use TLS 1.2 for secure point-to-point communication.

We use ``Biffle'' in the Kyber library for shuffling the noise vectors during the noise generation phase. Biffle is a fast binary shuffle
for two ciphertexts based on generic zero-knowledge proofs.

For zero-knowledge proofs, we use Schnorr-type proofs~\cite{schnorr-joc1991, camenisch-delegatable-2017} for the combined re-encryption, re-randomization and
decryption of ElGamal ciphertexts and knowledge of discrete log of the Elgamal public key and blinding factors. Non-interactive versions of
all these proofs are produced using the Fiat-Shamir heuristic~\cite{fiat-shamir-crypto1986}.

A single \dparty program emulates all the \dps in our implementation \ie, each operation (\eg, blind submission, data collection, \etc) for every \dparty is performed
sequentially by a single \dparty program. However, in a real deployment, the \dps would be distributed.

To encourage the use of \Sys by privacy researchers and practitioners, we are releasing \Sys as free, open-source software, available for download at \url{https://github.com/GUSecLab/psc}.

\Paragraph{Optimization} While computing the aggregate, as an optimization the \cps first perform a secure broadcast of the aggregate received from the \dps. If a consensus can be reached among the \cps, then the agreed-upon aggregate is used.  Otherwise,
in the case where the \cps cannot reach a consensus (on the aggregate), they opt for the more communication-intense process of performing a secure broadcast
(using the heartbeat Dolev-Strong protocol) for each individual \dparty blind and plaintext responses.

%% file: sections/evaluation.tex
\subsection{Evaluation}
\label{subsec:evaluation}
Experiments are carried out on 10 core Intel Xeon machines with 32GB to 64GB of RAM running Linux kernel
4.4.0. Our implementation of \sys is currently single-threaded. Although the
computational cost of \sys's noise generation is significant and may be done by the \cps before
the inputs are received, we parallelize it so that it can be done on multi-core machines after
the inputs are received. We use ``parallel for'' from the Golang par package~\cite{paar} for this
parallel noise shuffling.

We instantiate all \cps and \dps on our 10 core servers. Google Protocol Buffers~\cite{varda2008protocol} is
used for serializing messages, which are communicated over TLS connections between \sys's parties.
We use Go's default crypto/tls package to implement TLS.

\begin{table}[t!]
\begin{minipage}[t]{.43\linewidth}
  \centering
  \footnotesize
  \caption{Default values and descriptions for system \mbox{parameters}.}
  \label{tbl:default}
  \begin{tabular}{ccc}
    {\bf Param.} & {\bf Description} & {\bf Default} \\ \toprule
    $\numcounters$ & \# of counters & 300,000 \\
    $\numcompute$ & \# of \computationparties & 5 \\
    $\numinput$ & number of \dataparties & 30 \\
    $\epsilon$ & privacy parameter & 0.3 \\
    $\delta$ & privacy parameter & $10^{-12}$ \\
  \end{tabular}
\end{minipage}
\hfill
\begin{minipage}[t]{.54\linewidth}
  \centering
  \footnotesize
  \caption{Actual, noisy aggregates, and standard deviation for various values of  $\epsilon$.}
  \label{tbl:epsilon}
  \begin{tabular}{cp{.75in}p{.75in}p{.75in}}
    {\bf $\epsilon$} & {\bf Actual \mbox{Agg.}} & {\bf Noisy \mbox{Agg.}} & {\bf Standard \mbox{Dev.}}\\ \toprule
    0.15 & 8927 & 8752 & 141.92 \\
    0.30 & 8927 & 8811 & 70.96 \\
    0.45 & 8927 & 8887 & 47.31 \\
    0.60 & 8927 & 8873 & 35.48 \\
    0.75 & 8927 & 8917 & 28.39 \\
  \end{tabular}
\end{minipage}
\vspace{-3mm}
\end{table}

\paragraph{Query and dataset.} We evaluate \Sys by considering the query: what is the number of unique IP connections as observed by the nodes in an anonymity network?
Although our implementation tolerates malicious \cps and dishonest majority of \dps, we analyze the performance of our protocol in an honest setting (\ie, all \cps and \dps are considered honest) for experimental purposes.

Rather than store $2^{32}$ (or, for IPv6, $2^{128}$) counters, we assume $b$ counters (where $b \ll 2^{32} $) and map IP addresses to
a $(\lg b)$-bit digest by considering the first $\lg b$ bits of a hash function; this results in some loss of accuracy due to collisions.  For each experiment, we
chose an integer uniformly at random from the range $[0, 30000]$. Then for each \dparty, we choose from its counters a random subset of that size to set
to {\tt 1}; the remaining counters are set to {\tt 0}.

We note that the performance of \Sys is independent of the number of unique IPs. Therefore, in our performance evaluation, we are interested in how the number
of counters $\numcounters$ affects the operation of \Sys, rather than the number of unique IPs.

\paragraph{Experimental setup.} The default values for the number of bins $\numcounters$, the number of \cps $\numcompute$, the number of \dps $\numinput$, $\epsilon$, and $\delta$ are listed in Table~\ref{tbl:default}.

We determine these default values by considering which values would be
appropriate for an anonymity network such as Tor. The Tor Project
reports approximately $2.8$
million user connections (using simple statistical techniques to roughly estimate the number of Tor users) and over $3,000$ guard
nodes~\cite{tormetrics}. Assuming that $1\%$ of Tor guards deploy \Sys, we have $\numinput = 30$. Also, given this level  of \Sys deployment,
we would expect a Tor guard to see approximately $30,000$ unique IPs (\ie, $1\%$ of ${\sim}3$ million user connections).

To measure the aggregate with high accuracy, we limit hash-table collisions to at most a fraction
$f$ of inputs in expectation by using a hash table of size $1/f$ times the number of inputs.
Therefore, for $f=10\%$, we set (unless otherwise specified) $\numcounters = 300,000$ in all our experiments. We set $\epsilon = 0.3$ as this is currently recommended for safe measurements in anonymity networks such as Tor~\cite{privcount}. To limit to $10^{-6}$ the chance of a privacy
``failure'' affecting any of $10^6$ users, we set $\delta$ to $10^{-12}$~\cite{dwork2014algorithmic}. We set these default values as system-wide parameters, unless
otherwise indicated.

\paragraph{Accuracy.} The trade-off between accuracy and privacy is governed by the choice of
$\epsilon$ and $\delta$. We vary $\epsilon$ from $0.15$ to $0.75$, keeping the number of bins at
$\numcounters = 300,000$. We found that values below $0.15$ produced too much noise and offered
low utility. Values of $\epsilon$ greater than $0.75$ would not provide a reasonable level of privacy.

The actual and the noisy aggregate values along with the standard deviation for the noisy aggregates
(the noise follows a binomial distribution) for different values of $\epsilon$ is shown in
Table~\ref{tbl:epsilon}. We observe that the standard deviation of the noisy value is at most
$141.92$ even for smaller values of $\epsilon$, such as $0.15$. Therefore, as expected the noisy
aggregates are very close to the actual aggregates. In summary, \Sys gives highly accurate results for
all tested privacy levels.



\Paragraph{Communication cost}
To be practical, a statistics gathering system  should impose a low communication overhead for the
\dps, which can have limited bandwidth.  However, we envision the \cps to be well-resourced
dedicated servers that can sustain at least moderate  communication costs.

\Sys incurs communication overhead by transmitting ElGamal encrypted blinds, a zero-knowledge proof
(for the knowledge of discrete log of the blinds), and {\em masked} plaintext counters between the \dps
and \cps and the ElGamal encrypted counters among the \cps.

We explore \Sys's communication costs by varying the number of bins $\numcounters$, the number of \cps $\numcompute$,
the number of \dps $\numinput$, $\epsilon$, and $\delta$. The average communication cost for the \cps and \dps are plotted in Figure~\ref{fig:comm}
and Figure~\ref{fig:eps-del-psc}. We omit the error bars as \Sys has a deterministic communication cost---there is no variance in
the communication cost incurred among the \cps (and similarly, among the \dps).



We first consider how the number of bins influences the communication cost. We run \Sys, varying $\numcounters$ from $100,000$
to $500,000$, and plot the results in Figure~\ref{fig:comm} {\it (left)}. The values of the bins (\ie, {\tt 0} or {\tt 1}) do not affect the outbound
communication cost of the \dps, as the \dps invariably transmit an encrypted blind, a zero-knowledge proof and a plaintext value for
either {\tt 0} or {\tt 1}. For up to $300,000$ bins, the outbound communication cost for each \dparty is fairly modest. For example, if
\Sys is run once an hour, then the outbound communication cost is approximately $252$~MB/hr ($70$~KBps).
We also find that the inbound communication cost for each \dparty is constant (${\sim}1.8$~KB), as the \dps only receive the signed
ElGamal joint public key from every \cp, irrespective of the number of bins.

The communication costs are more significant for the \cps, which we envision are dedicated machines for \Sys. With $300,000$ bins,
each \cp requires a bandwidth of $962.6$ MB for sending and $2.5$ GB for receiving (less than $700$ KBps if executed once per
hour). We observe that the difference between the \cp outbound and inbound communication cost is large and roughly equal to the
sum of the \dparty outbound communication cost per \cp (\ie, ${252/5} \approx 50.4$ MB) across all $30$ \dps. This is because each
\cp sends the signed ElGamal joint public key (\ie, a single group
element) to every \dparty, whereas each \cp receives $b$ ElGamal encrypted
blinds, zero-knowledge proofs, and plaintext counters from every \dparty.

We next consider how $\numcompute$, the number of \cps, affects the
communication cost. We vary $\numcompute$ from $3$ to~$7$ and plot the
results in Figure~\ref{fig:comm} {\it (center)}. The outbound communication cost for the \dps increases at a much slower pace
when varying $\numcompute$ (as opposed to $\numcounters$), as the \dps invariably send $b$ (= $300,000$) ElGamal encrypted
blinds, zero-knowledge proofs and plaintext counters to each \cp. Therefore, even up to seven \cps, the outbound communication
cost for each \dparty is fairly modest --- approximately $1.4$~GB (or $98$~KBps, if run every hour). We note that the inbound
communication cost per \dparty increases by a constant factor
(${\sim}0.93$~KB) as the number of \cps increases. This increase is
due to
each \dparty receiving a copy of the signed ElGamal joint public key from every \cp.

\begin{figure}[!tbp]
	\centering
	\begin{minipage}[b]{0.8\textwidth}
		\centering
		\includegraphics[width=\linewidth, height=0.4\textwidth]{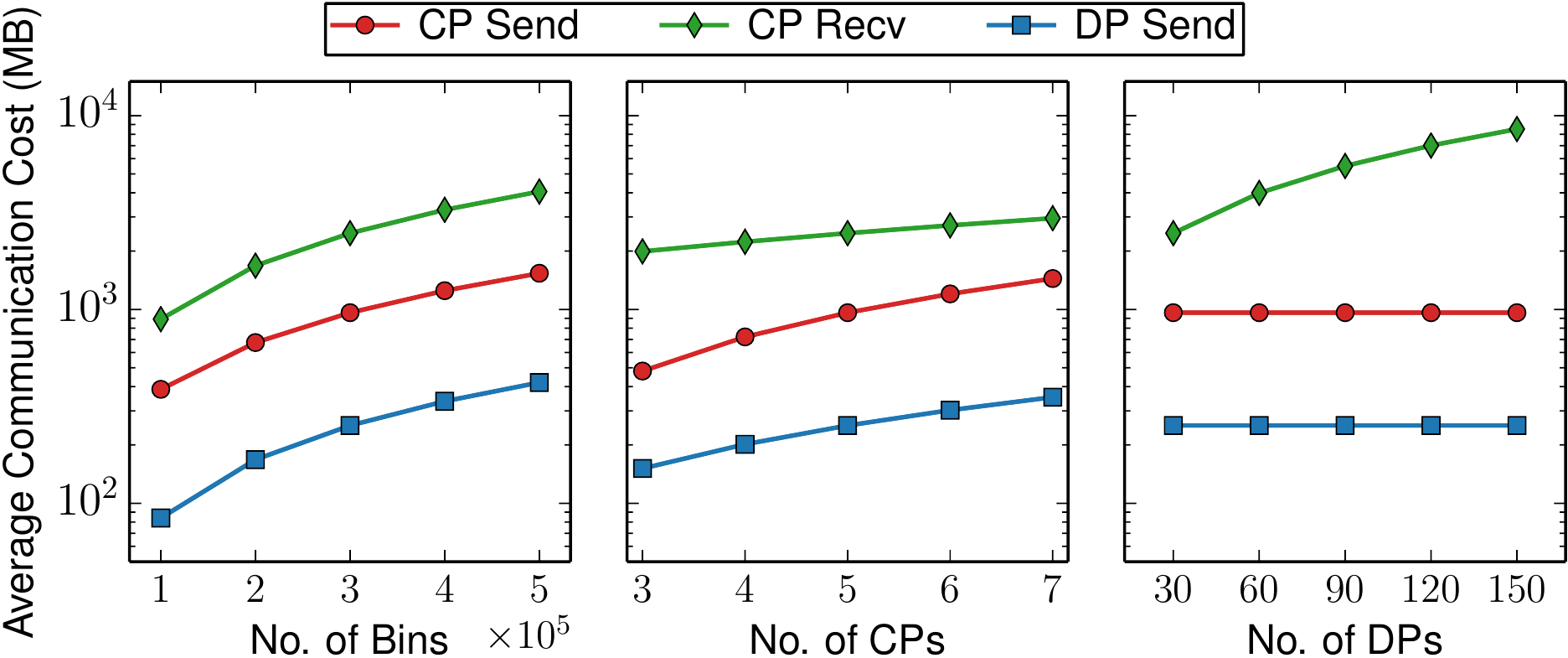}
	\end{minipage}
	\vspace{-7pt}
	\caption{The communication cost incurred by the \cps and \dps varying the number of bins ({\em left}), the number of \cps ({\em center})
		and the number of \dps ({\em right}).}
	\label{fig:comm}
	\vspace{-0.1cm}
\end{figure}

\begin{figure}[!tbp]
	\centering
	\begin{minipage}[b]{0.45\textwidth}
		\centering
		\includegraphics[width=0.95\linewidth]{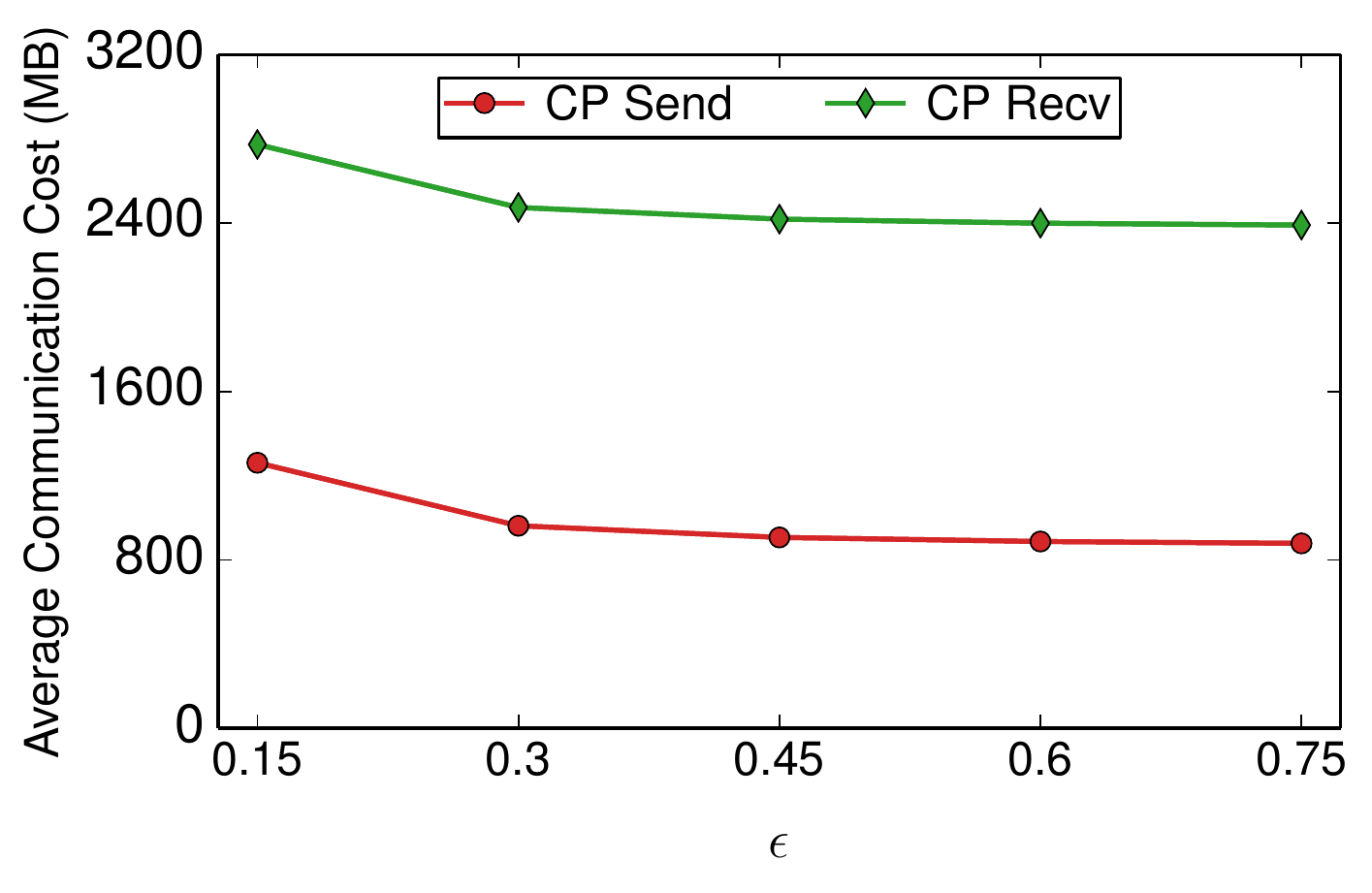}
	\end{minipage}
	\hfill
	\begin{minipage}[b]{0.45\textwidth}
		\centering
		\includegraphics[width=0.95\linewidth]{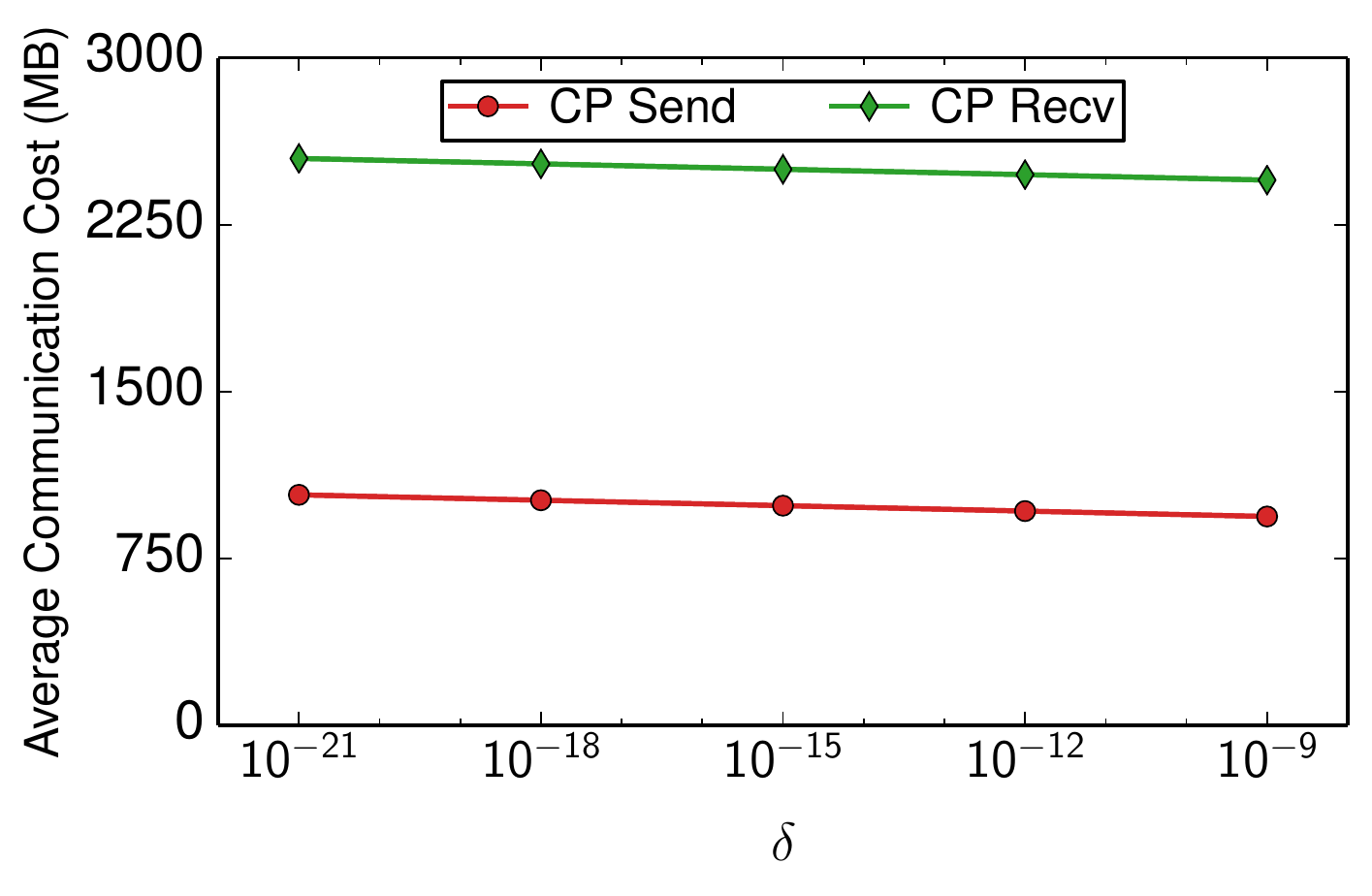}
	\end{minipage}
	\vspace{-7pt}
	\caption{The communication cost incurred by the \cps varying $\epsilon$ ({\em left}) and $\delta$ ({\em right}).}
	\label{fig:eps-del-psc}
\end{figure}
The inbound and outbound communication costs for each \cp increase {\em almost} linearly with the number of \cps. Recall from Section~\ref{sec:consensus} that each \cp
participates in $O(\numcompute)$ broadcasts, $O(\numcompute^2)$ {\em echoes}, and $O(\numcompute^3)$ {\em heartbeats} in the optimized heartbeat
version of the Dolev-Strong protocol. Although the total communication cost of the \cps is $O(\numcompute^3)$ (\ie cubic in $\numcompute$),
the communication cost for the echoes ($32$~bytes) and heartbeats (${\sim}100$~bytes) is significantly less than that for the broadcasts (which involve
transmitting zero-knowledge proofs and ElGamal ciphertexts for all counters and noise). Therefore, the communication cost per \cp roughly increases by a constant
factor (${\sim}240$~MB) as we increase $\numcompute$.

%

We rerun \Sys with different numbers of \dps. Figure~\ref{fig:comm}
({\it right}) shows that varying the number of \dps has no effect on the
average communication cost for the \dps and the average outbound communication cost for the \cps. This is because the \dps invariably
send $b$ ElGamal encrypted blinds, zero-knowledge proofs and plaintext counters to each \cp and receive a copy of the signed joint public key
from each \cp. Likewise, the number of ElGamal ciphertexts and proofs transmitted by the \cps also remains the same across these
experiments. However, there is a large linear increase in the inbound communication cost of the \cps as each \cp receives $b$ blinds,
zero-knowledge proofs and plaintext counters from each \dparty. The variation in the inbound communication cost per
\cp is approximately $1.5$~GB.

Next, we run \Sys with different values of $\epsilon$ to determine how the choice of privacy parameter affects the communication costs.
Figure~\ref{fig:eps-del-psc} {\it (left)} shows that the average communication costs for the \cps decrease when $\epsilon$ is increased. The communication
costs for the \cps are reasonable even for a low value of $\epsilon = 0.15$. On average, each \cp requires a bandwidth of at most
$1.3$~GB for sending and $2.8$~GB for receiving (less than $800$~KBps, if performed once an hour).

Lastly, we consider how the choice of $\delta$ affects the
communication cost. Figure~\ref{fig:eps-del-psc} {\it (right)} shows that the average communication
costs for the \cps decrease as $\delta$ increases. The communication costs for the \cps are reasonable even for a low value of $\delta = 10^{-21}$.
On average, each \cp requires at most $1$~GB for sending and $2.5$~GB for receiving (less than $700$~KBps, if performed once an hour).

In summary, we find that \Sys incurs reasonable communication overhead: the costs to \dps are moderate, and, while slightly higher for \cps, they remain practical.

\Paragraph{Overall runtime}
We explore the overall running time (including the time required for network communication) of \Sys by varying the number of bins
$\numcounters$ and the number of \cps $\numcompute$. The average overall running as a function of $\numcounters$ and
$\numcompute$ is plotted in Figure~\ref{fig:time-psc}.

\begin{figure}[!tbp]
	\centering
	\begin{minipage}[b]{0.45\textwidth}
		\centering
		\includegraphics[width=0.95\linewidth]{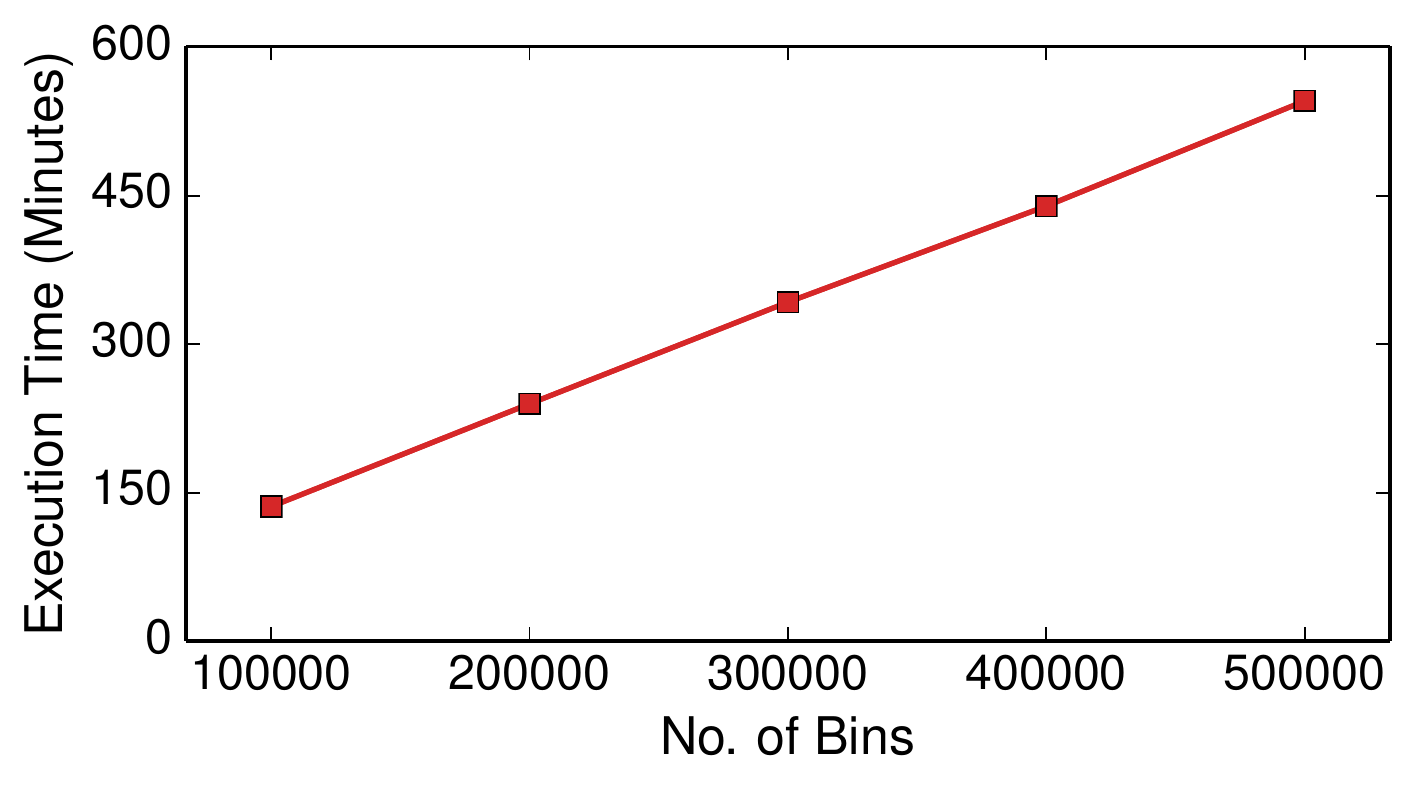}
	\end{minipage}
	\hfill
	\begin{minipage}[b]{0.45\textwidth}
		\centering
		\includegraphics[width=0.95\linewidth]{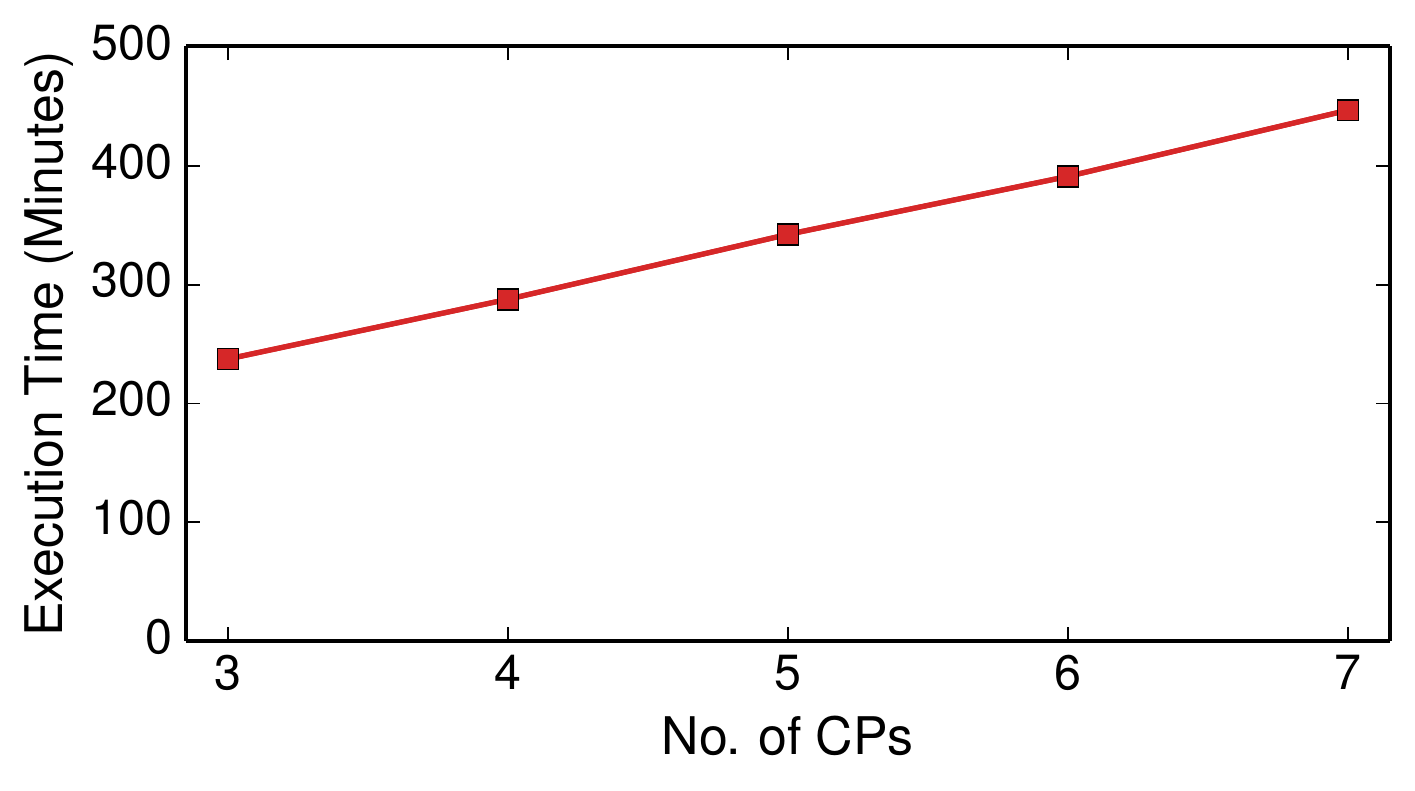}
	\end{minipage}
	\vspace{-7pt}
	\caption{The overall execution time as a function of the number of bins ({\em left}) and the number of \cps ({\em right}).}
	\label{fig:time-psc}
\end{figure}

\begin{figure*}[t!]
	\centering
	\begin{minipage}[t]{0.9\linewidth}
		\centering
		\includegraphics[width=0.65\linewidth, height=0.25\textwidth]{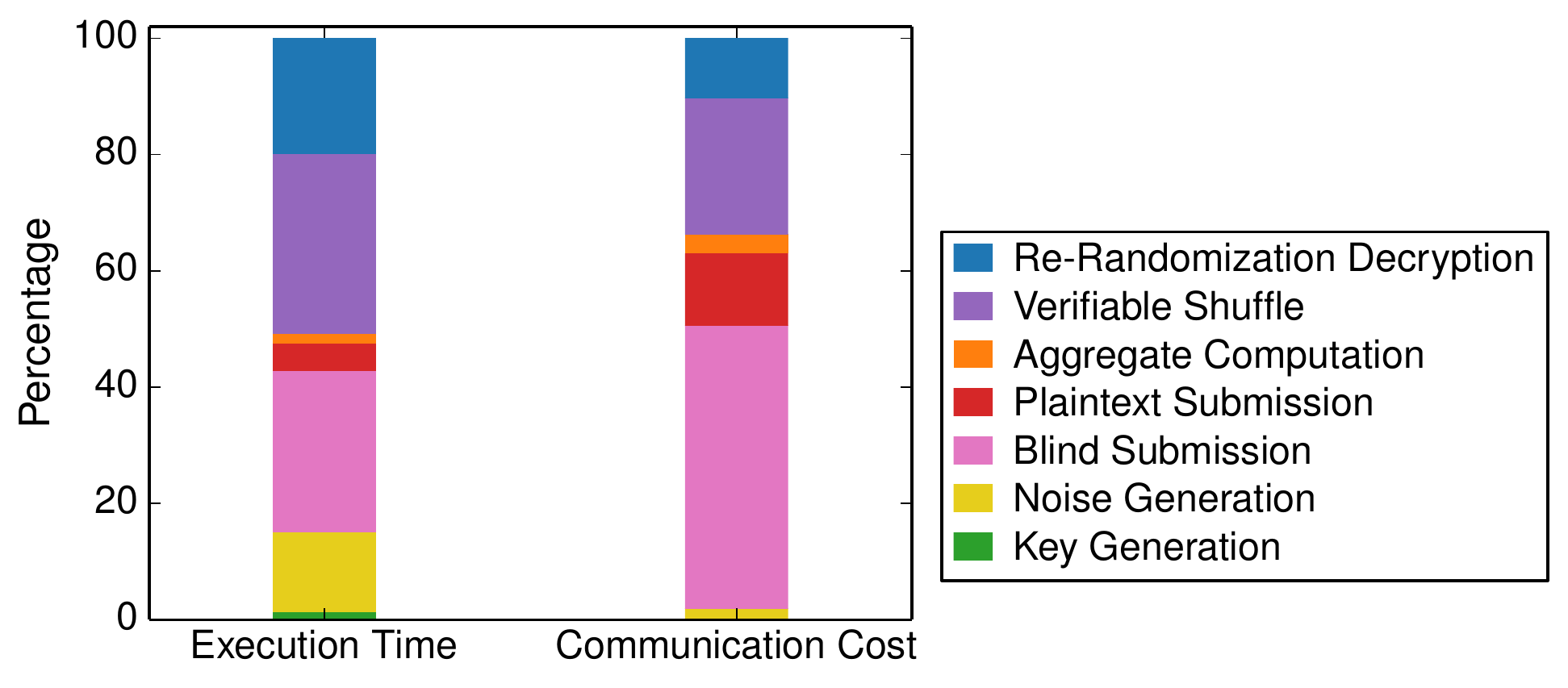}
	\end{minipage}
	\vspace{-7pt}
	\caption{The overall execution time (\em right) and the
          communication cost (\em left) incurred by the \cps for
          different operations, varying the number of
          bins.}
	\label{fig:time-benchmark}
\end{figure*}

We first consider how the number of bins $\numcounters$ affects the computation cost (note that more data must be communicated
using the optimized heartbeat DS protocol as $\numcounters$ increases). The overall runtime is moderate. It takes approximately
$9$~hr $6$~min even for an experiment with $500,000$ bins.

We next consider how the number of \cps $\numcompute$ affects the computation cost. The computation cost increases with the
\cps at a slower pace than with the bins. Even up to seven \cps, the average computation cost for each \cp is fairly modest --- approximately
$7$~hr $27$~min.

\Paragraph{Microbenchmarks} To better understand the computational and communication overhead of \Sys, we measure the execution time
(including the network latency) and communication cost for different operations (using the default values for all system parameters). The results
are plotted in Figure~\ref{fig:time-benchmark}. We observe that the time taken for re-randomization decryption, blind submission, and
verifiable shuffles account for $19.8\%$, $27.8\%$, and $31\%$ of the
total runtime of the protocol respectively.  We note that the time taken for blind submission
is the total time taken across all $\numinput$ ($= 30$) \dps (recall that our implementation is currently single-threaded with a single \dparty program
emulating all the \dps). Therefore, as expected the verifiable shuffle followed by the re-randomization decryption are the most time-consuming operations.

We find that the communication cost for re-randomization decryption, plaintext submission, verifiable shuffles, and blind submission account for
$10.3\%$, $12.4\%$, $23.4\%$, and $48.7\%$ of the overall communication cost of the protocol. Here again, the communication cost for the blind and
plaintext submission is the total cost across all $\numinput$ ($= 30$)
\dps. Therefore, as expected, the verifiable shuffle and the re-randomization decryption are the
most communication-intense operations.

%% file: sections/intersection.tex
\section{Private Set-Intersection Cardinality} \label{sec:intersection}

We outline how our protocol can be modified to calculate a private
set-intersection rather than a private set-union. To turn the homomorphic
operation on ciphertexts from performing ``or'' to performing ''and'', we invert
the logical bit representation by encoding a logical $0$ with any nonzero value
and encoding a logical $1$ with a zero. However, under this change of
representation it would not be possible to record an observation obliviously
(i.e. flip a logical $0$ to a $1$) without further protocol changes, since to
record an observation each \dparty would need to know the original blind it
submitted in order to provide that blind's negation. A \dparty cannot store in
plaintext both the blind and the current counter value, or it would reveal its
counter value to the adversary upon adaptive corruption. We can solve this
problem by encrypting the counter value, but simply submitting the encrypted
counter would enable a related-ciphertext attack. Moreover, the same type of
proof arguments would not work as the simulator would not be able to extract the
adversary's counter values during input submission. We solve these problems by
adding a non-interactive Zero Knowledge Proof of Knowledge for each encrypted
counter, which enables extraction during the simulation proof. We note, however,
that while such ZKPs exist in the UC model, their constructions either require
stronger setup assumptions than a CRS~\cite{lindellsigmaprotocols} or are not
structured similarly to the many-verifier interactive $\Sigma$-protocols we use
for the protocol for set-union~\cite{UC-NIZKP}. We give some details of this
construction and then discuss resulting changes in the protocol's security and
efficiency.

\textbf{Construction}. To perform set-intersection cardinality, we modify the
protocol (Sec.~\ref{sec:protocol}) as follows:

\begin{enumerate}
\item Blinds are submitted as before: for a counter $t$, $\dpmath{j}$
submits $E(b_{tj})$ and the corresponding proof of knowledge of the plaintext.
\dps save two values: (1) $-b_{tj}$ as a ``zero'' value; and (2) the pair
($E_y(r)$, $P(E_y(r))$) and as a ``submission'' value, where $E_y(r)$ is an
encryption of a random value $r$ and $P(E_y(r))$ is a non-interactive
zero-knowledge proof of knowledge of $r$. \item To record an observation for
counter $t$, $\dpmath{j}$ replaces its submission value with $(E(-b_{tj})$,
$P(E(-b_{tj})))$ and erases from memory information about this operation. \item
When the collection period is complete, each \dparty submits for every counter
its submission value for that counter. The protocol proceeds identically as
before, but in the final result we count the number of plaintext counters that
are equal to $0$.
\end{enumerate}

\textbf{Security}. Upon corrupting a $\dpmath{j}$, an adaptive adversary now
learns the original blind values for all counters. This
means adaptive adversaries are able not only to set the submission of corrupt
\dps to $1$, as before, but may also force this value to be $0$ as well
regardless of previous observations by that adaptively corrupted \dparty. For
set-union, adaptive adversaries could not revert previously recorded
observations. We note however, that in both cases when the adversary submits
malicious values on behalf of adaptively corrupted \dps, it knows nothing about
the previous value for each counter. Thus the intersection protocol preserves
the most important adaptive security property: that no private data is revealed
upon an adaptive corruption.

\textbf{Efficiency}. With respect to efficiency, we expect the additional computation and communication costs to be substantial but small enough that the protocol remains usable in practice. More precisely, the adjustments require each \dparty to perform a public-key encryption and non-interactive zero-knowledge proof to record an observation, whereas for set union, they simply randomize a plaintext counter. For communication cost, using the Fiat-Shamir heuristic and assuming group elements and plaintext elements of $\mathbb{Z}_q$ are roughly equal size, the total cost during input submission should increase from one plaintext counter (for each counter held by each \dparty) to $4$ (two each for the ciphertext and proof) with other phases of the protocol unchanged.

%% file: sections/conclusion.tex
\section{Conclusion and Future Work}
We present the \Sys protocols for securely computing the cardinality of set union and set
intersection across multiple parties. \Sys is secure against an active adversary controlling all but
one honest party, provides adaptive security for the Data Parties, makes the adversary accountable
for disrupting the protocol, and produces differentially-private outputs. Our implementation of PSC is available for download at \url{https://github.com/GUSecLab/psc}. Future work includes
developing the sketch-based approach~\cite{diffPrivateSketches} by modifying generic MPC protocols to provide all our security properties, applying efficiency improvements
in zero-knowledge proofs~\cite{shuffle-bayer-eurocrypt2012,bunz2018bulletproofs}, and investigating batching broadcasts~\cite{mpc-identifiable-abort} to improve performance while mainting our termination guarantees.

%% file: sections/acks.tex
\begin{acks}
This work has been supported by the National Science Foundation under grant \mbox{CNS-1718498}, the
Canada Research Chairs program, the Office of Naval Research, and the
Callahan Family Professor endowment. The  views expressed in this
work are those of the authors and do not necessarily reflect those of
the funding agencies or organizations.
\end{acks}

%% file: sections/appendix.tex
\appendix
\section{Commitments and Coins}
\label{sec:appendix}
For reference, in the section we repeat the one-to-many commitment functionality given in~\cite{canetti2002universally} (Figure~\ref{fig:psc-mcom}), present an accountable coin-flipping functionality (Figure~\ref{fig:psc-coin-func}), and give a protocol that UC-realizes the coin flipping functionality and the corresponding proof (Figure~\ref{fig:psc-coin-protocol}).

\begin{figure}[h]
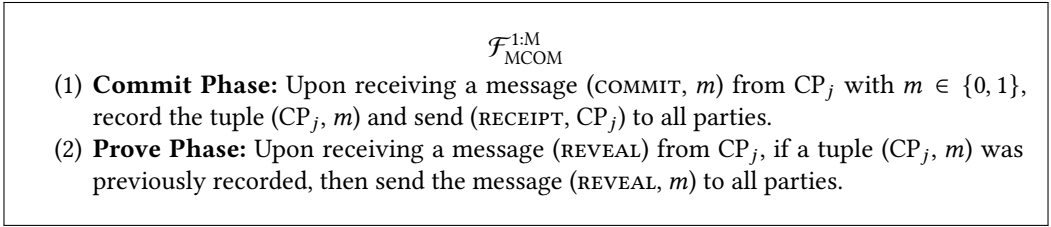

\begin{framed}

	\center{\f[1:M]{MCOM}}\\
\begin{enumerate}
	\item \textbf{Commit Phase:} Upon receiving a message (\textsc{commit}, $m$) from $\cpmath{j}$ with $m \in \{0,1\}$, record the tuple ($\cpmath{j}$, $m$) and send (\textsc{receipt, $\cpmath{j}$}) to all parties.
	\item \textbf{Prove Phase:} Upon receiving a message (\textsc{reveal}) from $\cpmath{j}$, if a tuple ($\cpmath{j}$, $m$) was previously recorded, then send the message (\textsc{reveal}, $m$) to all parties.
\end{enumerate}
\end{framed}
\caption{\f[1:M]{MCOM}, the ideal functionality for one-to-many commitments.}
\label{fig:psc-mcom}
\end{figure}

\begin{figure}[h]
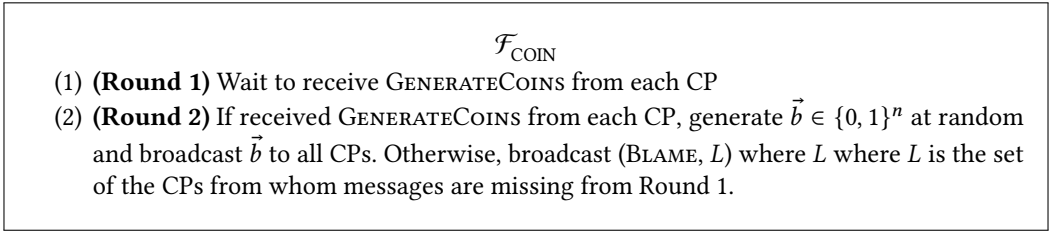

\begin{framed}

\center{\f{COIN}}\\
\begin{enumerate}
	\item \textbf{(Round 1)} Wait to receive \textsc{GenerateCoins} from each \cp
	\item \textbf{(Round 2)} If received \textsc{GenerateCoins} from each \cp, generate $\vec{b} \in \{0,1\}^n$ at random and broadcast $\vec{b}$ to all CPs. Otherwise, broadcast (\textsc{Blame}, $L$) where $L$ where $L$ is the set of the \cps from whom messages are missing from Round 1.
\end{enumerate}
\end{framed}
\caption{\f{COIN}, the ideal coin flipping functionality.}
\label{fig:psc-coin-func}
\end{figure}

\begin{figure}[h]
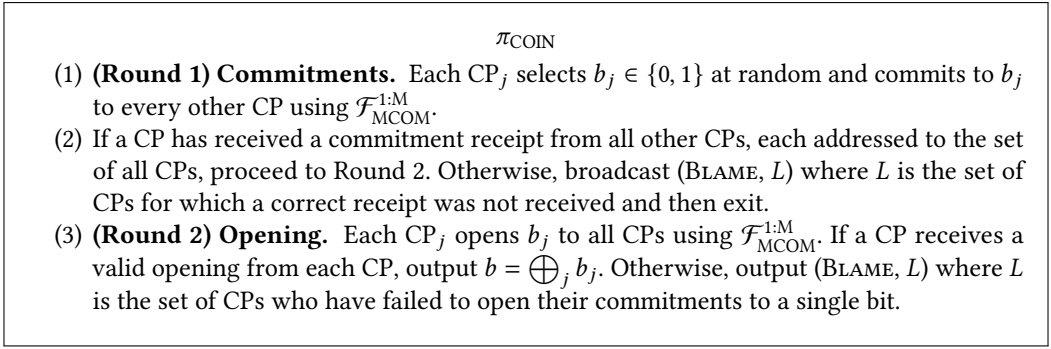

\begin{framed}

	\center{$\pi_{\textsc{COIN}}$}\\
\begin{enumerate}
	\item \textbf{(Round 1) Commitments. }  Each $\cpmath{j}$ selects $b_j \in \{0,1\}$ at random and commits to $b_j$ to every other \cp using \f[1:M]{MCOM}.
	\item If a \cp has received a commitment receipt from all other \cps, each addressed to the set of all \cps, proceed to Round 2. Otherwise, broadcast (\textsc{Blame}, $L$) where $L$ is the set of \cps for which a correct receipt was not received and then exit.
	\item \textbf{(Round 2) Opening. } Each $\cpmath{j}$ opens $b_j$ to all \cps using \f[1:M]{MCOM}. If a \cp receives a valid opening from each \cp, output $b = \bigoplus_j b_j$. Otherwise, output (\textsc{Blame}, $L$) where $L$ is the set of \cps who have failed to open their commitments to a single bit.
\end{enumerate}
\end{framed}
\caption{$\pi_{\textsc{COIN}}$, an accountable coin-flipping protocol}
\label{fig:psc-coin-protocol}
\end{figure}
\begin{theorem} The protocol $\pi_{\textsc{COIN}}$ UC-realizes the \f{COIN} functionality in the static-corruption, synchronous, (\f{BC}, \f[1:M]{MCOM})-hybrid model if there is at least one honest \cp. \end{theorem}
\begin{proof}
	We define a simulator $S$ as follows: $S$ runs a copy of the adversary $A$, and simulates the honest parties and ideal functionalities honestly. If a malicious party commits during the commit phase and opens it to an element of $\{0,1\}$ in the opening phase, S sends \textsc{GenerateCoin} to \f{COIN} on behalf of that malicious party. We observe that every dishonest party $P$ either completes a valid opening in the second phase of the protocol, in which case \f{COIN} has received \textsc{GenerateCoin} from $P$ (through $S$), or $P$ has failed to produce a valid opening, in which it is blamed by all honest (simulated) \cps. Finally, in the case that all parties open commitments and \f{COIN} outputs a uniformly random bit, this bit is identically distributed to $\oplus_j b_j$ in the simulation, since $b_j$ is random for honest $j$, and dishonest \cps must commit to their bits without any information about the honest parties' selections.
\end{proof}
We will need a vector of $n$ coins, so we write \f{N-COIN} as $n$ independent copies of \f{COIN}, which can be joined into a single protocol using the JUC theorem from~\cite{UC-joint}. Similarly, we will require commitments to strings of bits, so we write \f[1:M]{S-MCOM} to indicate the protocol where \f[1:M]{MCOM} is invoked multiple times to produce a commitment to a bit string rather than a single bit.

%% file: paper.bbl

\begin{thebibliography}{70}


\ifx \showCODEN    \undefined \def \showCODEN     #1{\unskip}     \fi
\ifx \showDOI      \undefined \def \showDOI       #1{#1}\fi
\ifx \showISBNx    \undefined \def \showISBNx     #1{\unskip}     \fi
\ifx \showISBNxiii \undefined \def \showISBNxiii  #1{\unskip}     \fi
\ifx \showISSN     \undefined \def \showISSN      #1{\unskip}     \fi
\ifx \showLCCN     \undefined \def \showLCCN      #1{\unskip}     \fi
\ifx \shownote     \undefined \def \shownote      #1{#1}          \fi
\ifx \showarticletitle \undefined \def \showarticletitle #1{#1}   \fi
\ifx \showURL      \undefined \def \showURL       {\relax}        \fi
\providecommand\bibfield[2]{#2}
\providecommand\bibinfo[2]{#2}
\providecommand\natexlab[1]{#1}
\providecommand\showeprint[2][]{arXiv:#2}

\bibitem[\protect\citeauthoryear{Asharov and Lindell}{Asharov and
  Lindell}{2017}]%
        {uc-bgw-asharov}
\bibfield{author}{\bibinfo{person}{Gilad Asharov} {and} \bibinfo{person}{Yehuda
  Lindell}.} \bibinfo{year}{2017}\natexlab{}.
\newblock \showarticletitle{A full proof of the BGW protocol for perfectly
  secure multiparty computation}.
\newblock \bibinfo{journal}{\emph{Journal of Cryptology}} \bibinfo{volume}{30},
  \bibinfo{number}{1} (\bibinfo{year}{2017}).
\newblock


\bibitem[\protect\citeauthoryear{Baum, Orsini, and Scholl}{Baum
  et~al\mbox{.}}{2016}]%
        {mpc-identifiable-abort}
\bibfield{author}{\bibinfo{person}{Carsten Baum}, \bibinfo{person}{Emmanuela
  Orsini}, {and} \bibinfo{person}{Peter Scholl}.}
  \bibinfo{year}{2016}\natexlab{}.
\newblock \showarticletitle{Efficient secure multiparty computation with
  identifiable abort}. In \bibinfo{booktitle}{\emph{Theory of Cryptography
  Conference (TCC)}}.
\newblock


\bibitem[\protect\citeauthoryear{Baum, Orsini, Scholl, and Soria-Vazquez}{Baum
  et~al\mbox{.}}{2020}]%
        {constant-ia-mpc-public}
\bibfield{author}{\bibinfo{person}{Carsten Baum}, \bibinfo{person}{Emmanuela
  Orsini}, \bibinfo{person}{Peter Scholl}, {and} \bibinfo{person}{Eduardo
  Soria-Vazquez}.} \bibinfo{year}{2020}\natexlab{}.
\newblock \showarticletitle{Efficient Constant-Round MPC with Identifiable
  Abort and Public Verifiability}. In \bibinfo{booktitle}{\emph{Annual
  International Cryptology Conference (Crypto)}}.
\newblock


\bibitem[\protect\citeauthoryear{Bayer and Groth}{Bayer and Groth}{2012}]%
        {shuffle-bayer-eurocrypt2012}
\bibfield{author}{\bibinfo{person}{Stephanie Bayer} {and} \bibinfo{person}{Jens
  Groth}.} \bibinfo{year}{2012}\natexlab{}.
\newblock \showarticletitle{Efficient Zero-knowledge Argument for Correctness
  of a Shuffle}. In \bibinfo{booktitle}{\emph{Annual International Conference
  on the Theory and Applications of Cryptographic Techniques (Eurocrypt)}}.
\newblock


\bibitem[\protect\citeauthoryear{Bellare and Rogaway}{Bellare and
  Rogaway}{1993}]%
        {rom}
\bibfield{author}{\bibinfo{person}{Mihir Bellare} {and}
  \bibinfo{person}{Phillip Rogaway}.} \bibinfo{year}{1993}\natexlab{}.
\newblock \showarticletitle{Random oracles are practical: A paradigm for
  designing efficient protocols}. In \bibinfo{booktitle}{\emph{ACM Conference
  on Computer and Communications Security (CCS)}}.
\newblock


\bibitem[\protect\citeauthoryear{Benaloh}{Benaloh}{1994}]%
        {benaloh1994dense}
\bibfield{author}{\bibinfo{person}{Josh Benaloh}.}
  \bibinfo{year}{1994}\natexlab{}.
\newblock \showarticletitle{Dense probabilistic encryption}. In
  \bibinfo{booktitle}{\emph{Workshop on selected areas of cryptography (SAC)}}.
\newblock


\bibitem[\protect\citeauthoryear{Blum, Ligett, and Roth}{Blum
  et~al\mbox{.}}{2013}]%
        {blum2013learning}
\bibfield{author}{\bibinfo{person}{Avrim Blum}, \bibinfo{person}{Katrina
  Ligett}, {and} \bibinfo{person}{Aaron Roth}.}
  \bibinfo{year}{2013}\natexlab{}.
\newblock \showarticletitle{A learning theory approach to noninteractive
  database privacy}.
\newblock \bibinfo{journal}{\emph{Journal of the ACM (JACM)}}
  \bibinfo{volume}{60}, \bibinfo{number}{2} (\bibinfo{year}{2013}).
\newblock


\bibitem[\protect\citeauthoryear{Brandt}{Brandt}{2005}]%
        {brandt2005efficient}
\bibfield{author}{\bibinfo{person}{Felix Brandt}.}
  \bibinfo{year}{2005}\natexlab{}.
\newblock \showarticletitle{Efficient cryptographic protocol design based on
  distributed El Gamal encryption}. In \bibinfo{booktitle}{\emph{International
  Conference on Information Security and Cryptology (ICISC)}}.
\newblock


\bibitem[\protect\citeauthoryear{B{\"u}nz, Bootle, Boneh, Poelstra, Wuille, and
  Maxwell}{B{\"u}nz et~al\mbox{.}}{2018}]%
        {bunz2018bulletproofs}
\bibfield{author}{\bibinfo{person}{Benedikt B{\"u}nz},
  \bibinfo{person}{Jonathan Bootle}, \bibinfo{person}{Dan Boneh},
  \bibinfo{person}{Andrew Poelstra}, \bibinfo{person}{Pieter Wuille}, {and}
  \bibinfo{person}{Greg Maxwell}.} \bibinfo{year}{2018}\natexlab{}.
\newblock \showarticletitle{Bulletproofs: Short proofs for confidential
  transactions and more}. In \bibinfo{booktitle}{\emph{IEEE Symposium on
  Security and Privacy (S\&P)}}.
\newblock


\bibitem[\protect\citeauthoryear{Camenisch, Drijvers, and
  Dubovitskaya}{Camenisch et~al\mbox{.}}{2017}]%
        {camenisch-delegatable-2017}
\bibfield{author}{\bibinfo{person}{Jan Camenisch}, \bibinfo{person}{Manu
  Drijvers}, {and} \bibinfo{person}{Maria Dubovitskaya}.}
  \bibinfo{year}{2017}\natexlab{}.
\newblock \showarticletitle{{Practical UC-secure delegatable credentials with
  attributes and their application to blockchain}}. In
  \bibinfo{booktitle}{\emph{ACM Conference on Computer and Communications
  Security (CCS)}}.
\newblock


\bibitem[\protect\citeauthoryear{Canetti}{Canetti}{2001}]%
        {uc-focs2001}
\bibfield{author}{\bibinfo{person}{Ran Canetti}.}
  \bibinfo{year}{2001}\natexlab{}.
\newblock \showarticletitle{Universally Composable Security: A New Paradigm for
  Cryptographic Protocols}. In \bibinfo{booktitle}{\emph{Foundations of
  Computer Science (FOCS)}}.
\newblock


\bibitem[\protect\citeauthoryear{Canetti, Lindell, Ostrovsky, and
  Sahai}{Canetti et~al\mbox{.}}{2002}]%
        {canetti2002universally}
\bibfield{author}{\bibinfo{person}{Ran Canetti}, \bibinfo{person}{Yehuda
  Lindell}, \bibinfo{person}{Rafail Ostrovsky}, {and} \bibinfo{person}{Amit
  Sahai}.} \bibinfo{year}{2002}\natexlab{}.
\newblock \showarticletitle{Universally composable two-party and multi-party
  secure computation}. In \bibinfo{booktitle}{\emph{Symposium on Theory of
  Computing (STOC)}}.
\newblock


\bibitem[\protect\citeauthoryear{Canetti and Rabin}{Canetti and Rabin}{2002}]%
        {UC-joint}
\bibfield{author}{\bibinfo{person}{Ran Canetti} {and} \bibinfo{person}{Tal
  Rabin}.} \bibinfo{year}{2002}\natexlab{}.
\newblock \bibinfo{title}{Universal Composition with Joint State}.
\newblock \bibinfo{howpublished}{Cryptology ePrint Archive, Report 2002/047}.
\newblock
\newblock
\shownote{\url{https://eprint.iacr.org/2002/047}.}


\bibitem[\protect\citeauthoryear{Choi, Dachman-Soled, Kulkarni, and
  Yerukhimovich}{Choi et~al\mbox{.}}{2020}]%
        {diffPrivateSketches}
\bibfield{author}{\bibinfo{person}{Seung~Geol Choi}, \bibinfo{person}{Dana
  Dachman-Soled}, \bibinfo{person}{Mukul Kulkarni}, {and}
  \bibinfo{person}{Arkady Yerukhimovich}.} \bibinfo{year}{2020}\natexlab{}.
\newblock \showarticletitle{Differentially-Private Multi-Party Sketching for
  Large-Scale Statistics}.
\newblock \bibinfo{journal}{\emph{Proceedings on Privacy Enhancing
  Technologies}}  \bibinfo{volume}{3} (\bibinfo{year}{2020}).
\newblock


\bibitem[\protect\citeauthoryear{Damg{\aa}rd}{Damg{\aa}rd}{2010}]%
        {damgaard2002sigma}
\bibfield{author}{\bibinfo{person}{Ivan Damg{\aa}rd}.}
  \bibinfo{year}{2010}\natexlab{}.
\newblock \bibinfo{title}{On $\Sigma$-protocols}.
\newblock
\newblock
\newblock
\shownote{Lecture Notes on Cryptologic Protocol Theory, v.2.}


\bibitem[\protect\citeauthoryear{Damg{\aa}rd, Pastro, Smart, and
  Zakarias}{Damg{\aa}rd et~al\mbox{.}}{2012}]%
        {spdz-crypto2012}
\bibfield{author}{\bibinfo{person}{Ivan Damg{\aa}rd}, \bibinfo{person}{Valerio
  Pastro}, \bibinfo{person}{Nigel Smart}, {and} \bibinfo{person}{Sarah
  Zakarias}.} \bibinfo{year}{2012}\natexlab{}.
\newblock \showarticletitle{Multiparty computation from somewhat homomorphic
  encryption}. In \bibinfo{booktitle}{\emph{Annual International Cryptology
  Conference (Crypto)}}.
\newblock


\bibitem[\protect\citeauthoryear{De~Cristofaro, Gasti, and
  Tsudik}{De~Cristofaro et~al\mbox{.}}{2012}]%
        {setunion-cans2012}
\bibfield{author}{\bibinfo{person}{Emiliano De~Cristofaro},
  \bibinfo{person}{Paolo Gasti}, {and} \bibinfo{person}{Gene Tsudik}.}
  \bibinfo{year}{2012}\natexlab{}.
\newblock \showarticletitle{Fast and private computation of cardinality of set
  intersection and union}. In \bibinfo{booktitle}{\emph{International
  Conference on Cryptology and Network Security (CANS)}}.
\newblock


\bibitem[\protect\citeauthoryear{de~Kok}{de~Kok}{2020}]%
        {paar}
\bibfield{author}{\bibinfo{person}{Daniël de Kok}.}
  \bibinfo{year}{2020}\natexlab{}.
\newblock \bibinfo{title}{Go par package for parallel for-loops}.
\newblock
\newblock
\newblock
\shownote{\url{https://github.com/danieldk/par}.}


\bibitem[\protect\citeauthoryear{Dingledine, Mathewson, and
  Syverson}{Dingledine et~al\mbox{.}}{2004}]%
        {tor}
\bibfield{author}{\bibinfo{person}{Roger Dingledine}, \bibinfo{person}{Nick
  Mathewson}, {and} \bibinfo{person}{Paul Syverson}.}
  \bibinfo{year}{2004}\natexlab{}.
\newblock \showarticletitle{{Tor: The Second-Generation Onion Router}}. In
  \bibinfo{booktitle}{\emph{USENIX Security Symposium (USENIX)}}.
\newblock


\bibitem[\protect\citeauthoryear{Dolev and Strong}{Dolev and Strong}{1983}]%
        {dolev-strong}
\bibfield{author}{\bibinfo{person}{Danny Dolev} {and}
  \bibinfo{person}{H.~Raymond Strong}.} \bibinfo{year}{1983}\natexlab{}.
\newblock \showarticletitle{Authenticated Algorithms for Byzantine Agreement}.
\newblock \bibinfo{journal}{\emph{SIAM J. Comput.}} \bibinfo{volume}{12},
  \bibinfo{number}{4} (\bibinfo{year}{1983}).
\newblock


\bibitem[\protect\citeauthoryear{Durand and Flajolet}{Durand and
  Flajolet}{2003}]%
        {durand2003loglog}
\bibfield{author}{\bibinfo{person}{Marianne Durand} {and}
  \bibinfo{person}{Philippe Flajolet}.} \bibinfo{year}{2003}\natexlab{}.
\newblock \showarticletitle{Loglog counting of large cardinalities}. In
  \bibinfo{booktitle}{\emph{European Symposium on Algorithms (ESA)}}.
\newblock


\bibitem[\protect\citeauthoryear{Dwork, Kenthapadi, McSherry, Mironov, and
  Naor}{Dwork et~al\mbox{.}}{2006a}]%
        {ourdata}
\bibfield{author}{\bibinfo{person}{Cynthia Dwork}, \bibinfo{person}{Krishnaram
  Kenthapadi}, \bibinfo{person}{Frank McSherry}, \bibinfo{person}{Ilya
  Mironov}, {and} \bibinfo{person}{Moni Naor}.}
  \bibinfo{year}{2006}\natexlab{a}.
\newblock \showarticletitle{{Our Data, Ourselves: Privacy via Distributed Noise
  Generation}}. In \bibinfo{booktitle}{\emph{Advances in Cryptology
  (Eurocrypt)}}.
\newblock


\bibitem[\protect\citeauthoryear{Dwork, McSherry, Nissim, and Smith}{Dwork
  et~al\mbox{.}}{2006b}]%
        {calibrating-noise}
\bibfield{author}{\bibinfo{person}{Cynthia Dwork}, \bibinfo{person}{Frank
  McSherry}, \bibinfo{person}{Kobbi Nissim}, {and} \bibinfo{person}{Adam
  Smith}.} \bibinfo{year}{2006}\natexlab{b}.
\newblock \showarticletitle{{Calibrating Noise to Sensitivity in Private Data
  Analysis}}. In \bibinfo{booktitle}{\emph{Theory of Cryptography Conference
  (TCC)}}.
\newblock


\bibitem[\protect\citeauthoryear{Dwork, Roth, et~al\mbox{.}}{Dwork
  et~al\mbox{.}}{2014}]%
        {dwork2014algorithmic}
\bibfield{author}{\bibinfo{person}{Cynthia Dwork}, \bibinfo{person}{Aaron
  Roth}, {et~al\mbox{.}}} \bibinfo{year}{2014}\natexlab{}.
\newblock \showarticletitle{The algorithmic foundations of differential
  privacy.}
\newblock \bibinfo{journal}{\emph{Foundations and Trends in Theoretical
  Computer Science}} \bibinfo{volume}{9}, \bibinfo{number}{3-4}
  (\bibinfo{year}{2014}).
\newblock


\bibitem[\protect\citeauthoryear{Egert, Fischlin, Gens, Jacob, Senker, and
  Tillmanns}{Egert et~al\mbox{.}}{2015}]%
        {setunion-acisp2015}
\bibfield{author}{\bibinfo{person}{Rolf Egert}, \bibinfo{person}{Marc
  Fischlin}, \bibinfo{person}{David Gens}, \bibinfo{person}{Sven Jacob},
  \bibinfo{person}{Matthias Senker}, {and} \bibinfo{person}{J{\"o}rn
  Tillmanns}.} \bibinfo{year}{2015}\natexlab{}.
\newblock \showarticletitle{Privately Computing Set-Union and Set-Intersection
  Cardinality via Bloom Filters}. In \bibinfo{booktitle}{\emph{Australasian
  Conference on Information Security and Privacy}}.
\newblock


\bibitem[\protect\citeauthoryear{Elahi, Danezis, and Goldberg}{Elahi
  et~al\mbox{.}}{2014}]%
        {privex}
\bibfield{author}{\bibinfo{person}{Tariq Elahi}, \bibinfo{person}{George
  Danezis}, {and} \bibinfo{person}{Ian Goldberg}.}
  \bibinfo{year}{2014}\natexlab{}.
\newblock \showarticletitle{{PrivEx: Private Collection of Traffic Statistics
  for Anonymous Communication Networks}}. In \bibinfo{booktitle}{\emph{ACM
  Conference on Computer and Communications Security (CCS)}}.
\newblock


\bibitem[\protect\citeauthoryear{Fenske, Mani, Johnson, and Sherr}{Fenske
  et~al\mbox{.}}{2017}]%
        {psc-ccs2017}
\bibfield{author}{\bibinfo{person}{Ellis Fenske}, \bibinfo{person}{Akshaya
  Mani}, \bibinfo{person}{Aaron Johnson}, {and} \bibinfo{person}{Micah Sherr}.}
  \bibinfo{year}{2017}\natexlab{}.
\newblock \showarticletitle{Distributed Measurement with Private Set-Union
  Cardinality}. In \bibinfo{booktitle}{\emph{ACM Conference on Computer and
  Communications Security (CCS)}}. \bibinfo{publisher}{ACM}.
\newblock


\bibitem[\protect\citeauthoryear{Fiat and Shamir}{Fiat and Shamir}{1987}]%
        {fiat-shamir-crypto1986}
\bibfield{author}{\bibinfo{person}{Amos Fiat} {and} \bibinfo{person}{Adi
  Shamir}.} \bibinfo{year}{1987}\natexlab{}.
\newblock \showarticletitle{How to Prove Yourself: Practical Solutions to
  Identification and Signature Problems}. In \bibinfo{booktitle}{\emph{Advances
  in Cryptology (CRYPTO '86)}}.
\newblock


\bibitem[\protect\citeauthoryear{Freedman, Nissim, and Pinkas}{Freedman
  et~al\mbox{.}}{2004}]%
        {private-set-intersection}
\bibfield{author}{\bibinfo{person}{Michael~J Freedman}, \bibinfo{person}{Kobbi
  Nissim}, {and} \bibinfo{person}{Benny Pinkas}.}
  \bibinfo{year}{2004}\natexlab{}.
\newblock \showarticletitle{{Efficient Private Matching and Set Intersection}}.
  In \bibinfo{booktitle}{\emph{Advances in Cryptology (Eurocrypt)}}.
\newblock


\bibitem[\protect\citeauthoryear{Furukawa, Miyauchi, Mori, Obana, and
  Sako}{Furukawa et~al\mbox{.}}{2003}]%
        {shuffle-furukawa-fc2003}
\bibfield{author}{\bibinfo{person}{Jun Furukawa}, \bibinfo{person}{Hiroshi
  Miyauchi}, \bibinfo{person}{Kengo Mori}, \bibinfo{person}{Satoshi Obana},
  {and} \bibinfo{person}{Kazue Sako}.} \bibinfo{year}{2003}\natexlab{}.
\newblock \showarticletitle{An Implementation of a Universally Verifiable
  Electronic Voting Scheme Based on Shuffling}. In
  \bibinfo{booktitle}{\emph{Financial Cryptography}}
  \emph{(\bibinfo{series}{FC'02})}.
\newblock


\bibitem[\protect\citeauthoryear{Goldreich}{Goldreich}{2001}]%
        {goldreich-foc1-2004}
\bibfield{author}{\bibinfo{person}{O. Goldreich}.}
  \bibinfo{year}{2001}\natexlab{}.
\newblock \bibinfo{booktitle}{\emph{Foundations of Cryptography: Volume 2,
  Basic Applications}}.
\newblock \bibinfo{publisher}{Cambridge University Press}.
\newblock


\bibitem[\protect\citeauthoryear{Goldreich, Micali, and Wigderson}{Goldreich
  et~al\mbox{.}}{1987}]%
        {gmw-mpc}
\bibfield{author}{\bibinfo{person}{Oded Goldreich}, \bibinfo{person}{Silvio
  Micali}, {and} \bibinfo{person}{Avi Wigderson}.}
  \bibinfo{year}{1987}\natexlab{}.
\newblock \showarticletitle{{How to play ANY mental game}}. In
  \bibinfo{booktitle}{\emph{ACM Symposium on Theory of Computing (STOC)}}.
\newblock


\bibitem[\protect\citeauthoryear{Goldreich and Oren}{Goldreich and
  Oren}{1994}]%
        {seqcom}
\bibfield{author}{\bibinfo{person}{Oded Goldreich} {and} \bibinfo{person}{Yair
  Oren}.} \bibinfo{year}{1994}\natexlab{}.
\newblock \showarticletitle{Definitions and properties of zero-knowledge proof
  systems}.
\newblock \bibinfo{journal}{\emph{Journal of Cryptology}} \bibinfo{volume}{7},
  \bibinfo{number}{1} (\bibinfo{year}{1994}).
\newblock


\bibitem[\protect\citeauthoryear{Groth}{Groth}{2003}]%
        {shuffle-groth-pkc2003}
\bibfield{author}{\bibinfo{person}{Jens Groth}.}
  \bibinfo{year}{2003}\natexlab{}.
\newblock \showarticletitle{A Verifiable Secret Shuffle of Homomorphic
  Encryptions}. In \bibinfo{booktitle}{\emph{Theory and Practice in Public Key
  Cryptography (PKC)}}.
\newblock


\bibitem[\protect\citeauthoryear{Groth, Ostrovsky, and Sahai}{Groth
  et~al\mbox{.}}{2006}]%
        {UC-NIZKP}
\bibfield{author}{\bibinfo{person}{Jens Groth}, \bibinfo{person}{Rafail
  Ostrovsky}, {and} \bibinfo{person}{Amit Sahai}.}
  \bibinfo{year}{2006}\natexlab{}.
\newblock \showarticletitle{Perfect non-interactive zero knowledge for NP}. In
  \bibinfo{booktitle}{\emph{Annual International Conference on the Theory and
  Applications of Cryptographic Techniques (EUROCRYPT)}}.
\newblock


\bibitem[\protect\citeauthoryear{Hazay, Mikkelsen, Rabin, Toft, and
  Nicolosi}{Hazay et~al\mbox{.}}{2012}]%
        {efficient-rsa-keygen}
\bibfield{author}{\bibinfo{person}{Carmit Hazay},
  \bibinfo{person}{Gert~L{\ae}ss{\o}e Mikkelsen}, \bibinfo{person}{Tal Rabin},
  \bibinfo{person}{Tomas Toft}, {and} \bibinfo{person}{Angelo~Agatino
  Nicolosi}.} \bibinfo{year}{2012}\natexlab{}.
\newblock \showarticletitle{Efficient RSA Key Generation and Threshold Paillier
  in the Two-Party Setting}. In \bibinfo{booktitle}{\emph{Topics in Cryptology
  -- CT-RSA}}.
\newblock


\bibitem[\protect\citeauthoryear{Hazay and Nissim}{Hazay and Nissim}{2012}]%
        {setops-joc2012}
\bibfield{author}{\bibinfo{person}{Carmit Hazay} {and} \bibinfo{person}{Kobbi
  Nissim}.} \bibinfo{year}{2012}\natexlab{}.
\newblock \showarticletitle{Efficient Set Operations in the Presence of
  Malicious Adversaries}.
\newblock \bibinfo{journal}{\emph{Journal of Cryptology}} \bibinfo{volume}{25},
  \bibinfo{number}{3} (\bibinfo{year}{2012}).
\newblock


\bibitem[\protect\citeauthoryear{Inan, Kantarcioglu, Ghinita, and Bertino}{Inan
  et~al\mbox{.}}{2010}]%
        {privaterecordInan}
\bibfield{author}{\bibinfo{person}{Ali Inan}, \bibinfo{person}{Murat
  Kantarcioglu}, \bibinfo{person}{Gabriel Ghinita}, {and}
  \bibinfo{person}{Elisa Bertino}.} \bibinfo{year}{2010}\natexlab{}.
\newblock \showarticletitle{Private record matching using differential
  privacy}. In \bibinfo{booktitle}{\emph{International Conference on Extending
  Database Technology}}.
\newblock


\bibitem[\protect\citeauthoryear{Ishai, Ostrovsky, and Zikas}{Ishai
  et~al\mbox{.}}{2014}]%
        {ishai2014secure}
\bibfield{author}{\bibinfo{person}{Yuval Ishai}, \bibinfo{person}{Rafail
  Ostrovsky}, {and} \bibinfo{person}{Vassilis Zikas}.}
  \bibinfo{year}{2014}\natexlab{}.
\newblock \showarticletitle{Secure multi-party computation with identifiable
  abort}. In \bibinfo{booktitle}{\emph{Annual Cryptology Conference (CRYPTO)}}.
\newblock


\bibitem[\protect\citeauthoryear{Jansen and Johnson}{Jansen and
  Johnson}{2016}]%
        {privcount}
\bibfield{author}{\bibinfo{person}{Rob Jansen} {and} \bibinfo{person}{Aaron
  Johnson}.} \bibinfo{year}{2016}\natexlab{}.
\newblock \showarticletitle{{Safely Measuring Tor}}. In
  \bibinfo{booktitle}{\emph{ACM Conference on Computer and Communications
  Security (CCS)}}.
\newblock


\bibitem[\protect\citeauthoryear{Kasiviswanathan and Smith}{Kasiviswanathan and
  Smith}{2014}]%
        {dp-semantics-jp08}
\bibfield{author}{\bibinfo{person}{Shiva~P Kasiviswanathan} {and}
  \bibinfo{person}{Adam Smith}.} \bibinfo{year}{2014}\natexlab{}.
\newblock \showarticletitle{On the `Semantics' of Differential Privacy: A
  Bayesian Formulation}.
\newblock \bibinfo{journal}{\emph{Journal of Privacy and Confidentiality}}
  \bibinfo{volume}{6}, \bibinfo{number}{1} (\bibinfo{year}{2014}).
\newblock


\bibitem[\protect\citeauthoryear{Katz, Maurer, Tackmann, and Zikas}{Katz
  et~al\mbox{.}}{2013}]%
        {uc-sync}
\bibfield{author}{\bibinfo{person}{Jonathan Katz}, \bibinfo{person}{Ueli
  Maurer}, \bibinfo{person}{Bj{\"o}rn Tackmann}, {and}
  \bibinfo{person}{Vassilis Zikas}.} \bibinfo{year}{2013}\natexlab{}.
\newblock \showarticletitle{Universally composable synchronous computation}. In
  \bibinfo{booktitle}{\emph{Theory of Cryptography Conference (TCC)}}.
\newblock


\bibitem[\protect\citeauthoryear{Kiayias, Zhou, and Zikas}{Kiayias
  et~al\mbox{.}}{2016}]%
        {fair-mpc-global-ledger}
\bibfield{author}{\bibinfo{person}{Aggelos Kiayias},
  \bibinfo{person}{Hong-Sheng Zhou}, {and} \bibinfo{person}{Vassilis Zikas}.}
  \bibinfo{year}{2016}\natexlab{}.
\newblock \showarticletitle{Fair and Robust Multi-party Computation Using a
  Global Transaction Ledger}. In \bibinfo{booktitle}{\emph{Advances in
  Cryptology (EUROCRYPT)}}.
\newblock


\bibitem[\protect\citeauthoryear{Kissner and Song}{Kissner and Song}{2005}]%
        {setops-crypto2005}
\bibfield{author}{\bibinfo{person}{Lea Kissner} {and} \bibinfo{person}{Dawn
  Song}.} \bibinfo{year}{2005}\natexlab{}.
\newblock \showarticletitle{Privacy-Preserving Set Operations}. In
  \bibinfo{booktitle}{\emph{Annual International Cryptology Conference
  (Crypto)}}.
\newblock


\bibitem[\protect\citeauthoryear{kyber.}{kyber.}{2020}]%
        {kyber}
kyber. \bibinfo{year}{2020}\natexlab{}.
\newblock \bibinfo{title}{kyber: DEDIS Advanced Crypto Library for Go}.
\newblock
\newblock
\newblock
\shownote{\url{https://godoc.org/go.dedis.ch/kyber}.}


\bibitem[\protect\citeauthoryear{Larraia, Orsini, and Smart}{Larraia
  et~al\mbox{.}}{2014}]%
        {mpc-bin-crypto2014}
\bibfield{author}{\bibinfo{person}{Enrique Larraia}, \bibinfo{person}{Emmanuela
  Orsini}, {and} \bibinfo{person}{Nigel~P Smart}.}
  \bibinfo{year}{2014}\natexlab{}.
\newblock \showarticletitle{Dishonest Majority Multi-Party Computation for
  Binary Circuits}. In \bibinfo{booktitle}{\emph{Annual International
  Cryptology Conference (Crypto)}}.
\newblock


\bibitem[\protect\citeauthoryear{Lindell}{Lindell}{2005}]%
        {lindell2005secure}
\bibfield{author}{\bibinfo{person}{Yehida Lindell}.}
  \bibinfo{year}{2005}\natexlab{}.
\newblock \showarticletitle{Secure multiparty computation for privacy
  preserving data mining}.
\newblock In \bibinfo{booktitle}{\emph{Encyclopedia of Data Warehousing and
  Mining}}. \bibinfo{pages}{1005--1009}.
\newblock


\bibitem[\protect\citeauthoryear{Lindell}{Lindell}{2015}]%
        {lindellsigmaprotocols}
\bibfield{author}{\bibinfo{person}{Yehuda Lindell}.}
  \bibinfo{year}{2015}\natexlab{}.
\newblock \showarticletitle{An Efficient Transform from Sigma Protocols to NIZK
  with a CRS and Non-programmable Random Oracle}. In
  \bibinfo{booktitle}{\emph{Theory of Cryptography (TCC)}}.
\newblock


\bibitem[\protect\citeauthoryear{Lindell, Pinkas, Smart, and Yanai}{Lindell
  et~al\mbox{.}}{2015}]%
        {lindell2015efficient}
\bibfield{author}{\bibinfo{person}{Yehuda Lindell}, \bibinfo{person}{Benny
  Pinkas}, \bibinfo{person}{Nigel~P Smart}, {and} \bibinfo{person}{Avishay
  Yanai}.} \bibinfo{year}{2015}\natexlab{}.
\newblock \showarticletitle{Efficient constant round multi-party computation
  combining BMR and SPDZ}. In \bibinfo{booktitle}{\emph{Annual Cryptology
  Conference (Crypto)}}.
\newblock


\bibitem[\protect\citeauthoryear{Mani and Sherr}{Mani and Sherr}{2017}]%
        {histore}
\bibfield{author}{\bibinfo{person}{Akshaya Mani} {and} \bibinfo{person}{M.
  Sherr}.} \bibinfo{year}{2017}\natexlab{}.
\newblock \showarticletitle{{Histor$\epsilon$: Differentially Private and
  Robust Statistics Collection for Tor}}. In \bibinfo{booktitle}{\emph{Network
  and Distributed System Security Symposium (NDSS)}}.
\newblock


\bibitem[\protect\citeauthoryear{McCoy, Bauer, Grunwald, Kohno, and
  Sicker}{McCoy et~al\mbox{.}}{2008}]%
        {tor-usage}
\bibfield{author}{\bibinfo{person}{Damon McCoy}, \bibinfo{person}{Kevin Bauer},
  \bibinfo{person}{Dirk Grunwald}, \bibinfo{person}{Tadayoshi Kohno}, {and}
  \bibinfo{person}{Douglas Sicker}.} \bibinfo{year}{2008}\natexlab{}.
\newblock \showarticletitle{{Shining Light in Dark Places: Understanding the
  Tor Network}}. In \bibinfo{booktitle}{\emph{Privacy Enhancing Technologies
  Symposium (PETS)}}.
\newblock


\bibitem[\protect\citeauthoryear{McSherry and Talwar}{McSherry and
  Talwar}{2007}]%
        {exponential-mech-focs07}
\bibfield{author}{\bibinfo{person}{Frank McSherry} {and} \bibinfo{person}{Kunal
  Talwar}.} \bibinfo{year}{2007}\natexlab{}.
\newblock \showarticletitle{Mechanism design via differential privacy}. In
  \bibinfo{booktitle}{\emph{Foundations of Computer Science (FOCS)}}.
\newblock


\bibitem[\protect\citeauthoryear{Melis, Danezis, and Cristofaro}{Melis
  et~al\mbox{.}}{2016}]%
        {succinct-sketches}
\bibfield{author}{\bibinfo{person}{Luca Melis}, \bibinfo{person}{George
  Danezis}, {and} \bibinfo{person}{Emiliano~De Cristofaro}.}
  \bibinfo{year}{2016}\natexlab{}.
\newblock \showarticletitle{{Efficient Private Statistics with Succinct
  Sketches}}. In \bibinfo{booktitle}{\emph{Network and Distributed System
  Security Symposium (NDSS)}}.
\newblock


\bibitem[\protect\citeauthoryear{Neff}{Neff}{2001}]%
        {neff-shuffle}
\bibfield{author}{\bibinfo{person}{C.~Andrew Neff}.}
  \bibinfo{year}{2001}\natexlab{}.
\newblock \showarticletitle{{A Verifiable Secret Shuffle and its Application to
  E-voting}}. In \bibinfo{booktitle}{\emph{ACM Conference on Computer and
  Communications Security (CCS)}}.
\newblock


\bibitem[\protect\citeauthoryear{Nguyen, Safavi-Naini, and Kurosawa}{Nguyen
  et~al\mbox{.}}{2004}]%
        {nguyen2004verifiable}
\bibfield{author}{\bibinfo{person}{Lan Nguyen}, \bibinfo{person}{Rei
  Safavi-Naini}, {and} \bibinfo{person}{Kaoru Kurosawa}.}
  \bibinfo{year}{2004}\natexlab{}.
\newblock \showarticletitle{{Verifiable Shuffles: A Formal Model and a
  Paillier-based Efficient Construction with Provable Security}}. In
  \bibinfo{booktitle}{\emph{Applied Cryptography and Network Security (ACNS)}}.
\newblock


\bibitem[\protect\citeauthoryear{Nissim, Raskhodnikova, and Smith}{Nissim
  et~al\mbox{.}}{2007}]%
        {smooth-sensitivity-stoc07}
\bibfield{author}{\bibinfo{person}{Kobbi Nissim}, \bibinfo{person}{Sofya
  Raskhodnikova}, {and} \bibinfo{person}{Adam Smith}.}
  \bibinfo{year}{2007}\natexlab{}.
\newblock \showarticletitle{Smooth Sensitivity and Sampling in Private Data
  Analysis}. In \bibinfo{booktitle}{\emph{Symposium on Theory of Computing
  (STOC)}}.
\newblock


\bibitem[\protect\citeauthoryear{Partridge and Allman}{Partridge and
  Allman}{2016}]%
        {partridge2016ethical}
\bibfield{author}{\bibinfo{person}{Craig Partridge} {and} \bibinfo{person}{Mark
  Allman}.} \bibinfo{year}{2016}\natexlab{}.
\newblock \showarticletitle{Ethical considerations in network measurement
  papers}.
\newblock \bibinfo{journal}{\emph{Commun. ACM}} \bibinfo{volume}{59},
  \bibinfo{number}{10} (\bibinfo{year}{2016}).
\newblock


\bibitem[\protect\citeauthoryear{Pettai and Laud}{Pettai and Laud}{2015}]%
        {pettai2015combining}
\bibfield{author}{\bibinfo{person}{Martin Pettai} {and} \bibinfo{person}{Peeter
  Laud}.} \bibinfo{year}{2015}\natexlab{}.
\newblock \showarticletitle{Combining differential privacy and secure
  multiparty computation}. In \bibinfo{booktitle}{\emph{Annual Computer
  Security Applications Conference (ACSAC)}}.
\newblock


\bibitem[\protect\citeauthoryear{Ran~Canetti}{Ran~Canetti}{2004}]%
        {rom-imposs}
\bibfield{author}{\bibinfo{person}{Shai~Halevi Ran~Canetti, Oded~Goldreich}.}
  \bibinfo{year}{2004}\natexlab{}.
\newblock \showarticletitle{The Random Oracle Methodology, Revisited}.
\newblock \bibinfo{journal}{\emph{Journal of the ACM (JACM)}}
  \bibinfo{volume}{51}, \bibinfo{number}{4} (\bibinfo{year}{2004}).
\newblock


\bibitem[\protect\citeauthoryear{Schnorr}{Schnorr}{1991}]%
        {schnorr-joc1991}
\bibfield{author}{\bibinfo{person}{Claus-Peter Schnorr}.}
  \bibinfo{year}{1991}\natexlab{}.
\newblock \showarticletitle{Efficient Signature Generation by Smart Cards}.
\newblock \bibinfo{journal}{\emph{Journal of Cryptology}} \bibinfo{volume}{4},
  \bibinfo{number}{3} (\bibinfo{year}{1991}).
\newblock


\bibitem[\protect\citeauthoryear{Soghoian}{Soghoian}{2011}]%
        {tor-research}
\bibfield{author}{\bibinfo{person}{Christopher Soghoian}.}
  \bibinfo{year}{2011}\natexlab{}.
\newblock \showarticletitle{{Enforced Community Standards For Research on Users
  of the Tor Anonymity Network}}. In \bibinfo{booktitle}{\emph{Workshop on
  Ethics in Computer Security Research (WECSR)}}.
\newblock


\bibitem[\protect\citeauthoryear{Stanojevic, Nabeel, and Yu}{Stanojevic
  et~al\mbox{.}}{2017}]%
        {distributed-cardinality-with-dp}
\bibfield{author}{\bibinfo{person}{Rade Stanojevic}, \bibinfo{person}{Mohamed
  Nabeel}, {and} \bibinfo{person}{Ting Yu}.} \bibinfo{year}{2017}\natexlab{}.
\newblock \showarticletitle{Distributed Cardinality Estimation of Set
  Operations with Differential Privacy}. In \bibinfo{booktitle}{\emph{IEEE
  Symposium on Privacy-Aware Computing (PAC)}}.
\newblock


\bibitem[\protect\citeauthoryear{Stoica, Morris, Karger, Kaashoek, and
  Balakrishnan}{Stoica et~al\mbox{.}}{2001}]%
        {chord}
\bibfield{author}{\bibinfo{person}{Ion Stoica}, \bibinfo{person}{Robert
  Morris}, \bibinfo{person}{David Karger}, \bibinfo{person}{M~Frans Kaashoek},
  {and} \bibinfo{person}{Hari Balakrishnan}.} \bibinfo{year}{2001}\natexlab{}.
\newblock \showarticletitle{{Chord: A Scalable Peer-to-Peer Lookup Service for
  Internet Applications}}. In \bibinfo{booktitle}{\emph{Conference on
  Applications, Technologies, Architectures, and Protocols for Computer
  Communications (SIGCOMM)}}.
\newblock


\bibitem[\protect\citeauthoryear{Struik}{Struik}{2021}]%
        {edwards}
\bibfield{author}{\bibinfo{person}{Rene Struik}.}
  \bibinfo{year}{2021}\natexlab{}.
\newblock \bibinfo{booktitle}{\emph{{Alternative Elliptic Curve
  Representations}}}.
\newblock \bibinfo{type}{Internet-Draft}
  draft-ietf-lwig-curve-representations-20. \bibinfo{institution}{Internet
  Engineering Task Force}.
\newblock
\newblock
\shownote{Work in Progress.}


\bibitem[\protect\citeauthoryear{{Tor Metrics.}}{{Tor Metrics.}}{2020}]%
        {tormetrics}
{Tor Metrics.} \bibinfo{year}{2020}\natexlab{}.
\newblock \bibinfo{howpublished}{\url{https://metrics.torproject.org/}}.
\newblock


\bibitem[\protect\citeauthoryear{Vaidya and Clifton}{Vaidya and
  Clifton}{2005}]%
        {setint-jcs2005}
\bibfield{author}{\bibinfo{person}{Jaideep Vaidya} {and} \bibinfo{person}{Chris
  Clifton}.} \bibinfo{year}{2005}\natexlab{}.
\newblock \showarticletitle{Secure Set Intersection Cardinality with
  Application to Association Rule Mining}.
\newblock \bibinfo{journal}{\emph{Journal of Computer Security}}
  \bibinfo{volume}{13}, \bibinfo{number}{4} (\bibinfo{year}{2005}).
\newblock


\bibitem[\protect\citeauthoryear{Varda}{Varda}{2008}]%
        {varda2008protocol}
\bibfield{author}{\bibinfo{person}{Kenton Varda}.}
  \bibinfo{year}{2008}\natexlab{}.
\newblock \bibinfo{title}{Protocol Buffers: Google's Data Interchange Format}.
\newblock
\newblock
\newblock
\shownote{\url{https://opensource.googleblog.com/2008/07/protocol-buffers-googles-data.html}.}


\bibitem[\protect\citeauthoryear{Wails, Johnson, Starin, Yerukhimovich, and
  Gordon}{Wails et~al\mbox{.}}{2019}]%
        {stormy-ccs2019}
\bibfield{author}{\bibinfo{person}{Ryan Wails}, \bibinfo{person}{Aaron
  Johnson}, \bibinfo{person}{Daniel Starin}, \bibinfo{person}{Arkady
  Yerukhimovich}, {and} \bibinfo{person}{S.~Dov Gordon}.}
  \bibinfo{year}{2019}\natexlab{}.
\newblock \showarticletitle{Stormy: Statistics in Tor by Measuring Securely}.
\newblock \bibinfo{journal}{\emph{ACM Conference on Computer and Communications
  Security (CCS)}}.
\newblock


\bibitem[\protect\citeauthoryear{Wang, Ranellucci, and Katz}{Wang
  et~al\mbox{.}}{2017}]%
        {wang2017global}
\bibfield{author}{\bibinfo{person}{Xiao Wang}, \bibinfo{person}{Samuel
  Ranellucci}, {and} \bibinfo{person}{Jonathan Katz}.}
  \bibinfo{year}{2017}\natexlab{}.
\newblock \showarticletitle{Global-scale secure multiparty computation}. In
  \bibinfo{booktitle}{\emph{ACM Conference on Computer and Communications
  Security (CCS)}}.
\newblock


\bibitem[\protect\citeauthoryear{Wikstr{\"o}m}{Wikstr{\"o}m}{2005}]%
        {wikstrom-shuffle}
\bibfield{author}{\bibinfo{person}{Douglas Wikstr{\"o}m}.}
  \bibinfo{year}{2005}\natexlab{}.
\newblock \showarticletitle{A Sender Verifiable Mix-Net and a New Proof of a
  Shuffle}. In \bibinfo{booktitle}{\emph{International Conference on the Theory
  and Application of Cryptology and Information Security (ASIACRYPT)}}.
\newblock


\end{thebibliography}
